\documentclass[a4paper,UKenglish,cleveref, autoref]{oasics-v2021}

\title{The complexity of high-dimensional cuts} 

\titlerunning{The complexity of high-dimensional cuts}

\author{Ulrich Bauer}{Department of Mathematics, Technical University of Munich (TUM)\\{Boltzmannstr. 3, 85748 Garching b. M\"unchen, Germany}}{ulrich.bauer@tum.de}{https://orcid.org/0000-0002-9683-0724}{}

\author{Abhishek Rathod}{Department of Mathematics, Technical University of Munich (TUM)\\{Boltzmannstr. 3, 85748 Garching b. M\"unchen, Germany}}{abhishek.rathod@tum.de}{https://orcid.org/0000-0003-2533-3699}{}

\author{Meirav Zehavi}{Ben-Gurion University of the Negev, Beer-Sheva, Israel}{meiravze@bgu.ac.il}{https://orcid.org/0000-0002-3636-5322}{}

\authorrunning{Ulrich Bauer, Abhishek Rathod, Meirav Zehavi}

\Copyright{Ulrich Bauer, Abhishek Rathod, Meirav Zehavi}



\nolinenumbers 

\hideOASIcs  

\EventEditors{John Q. Open and Joan R. Access}
\EventNoEds{2}
\EventLongTitle{42nd Conference on Very Important Topics (CVIT 2016)}
\EventShortTitle{CVIT 2016}
\EventAcronym{CVIT}
\EventYear{2016}
\EventDate{December 24--27, 2016}
\EventLocation{Little Whinging, United Kingdom}
\EventLogo{}
\SeriesVolume{42}
\ArticleNo{23}

\usepackage[svgnames]{xcolor}

\usepackage[all]{xy}
\usepackage{todonotes}
\usepackage{cleveref}
\usepackage{csquotes}
\usepackage{nicefrac}
\usepackage{algorithm,algpseudocode}
\algnewcommand{\LineComment}[1]{\State \(\triangleright\) #1}
\usepackage[nocompress]{cite}
\usepackage{mathtools,amsmath,amssymb}

\usepackage{paralist}

\usepackage{tikz-cd}
\usepackage{tikz}
\usetikzlibrary{arrows,decorations.pathmorphing,backgrounds,positioning,fit,matrix}
\usepackage{tkz-graph}
\usetikzlibrary{calc}

\theoremstyle{remark}
\newtheorem{notation}{Notation}

\newcommand{\hone}{\mathsf{H}_{1}(\complex)}
\newcommand{\honer}{\mathsf{H}^{1}(\complex)}

\newcommand{\inclusion}{\xhookrightarrow{}}

\newcommand{\ACC}{\mathcal{A}}
\newcommand{\BCC}{\mathcal{B}}
\newcommand{\LCC}{\mathcal{L}}
\newcommand{\UCC}{\mathcal{U}}
\newcommand{\SCC}{\mathcal{S}}
\newcommand{\PCC}{\mathcal{P}}
\newcommand{\RCC}{\mathcal{R}}
\newcommand{\OCC}{\mathcal{O}}

\newcommand{\CCC}{\mathcal{C}}
\newcommand{\XCC}{\mathcal{X}}
\newcommand{\HCC}{\mathcal{H}}
\newcommand{\ECC}{\mathcal{E}}

\newcommand{\VCC}{U}
\newcommand{\YCC}{\mathcal{Y}}
\newcommand{\ZCC}{\mathcal{Z}}

\newcommand{\TCC}{\mathcal{T}}

\usepackage{pgfplots}
\pgfplotsset{compat=1.15}
\usepackage{mathrsfs}
\usetikzlibrary{arrows}
\definecolor{wrwrwr}{rgb}{0.3803921568627451,0.3803921568627451,0.3803921568627451}
\definecolor{wqwqwq}{rgb}{0.3764705882352941,0.3764705882352941,0.3764705882352941}
\definecolor{ffwwzz}{rgb}{1,0.4,0.6}
\definecolor{yqqqyq}{rgb}{0.5019607843137255,0,0.5019607843137255}
\definecolor{qqzzqq}{rgb}{0,0.6,0}
\definecolor{ffttww}{rgb}{1,0.2,0.4}
\definecolor{rvwvcq}{rgb}{0.08235294117647059,0.396078431372549,0.7529411764705882}
\definecolor{wwccqq}{rgb}{0.4,0.8,0}

\newcommand\NoIndent[1]{%
  \par\vbox{\parbox[t]{\linewidth}{#1}}%
}

\newcommand{\fpt}{{\bf FPT}\xspace}
\newcommand{\WWW}{{\bf W}\xspace}

\newcommand{\Wone}{{\bf W{[1]}}\xspace}
\newcommand{\XP}{{\bf XP}\xspace}

\newcommand{\NP}{{\bf NP}\xspace}
\newcommand{\PPP}{{\bf P}\xspace}

\newcommand{\typeone}{\textsf{type-1}\xspace}
\newcommand{\typetwo}{\textsf{type-2}\xspace}

\newcommand*{\medcap}{\mathbin{\scalebox{1.5}{\ensuremath{\cap}}}}

\newcommand{\complex}{\mathsf{K}}
\newcommand{\surface}{\mathsf{K}}

\newcommand{\delcomplex}{\complex_{\CCC}}
\newcommand{\delcomplextwo}{\complex_{\SCC}}

\newcommand{\gcc}{{\complex}}

\newcommand{\hasse}{H_\gcc}
\newcommand{\dualgraph}{D_\complex}

\newcommand{\OO}[1]{O{#1}}

\newcommand{\homr}{\mathsf{H}_r}

\newcommand{\cycr}{\mathsf{Z}_r}

\newcommand{\boundr}{\mathsf{B}_r}

\DeclareMathOperator*{\nbd}{nbd}
\DeclareMathOperator*{\im}{im}

\DeclareMathOperator*{\sol}{SOL}

\algnewcommand{\IfThenElse}[3]{
  \State \algorithmicif\ #1\ \algorithmicthen\ #2\ \algorithmicelse\ #3}
  
  \algnewcommand{\IfThen}[2]{
  \State \algorithmicif\ #1\ \algorithmicthen\ #2\ }

\algnewcommand{\LineFor}[2]{
    \State\algorithmicfor\  {#1}\ \algorithmicdo \ {#2}\ }

\newcounter{problemcounter}
\newenvironment{problem}[4][\unskip]
{
	\medskip\noindent {\textbf{Problem
	\arabic{problemcounter}} (\textsc{#1}).\nopagebreak

	{\begin{tabular}{p{2.3cm}p{10.2cm}}	
		\textsc{Instance}: & #2 \\
		\textsc{Parameter}: & #3 \\
		\textsc{Question}: & #4 \\
	\end{tabular}}
	\stepcounter{problemcounter}}
	\medskip
}

\definecolor{ccqqqq}{rgb}{0.8,0,0}
\definecolor{ttqqtt}{rgb}{0.2,0,0.2}
\definecolor{qqwuqq}{rgb}{0,0.39215686274509803,0}
\definecolor{wwqqcc}{rgb}{0.4,0,0.8}
\definecolor{ududff}{rgb}{0.30196078431372547,0.30196078431372547,1}

\definecolor{ffttww}{rgb}{1,0.2,0.4}

\definecolor{ffqqqq}{rgb}{1,0,0}
\definecolor{yqqqyq}{rgb}{0.5019607843137255,0,0.5019607843137255}
\definecolor{rvwvcq}{rgb}{0.08235294117647059,0.396078431372549,0.7529411764705882}

\newcommand{\boldb}{\textbf{b}}
\newcommand{\bolda}{\textbf{a}}
\newcommand{\boldc}{\textbf{c}}
\newcommand{\boldx}{\textbf{x}}
\newcommand{\boldy}{\textbf{y}}
\newcommand{\boldz}{\textbf{z}}

\newcommand{\boldA}{\textbf{A}}
\newcommand{\boldM}{\textbf{M}}
\newcommand{\boldB}{\textbf{B}}
\newcommand{\boldY}{\textbf{Y}}

\newcommand{\sets}{\mathcal{R}}
\newcommand{\elements}{\mathcal{C}}

\newcommand{\altcomplex}{\mathsf{L}}
\newcommand{\altaltcomplex}{\mathsf{C}}
\newcommand{\face}{\eta}
\newcommand{\smallface}{\omega}

\newcommand{\hitcycles}{\textsc{Topological Hitting Set}\xspace}
\newcommand{\killcycles}{\textsc{Global Topological Hitting Set}\xspace}

\newcommand{\setcover}{\textsc{Set Cover}\xspace}

\newcommand{\kmulticolorclique}{$k$-\textsc{Multicolored Clique}\xspace}

\newcommand{\createcycle}{\textsc{Boundary Nontrivialization}\xspace}
\newcommand{\createcycleg}{\textsc{Global Boundary Nontrivialization}\xspace}

\newcommand{\str}{\operatorname{star}}

\usepackage{xspace}

\usepackage{mathtools}

\begin{document}

\maketitle

\begin{abstract}
Cut problems form one of the most fundamental classes of problems
in algorithmic graph theory. For instance,
the minimum cut, the minimum $s$-$t$ cut, the minimum multiway cut,
and the minimum $k$-way cut are some of the commonly encountered cut problems. Many of these problems have been extensively
studied over several decades. 
 In this paper, we initiate the algorithmic study of some cut problems in high dimensions. 

The first problem we study, namely, \hitcycles (THS),
is defined as follows: Given a nontrivial $r$-cycle $\zeta$ in a
simplicial complex $\complex$, find a set $\SCC$ of $r$-dimensional simplices
of minimum cardinality so that $\SCC$ meets every cycle homologous to
$\zeta$. Our main result is that this problem admits a polynomial
time solution on triangulations of closed surfaces. Interestingly,
the optimal solution is given in terms of the cocycles of the surface.
For general complexes, we show that THS is \Wone-hard with respect
to the solution size $k$. On the positive side,
we show that THS admits an \fpt algorithm with respect to $k+d$, where
$d$ is the maximum degree of the Hasse graph of the complex $\complex$.

We also define a problem called \createcycle (BNT):
Given a bounding $r$-cycle $\zeta$ in a simplicial complex $\complex$,
find a set $\SCC$ of $(r+1)$-dimensional simplices of minimum cardinality
so that the removal of $\SCC$ from $\complex$ makes $\zeta$ non-bounding.
We show that BNT is  \Wone-hard with respect to the solution size as the parameter, 
and has an $O(\log n)$-approximation \fpt algorithm for $(r+1)$-dimensional complexes with
the $(r+1)$-th Betti number $\beta_{r+1}$ as the parameter.
Finally, we  provide randomized (approximation)  \fpt algorithms for the global variants of THS and BNT.

\end{abstract}

\section{Introduction} \label{sec:intro}

A \emph{graph cut} is a partition of the vertices of a graph into two disjoint subsets. 
The set of edges that have one vertex lying in each of the two subsets determines a so-called \emph{cut-set}. 
Typically, the objective function to optimize involves the size of the cut-set. 
Graph cuts have a ubiquitous presence in theoretical computer science. 
Cuts are also related to the spectra of the adjacency matrix of the graph leading to a beautiful mathematical theory~\cite{chung}.
Cuts have also found many real-world applications  in clustering, shape matching,  image segmentation and smoothing, and energy minimization problems in computer vision.

Cut problems are related to flow problems in graphs due to the duality between cuts and flows. In fact, the max-flow min-cut theorem  tells us that the
maximum value of flow between a vertex $s$ and and vertex $t$ equals the value of the minimum cut that separates $s$ and $t$. 
\Cref{fig:stcut} shows an example of an  $s$-$t$ cut on an undirected graph.

\begin{figure}

\centering

\begin{tikzpicture}
[scale=.45,auto=left,every node/.style={circle,fill=red!20, scale=1.2, minimum size=1.1cm}]
\tikzset{vertex/.style = {shape=circle,draw}}
\tikzset{edge/.style = {->,> = latex'}}

  \node[fill,draw] (s) at (0,5){s};
  \node[fill,draw] (v1) at (5,9){a};
  \node[fill,draw] (v4) at (13,9){b};
  \node[fill,draw] (v2) at (5,1){c};
  \node[fill,draw] (v3) at (6,5){d};
  \node[fill,draw] (v5) at (13,1){e};
  \node[fill,draw] (v6) at (16,5){f};
  \node[fill,draw] (t) at (20,5){t};

\tikzset{EdgeStyle/.style={-}}
\Edge(s)(v1);
\Edge(s)(v2);
\Edge(s)(v3);
\Edge(v2)(v4);
\Edge(v2)(v5);
\Edge(v1)(v4);
\Edge(v3)(v4);
\Edge(v6)(t);
\Edge(v5)(t);
\Edge(v4)(t);
\Edge(v5)(v6);
\Edge(v3)(v6);
\Edge(v4)(v6);

\Edge(v3)(v5);
\Edge(v6)(v1);

\draw[dashed] 
  ([yshift=60pt]$ (s)!0.5!(v1) $ ) -- 
  ([yshift=-60pt]$ (s)!0.5!(v2) $ );

\end{tikzpicture}

\caption{The dashed vertical line shows a minimum $s$-$t$ cut  in the graph.}
\label{fig:stcut}
\end{figure}
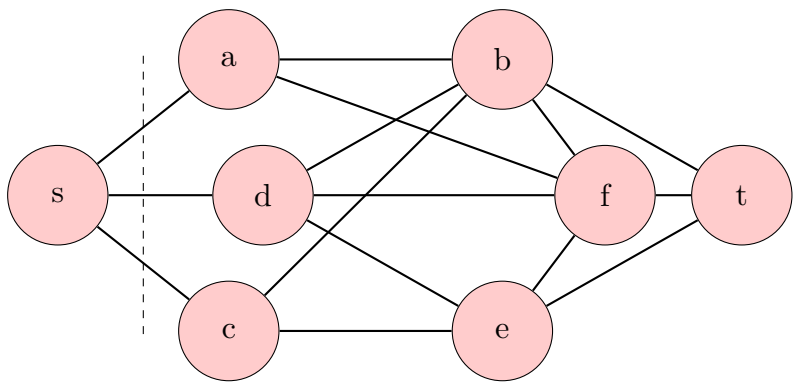

\begin{figure}
\includegraphics[clip,scale=0.18]{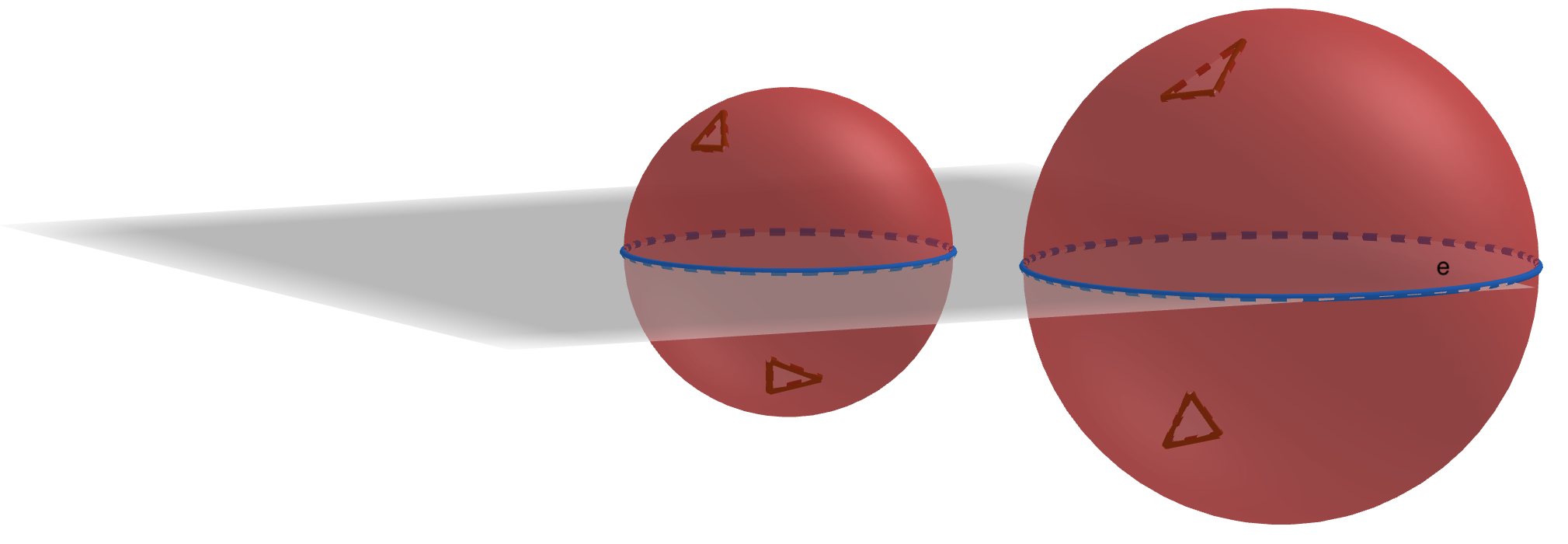}
\caption{The complex $\altcomplex_1$ consists of two disjoint triangulated spheres. We do not show the entire triangulation, only the four triangles of interest. The boundary of interest is the equator of the larger sphere on the right.}
\label{fig:sphere1}
\end{figure}

\begin{figure}

\includegraphics[clip,scale=0.18]{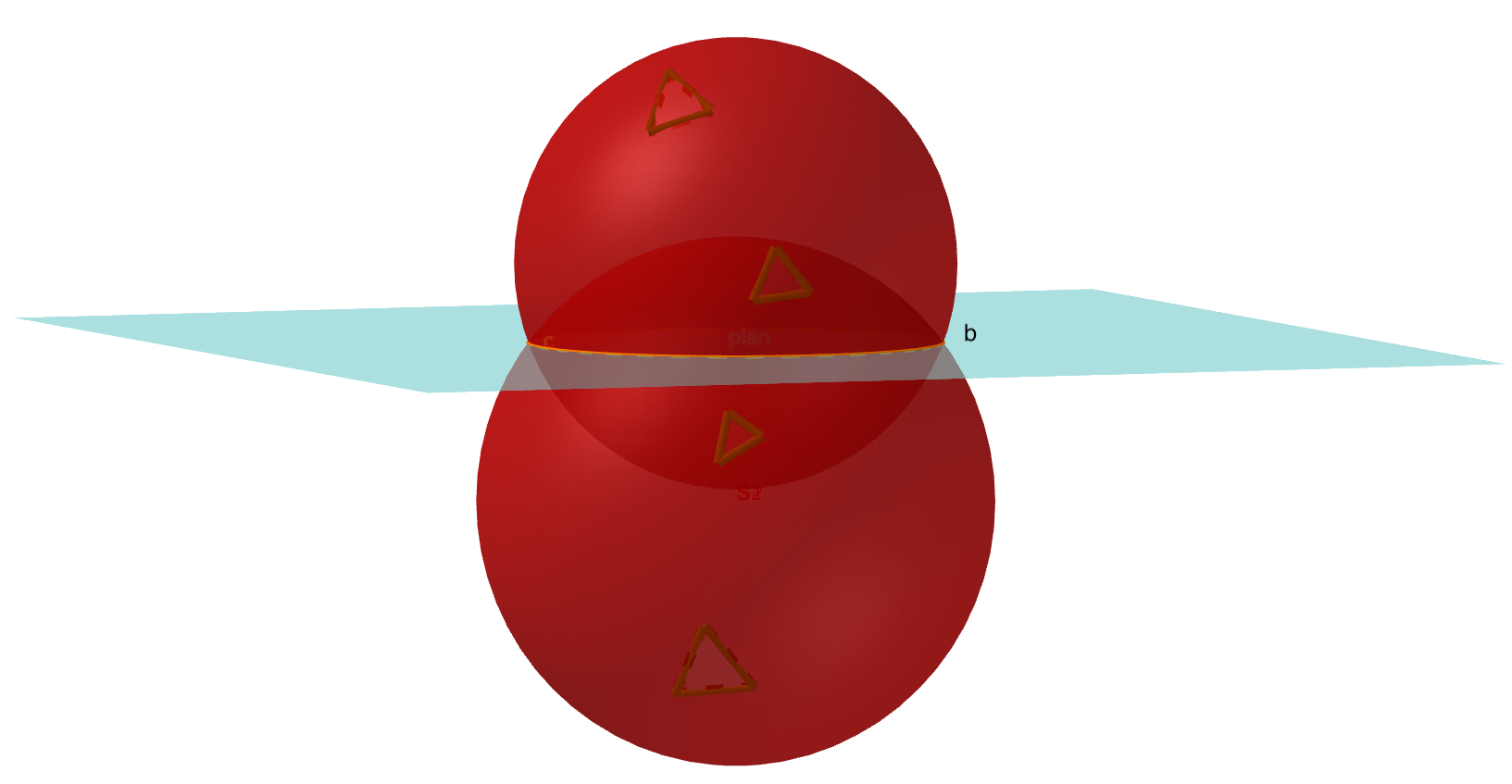}
\caption{The complex $\altcomplex_2$ consists of a triangulation of the  intersection of two  spheres. As before, we do not show the entire triangulation, only the four triangles of interest. The boundary of interest is the circle of intersection of the two spheres.}
\label{fig:sphere2}
\end{figure}

Incidentally, graphs happen to be $1$-dimensional simplicial complexes. And some of the cut problems have a natural homological interpretation. 
For instance, consider the following problem: What is the minimum number of edges you need to remove from 
a graph so that the vertices $\{s,t\}$ do not form a  bounding $0$-cycle of a  $1$-chain over $\mathbb{Z}_2$ in the resulting graph?
Since we have an  $s$-$t$ cut if and only if there are no paths connecting $s$ and $t$, it is easy to check that this problem is equivalent to finding the minimum $s$-$t$ cut  on graphs!
It is natural to ask the analogous question for complexes of higher dimension. 
 In particular, the question we ask, namely \createcycle, is the following one: 
 Given a bounding $\mathbb{Z}_2$ $r$-cycle $\zeta$ in a simplicial complex $\complex$, find a set $\SCC$ of $(r+1)$-dimensional simplices of minimum cardinality
so that the removal of $\SCC$ from $\complex$ makes $\zeta$ nontrivial.

For instance, consider the two complexes $\altcomplex_1$ and $\altcomplex_2$ shown in \Cref{fig:sphere1,fig:sphere2}, respectively.
For complex $\altcomplex_1$ shown in \Cref{fig:sphere1}, let the equator $e$ of the sphere on the right be the bounding $1$-cycle that we want to make nontrivial.  
Both hemispheres are bounded by the equator. So, the two highlighted triangles from the right sphere of the complex $\altcomplex_1$ constitute the optimal solution for  \createcycle.
That is, removing these two triangles makes $e$ a nontrivial $1$-cycle.
For complex $\altcomplex_2$ shown in \Cref{fig:sphere2}, the circle of  intersection of the two spheres is the bounding $1$-cycle of interest denoted by $b$.
Removing all the four highlighted triangles from complex $\altcomplex_2$ makes $b$ a nontrivial $1$-cycle.
This also happens to be the optimal solution for making $b$ nontrivial.

Complementary to the question of removing the minimal number of $r+1$-simplices in order to make a bounding cycle nontrivial, is the problem of removing the minimum number of $r$-simplices from a complex so that an entire homology class is destroyed.
More formally,  the problem  \hitcycles can be described as follows: given a nontrivial $\mathbb{Z}_2$ $r$-cycle $\zeta$ in a simplicial complex $\complex$, find a set $\SCC$ of $r$-dimensional simplices of minimum cardinality so that $\SCC$ meets every cycle homologous to $\zeta$. 

\hitcycles on graphs can be described as follows: Suppose we are given a graph $G$ with $k$ components. Let $C$ be one of the components of $G$. Then, $\beta_0(G) = |k|$, and each component 
determines a $0$-cycle. So the question of \hitcycles is to determine the minimum number of vertices  you need to remove so that $C$  is not a component anymore. The answer is trivial! 
One needs to remove all the vertices in $C$. For example in \Cref{fig:componentgraph}, $C_2$ ceases to be a component if and only if  all four vertices in $C_2$ are removed.  
It is worth noting that it is the unidimensionality of graphs that makes the problem trivial. What is more, even the \enquote{cut} aspect of the problem is not immediately visible for graphs.

In contrast, for higher-dimensional complexes, the problem has a distinct cut flavor. For instance, consider the planar complex shown in~\Cref{fig:planarcomplex}. The minimum number of edges that  need to be removed so that every cycle homologous to $\zeta$ is destroyed is three. In~\Cref{fig:planarcomplex}, an optimal set of edges is shown in red. Note that the edges happen to be in a \enquote{thin} portion of the complex, justifying our 
standpoint that (along with \createcycle) this problem can also be seen as a high dimensional cut problem.

In this work, we undertake an algorithmic study of the two high-dimensional cut problems: \createcycle and \hitcycles.

\begin{figure}
\usetikzlibrary{shapes.geometric}

\caption{Graph $G$ with three components}
\label{fig:componentgraph}

%
\caption{The figure shows two cycles that belong to $[\zeta]$ in green. Note that any cycle in $[\zeta]$ must pass through at least one of the three red edges.
Thus, the set of red edges constitutes an optimal solution for \hitcycles on this planar complex.}
\label{fig:planarcomplex}
\end{figure}

\subsection{Related work}

Duval et al.~\cite{duval} study the  vector spaces and integer lattices of cuts and flows associated to CW complexes and their relationships to group invariants.
Ghrist and Krishnan~\cite{ghrist} prove a topological version of the max-flow min-cut theorem for directed networks using methods from sheaf theory.
Then, there is also a long line of work on cuts in surface embedded graphs~\cite{Borradaile,borradaile2,Chamber2009,Chambers2010,Chambers2016,Chambers2019}, which is algorithmic in spirit and is loosely related to our work.

There is a growing body of work on parameterized complexity in topology~\cite{para1,para2,para3,para4,para5,para6,para7,para8}, and much of this paper  can be  characterized as such.

During the preparation of this article, we became aware of a recent paper by Maxwell and Nayyeri~\cite{maxwell} that studies problems similar to the ones we define but from a completely different point of view. While our focus was on surfaces and  parameterized complexity, the main focus of their work was to find out the extent to which the conceptual and the algorithmic framework of max-flow min-cut duality generalizes to the case of simplicial complexes. While we focus only on cuts, they study both cuts and flows.  

  We summarize the main results of Maxwell and Nayyeri~\cite{maxwell} as we understand them: They define a topological max-flow and a topological min-cut problem, and also a combinatorial min-cut problem. They show that unlike in the case of graphs, computing maximum integral flows  and combinatorial cuts on simplicial complexes is \NP-hard.  Moreover, they describe conditions under which the linear program gives the optimal value of a combinatorial cut, and also provide a generalization of the Ford-Fulkerson algorithm to the case of simplicial complexes. Their definition of combinatorial cut coincides with our definition of \createcycle, except for some important differences: they are interested in real coefficients and co-dimension one cycles, whereas we work with $\mathbb{Z}_2$ coefficients and cycles of all dimensions.  We implore the reader to  look up their interesting results~\cite{maxwell}.

 We note that while their paper is in the same spirit as ours, their focus is quite different from ours, and there is  very little overlap in terms of  hardness or algorithmic results.
In particular, they show \NP-hardness for combinatorial cuts with real coefficients, and we show \NP-hardness and \Wone-hardness for the same problem with $\mathbb{Z}_2$ coefficients.

\section{Summary of results}

{\bf Surfaces.} Our first result, expounded in \Cref{sec:surfaces}, is the following: \hitcycles admits a polynomial-time algorithm on triangulations of closed surfaces. At the heart of our proof lies an appealing characterization of the optimal solutions  in terms of the cocycles of the surface, which is of independent interest. Specifically, we show that a minimal solution set is necessarily a nontrivial cocycle. Further, we  show that the following are equivalent: 
\begin{inparaenum} 
\item A connected cocycle $\eta $ is a feasible set for  the input cycle $\zeta$. 
\item Every cycle in $[\zeta]$ intersects a connected cocycle $\eta $ in an odd number of edges. 
\item One of the cycles in $[\zeta]$ intersects a connected cocycle $\eta $ in an odd number of edges.
\end{inparaenum} 

 In particular, this allows us to identify the nontrivial cocycles that are solutions based on a parity-based property.
Having this characterization at hand, we proceed to characterize cohomology classes  that are  solutions. Eventually, we arrive at a very simple 3-step algorithm for \hitcycles on surfaces.

We remark that  \createcycle  is trivial for surfaces. In fact, it is easy to check that for some boundary $b$ and a $2$-chain $\zeta$, if $\partial \zeta = b$, then removing any one of the triangles that appears in the chain $\zeta$ makes $b$ nontrivial.

\bigskip
\noindent {\bf W[1]-hardness and NP-hardness.}
For general complexes, in \Cref{sub:wonetop},  we show that \hitcycles is \Wone-hard with respect
to the solution size $k$ as the parameter, (and hence, it is also \NP-hard). The proof is based on a reduction from the {\sc $k$-Multicolored Clique} problem. Here, the reduction shows the essence of hardness: its description is short, but its proof exposes various ``behaviors'' that we find interesting. In particular, the forward direction requires a nontrivial parity based argument, while the reverse direction shows how to ``trace'' a solution through the complex.

In addition, in \Cref{sub:createhard}, we show that \createcycle is also \Wone-hard with respect to the solution size $k$ as a parameter. The principles of this reduction follow the lines of the reduction for \hitcycles, though, here, both the description and the proof of the reduction are more involved because of subdivisions that help avoid some unhelpful incidences.

\bigskip
\noindent{\bf Fixed-parameter tractability.}
On the positive side, in  \Cref{sub:fpthit},
we show that \hitcycles admits an \fpt algorithm with respect to $k+\Delta$, where
$\Delta$ is the maximum degree of the Hasse graph of the complex $\complex$. Here, the main insight is that a minimal solution must be connected. Having this insight at hand, the algorithm follows: If we search across the geodesic ball of every $r$-simplex in the complex $\complex$, we will find a solution. 

In contrast, we observe that \createcycle does not admit this property because minimal solutions can be disconnected. This motivates the search of another parameter that makes the problem tractable. Exploiting the set-cover like structure of the problem, in \Cref{sub:createalgo}, we show that \createcycle with bounding  $r$-cycles as input has an $O(\log n)$-approximation \fpt algorithm with $\beta_{r+1}$ (the Betti number) as the parameter, when the input complex $\complex$ is  $(r+1)$-dimensional.
It is worth noting that  \createcycle is \Wone-hard  even for  $(r+1)$-dimensional complexes with solution size as the parameter since the hardness gadget used in  \Cref{sub:createhard} is $(r+1)$-dimensional.

By exploiting the vector space structure of the homology groups and the boundary groups, in \Cref{sec:randkill,sub:randfptapprox}, we provide a randomized \fpt algorithm for \killcycles and a randomized \fpt approximation algorithm for \createcycleg respectively.

\section{Preliminaries}

\subsection{Simplicial complexes}
A $k$\emph{-simplex} $\sigma$ is the convex hull of a set $V$ of $(k+1)$ affinely independent points in the Euclidean space of dimension $d\ge k$. We call $k$ the dimension of $\sigma$.
Any nonempty subset of $V$ also spans a simplex, which we call a \emph{face} of $\sigma$. A simplex $\sigma$ is said to be a \emph{coface} of a simplex  $\tau$ if and only if $\tau$ is face of $\sigma$. 
We say that $\sigma$ is a \emph{facet} of $\tau$, and $\tau$ a \emph{cofacet} of $\sigma$, if $\sigma$ is a face of $\tau$ with $\dim\sigma=\dim\tau-1$. We denote a facet-cofacet pair by $\sigma \prec \tau$.
A \emph{simplicial complex} $\complex$  is a collection of simplices that satisfies the following conditions:
\begin{compactitem}
\item	any face of a simplex in $\complex$ also belongs to $\complex$, and
\item	the intersection of two simplices $\sigma_1,\sigma_2\in \complex$ is either empty or a face of both $\sigma_1$ and $\sigma_2$. 
\end{compactitem}

An \emph{abstract simplicial complex} $\complex$  on a set of vertices $V$ is
a collection of subsets of $V$ that is closed under inclusion. The elements of $\complex$ are called its \emph{simplices}. 
An abstract simplicial complex $\altcomplex$ is said to be a \emph{subcomplex} of $\complex$ if every simplex of $\altcomplex$  belongs to $\complex$.

The collection of vertex sets of simplices in a geometric simplicial complex forms
an abstract simplicial complex. On the other hand, an abstract simplicial
complex $\complex$ has a geometric realization $|\complex|$ obtained by
embedding the points in $V$ in general position in a high-dimensional Euclidean space. Then, the complex $|\complex|$ is defined
as $\bigcup_{\sigma\in\complex}|\sigma|$, where $|\sigma|$ denotes the
span of points in $\sigma$. It is not very difficult to show that
any two geometric realizations of an abstract simplicial complex are
homeomorphic. Hence, going forward, we do not distinguish between
abstract and geometric simplicial complexes.

The \emph{star} of a vertex $v$ of complex $\complex$, written $\str_{\complex} (v)$, 
is the subcomplex consisting of all faces of $\complex$ containing $v$, together with 
their faces.

Let $V$ be the vertex set of $\complex$, $W$ be the vertex set of $\altcomplex$ and $\phi$ be a map from $V$ to $W$.
If for every simplex $ \{ v_0, v_1, \dots, v_r\} \in \complex$, the vertices  $\{ \phi(v_0),\phi( v_1), \dots,\phi( v_r)\}$ span a simplex in $\altcomplex$, then the $\phi$ induces a map, say $f$, from $\complex$ to $\altcomplex$.
The induced map $f: \complex \to \altcomplex$, is said to be \emph{simplicial}. 

We will denote by $\complex^{(p)}$
the set of $p$-dimensional simplices in $\complex$, and $n_p$ the number of $p$-dimensional simplices in $\complex$. 
 The  complex induced by $\complex^{(p)}$ is called the \emph{$p$-dimensional
skeleton} of $\complex$, and is denoted by $\complex_{p}$.
Given a simplicial complex $\complex$, we denote by the $\hasse$, the \emph{Hasse graph} of $\complex$, which is simply the graph that has a node for every  simplex of the complex,  and an edge for every facet-cofacet pair.
 Given a triangulated closed surface $\complex$, we denote by  $D_\complex$,  the \emph{dual graph} of $\complex$, which is simply the graph that has a node for every $2$-simplex and an edge connecting two nodes if the corresponding $2$-simplices are incident on a common edge in the complex.
The \emph{stellar subdivision} of a simplex (or a polytope) is the complex formed by taking a cone over its boundary.


\begin{notation}We use $[m]$ to denote the set $\{1,2,...,m\}$ for any $m\in \mathbb{N}$.
\end{notation}

\subsection{Homology and cohomology}

In this work, we restrict our attention to simplicial homology
with $\mathbb{Z}_{2}$ coefficients. For a general introduction to algebraic
topology, we refer the reader to ~\cite{hatcher}. 
Below we give a brief description of homology over $\mathbb{Z}_{2}$.

Let $\complex$ be a connected simplicial complex. 
We consider formal sums of simplices with $\mathbb{Z}_{2}$ coefficients,
that is, sums of the form $\sum_{\sigma\in \complex^{(p)}}a_{\sigma}\sigma$,
where each $a_{\sigma}\in\{0,1\}$. The expression $\sum_{\sigma\in \complex^{(p)}}a_{\sigma}\sigma$
is called a $p$\emph{-chain}. Since chains can be added to each other,
they form an Abelian group, denoted by $\mathsf{C}_{p}(\complex)$. Since
we consider formal sums with coefficients coming from $\mathbb{Z}_{2}$,
which is a field, $\mathsf{C}_{p}(\complex)$, in this case, is a vector
space of dimension $n_{p}$ over $\mathbb{Z}_{2}$. The $p$-simplices in $\complex$ form a (natural)
basis for $\mathsf{C}_{p}(\complex)$.  This establishes a natural
one-to-one correspondence between elements of $\mathsf{C}_{p}(\complex)$
and subsets of $\complex^{(p)}$, and we will freely make use of this identification.
The \emph{boundary} of a $p$-simplex is a $(p-1)$-chain that corresponds
to the set of its $(p - 1)$-faces. This map can be linearly
extended from $p$-simplices to $p$-chains, where the boundary of
a chain is the $\mathbb{Z}_{2}$-sum of the boundaries of its elements.
The resulting \emph{boundary homomorphism} is
denoted by $\partial_{p}:\mathsf{C}_{p}(\complex)\to\mathsf{C}_{p-1}(\complex)$.
A chain $\zeta\in\mathsf{C}_{p}(\complex)$ is called a \emph{$p$-cycle} if $\partial_{p}\zeta=0$,
that is, $\zeta\in \ker \partial_{p}$. The group of $p$-dimensional
cycles is denoted by $\mathsf{Z}_{p}(\complex)$. As before, since we are
working with $\mathbb{Z}_{2}$ coefficients, $\mathsf{Z}_{p}(\complex)$
is a vector space over $\mathbb{Z}_{2}$. A chain $\eta\in\mathsf{C}_{p}(\complex)$
is said to be a \emph{$p$-boundary} if $\eta=\partial_{p+1}c$ for some chain
$c\in\mathsf{C}_{p+1}(\complex)$, that is, $\eta \in \im \partial_{p+1}$. The
vector space of $p$-dimensional boundaries is denoted by $\mathsf{B}_{p}(\complex)$.

In our case, $\mathsf{B}_{p}(\complex)$ is also a vector space, and in fact
a subspace of $\mathsf{C}_{p}(\complex)$.
Thus, we can consider the quotient
space $\mathsf{H}_{p}(\complex)=\mathsf{Z}_{p}(\complex)/\mathsf{B}_{p}(\complex)$. The elements of the vector space $\mathsf{H}_{p}(\complex)$, known as the $p$-th
\emph{homology} of $\complex$, are equivalence classes of $p$-cycles, called \emph{homology classes}
where $p$-cycles are said to be \emph{homologous} if their $\mathbb{Z}_{2}$-difference is a $p$-boundary.
For a $p$-cycle $\zeta$, its corresponding homology class is denoted by $[\zeta]$.
Bases of $\mathsf{B}_{p}(\complex)$, $\mathsf{Z}_{p}(\complex)$ and $\mathsf{H}_{p}(\complex)$
are called \emph{boundary bases}, \emph{cycle bases}, and \emph{homology bases}, respectively.
The dimension of the $p$-th homology of $\complex$ is called the $p$-th \emph{Betti number} of $\complex$, denoted by $\beta_p(\complex)$.

Using the natural bases for $\mathsf{C}_{p}(\complex)$ and $\mathsf{C}_{p-1}(\complex)$,
the matrix $[\partial_{p}\sigma_{1}\:\partial_{p}\sigma_{2}\cdot\cdot\cdot\partial_{p}\sigma_{n_{p}}]$
whose column vectors are boundaries of $p$-simplices is called the
\emph{$p$-th boundary matrix}. Abusing notation, we also denote the $p$-th boundary
matrix by $\partial_{p}$. 

The dual vector space of $\mathsf{C}_{p}(\complex)$ (the vector space of linear maps $\mathsf{C}_{p}(\complex) \to \mathbb{Z}_{2}$) is called the space of \emph{cochain}, denoted by $\mathsf{C}^{p}(\complex) = \operatorname{Hom}(\mathsf{C}_{p}(\complex),\mathbb{Z}_{2})$.
Again, there is a natural basis corresponding to the $p$-simplices of $\complex$, with a $p$-simplex $\sigma$ corresponding to the linear map~$\eta$ with values $\eta(\sigma)=1$ and $\eta(\rho)=0$ for every other $p$-simplex $\rho\neq\sigma$.
The adjoint map to the boundary map $\partial_{p+1} \colon \mathsf{C}_{p+1}(\complex) \to \mathsf{C}_{p}(\complex)$ is the \emph{coboundary map} $\delta_{p} \colon \mathsf{C}^{p}(\complex) \to \mathsf{C}^{p+1}(\complex)$.
Similarly to chains and boundary maps, we may define subspaces of \emph{cocycles} $\mathsf{Z}^{p}(\complex) = \ker \delta_{p}$ and \emph{coboundaries} $\mathsf{B}^{p}(\complex) = \im \delta_{p+1} \subseteq\mathsf{Z}^{p}(\complex)$, and form their quotient $\mathsf{H}^{p}(\complex) = \mathsf{Z}^{p}(\complex)/\mathsf{B}^{p}(\complex)$, which is the \emph{cohomology} of $\complex$.
Again, for a $p$-cocycle $\eta$, the corresponding cohomology class is denoted by $[\eta]$.
The natural pairing of chains and cochains $\mathsf{C}_{p}(\complex) \times \mathsf{C}^{p}(\complex) \to \mathbb{Z}_{2}, (\zeta,\eta) \mapsto \eta(\zeta)$ induces a well-defined isomorphism $\mathsf{H}^{p}(\complex) \times \mathsf{H}_{p}(\complex) \to\mathbb{Z}_{2}, ([\zeta],[\eta]) \mapsto \eta(\zeta)$, identifying cohomology as the vector space dual to homology up to a natural isomorphism.

A set of $p$-cycles $\left\{ \zeta_{1},\dots,\zeta_{g}\right\} $
is called a \emph{homology cycle basis} if the set of classes $\left\{ [\zeta_{1}],\dots,[\zeta_{g}]\right\} $
forms a homology basis. For brevity, we abuse notation by using the term  \emph{ ($p$-th) homology
basis} for $\left\{ \zeta_{1},\dots,\zeta_{g}\right\} $. 
Similarly, a set of $p$-cocycles $\left\{\eta_{1},\dots,\eta_{g}\right\} $
is called a \emph{($p$-th) cohomology cocycle basis} if the set of classes $\left\{[\eta_{1}],\dots,[\eta_{g}]\right\} $
forms a cohomology basis.

Assigning non-negative weights to the edges of $\complex$,  the weight of a cycle is the sum of the weights of its edges, and the weight of a homology basis is the sum of the weights of the basis elements. We call the problem of computing a minimum weight  basis of $\hone$  the \emph{minimum homology basis} problem. Similarly, we call the problem of computing a minimum weight basis of $\honer$, the \emph{minimum cohomology basis} problem.

\begin{notation} Since there is a 1-to-1 correspondence between the $p$-chains of a complex $\complex$ and the subsets of $\complex^{(p)}$,  we abuse notation by writing $\partial \CCC$ in place of $\partial  (\sum_{\sigma\in \CCC} \sigma)$, for $\CCC\subset \complex^{(p)}$.  Likewise, for $p$-cochains $\delta  (\sum_{\tau\in \CCC'} \tau)$, we often write   $\delta \CCC'$.

We also abuse notation in the other direction. That is, we treat chains and cochains as sets.
 For instance, sometimes we say that a (co)chain $\gamma$ intersects a (co)chain $\zeta$, when we actually mean that the corresponding sets of simplices of the respective (co)chains intersect. Also, we say that a simplex $\sigma \in \zeta$, when indeed the simplex $\sigma$ belongs to the set associated to $\zeta$.
\end{notation}

\subsection{Parameterized complexity}

 Let $\Pi$ be an \NP-hard problem. In the framework of Parameterized Complexity, each instance of $\Pi$ is associated with a {\em parameter} $k$. Here, the goal is to confine the combinatorial explosion in the running time of an algorithm for $\Pi$ to depend only on $k$. Formally, we say that $\Pi$ is {\em fixed-parameter tractable (\fpt)} if any instance $(I, k)$ of $\Pi$ is solvable in time $f(k)\cdot |I|^{\OO(1)}$, where $f$ is an arbitrary computable function of $k$.
 
  A weaker request is that for every fixed $k$, the problem $\Pi$ would be solvable in polynomial time. 
 Formally, we say that $\Pi$ is {\em slice-wise polynomial (\XP)} if any instance $(I, k)$ of $\Pi$ is solvable in time $f(k)\cdot |I|^{g(k)}$, where $f$ and $g$ are arbitrary computable functions of $k$. 
In other words, for a fixed $k$,  $\Pi$ has a polynomial time algorithm, and we refer to such an algorithm as an \XP algorithm for $\Pi$.
 Nowadays, Parameterized Complexity supplies a rich toolkit to design \fpt and \XP algorithms~\cite{DBLP:series/txcs/DowneyF13,DBLP:books/sp/CyganFKLMPPS15,fomin2019kernelization}.

Parameterized Complexity also provides methods to show that a problem is unlikely to be \fpt. The main technique is the one of parameterized reductions analogous to those employed in classical complexity. Here, the concept of \WWW-hardness replaces the one of NP-hardness, and for reductions we need not only construct an equivalent instance in \fpt time, but also ensure that the size of the parameter in the new instance depends only on the size of the parameter in the original one.

\begin{definition}[{\bf Parameterized Reduction}]\label{definition:parameterized-reduction}
Let $\Pi$ and $\Pi'$ be two parameterized problems. A {\em parameterized reduction} from $\Pi$ to $\Pi'$ is an algorithm that, given an instance $(I,k)$ of $\Pi$, outputs an instance $(I',k')$ of $\Pi'$ such that:
\begin{itemize}
\item $(I,k)$ is a yes-instance of $\Pi$ if and only if $(I',k')$ is a yes-instance of $\Pi'$.
\item $k' \leq g(k)$ for some computable function $g$.
\item The running time is $f(k) \cdot |\Pi|^{O(1)}$ for some computable function $f$.
\end{itemize}
\end{definition}

If there exists such a reduction transforming a problem known to be \Wone-hard to another problem $\Pi$, then the problem $\Pi$ is \Wone-hard as well. Central \Wone-hard problems include, for example, deciding whether a nondeterministic single-tape Turing machine accepts within $k$ steps, {\sc Clique} parameterized by solution size, and {\sc Independent Set} parameterized by solution size. To show that a problem $\Pi$ is not \XP unless $\PPP = \NP$, it is \emph{sufficient} to show that there exists a fixed $k$ such that $\Pi$ is \NP-hard. If the  problem $\Pi$ is in \NP for a fixed $k$ then it is said to be in para-\NP, and if is  \NP-hard for a fixed $k$ then it is said to be para-\NP-hard.

Now, suppose that the parameter $k$ does not depend on the sought solution size, but it is a structural parameter. Then, we say that a minimization (maximization) problem $\Pi$ admits a $c$-approximation \fpt (with respect to $k$) if it admits a $f(k)\cdot |I|^{\OO(1)}$-time algorithm that, given an instance $(I, k)$ of $\Pi$, outputs a solution for $(I, k)$ that is larger (smaller) than the optimal solution for $(I, k)$ by a factor of at most $c$. When the parameter $k$ does depend on the sought solution size, the notion of a $c$-approximation \fpt algorithm is defined as well, but this definition is slightly more complicated and is not required in this paper. For more information on Parameterized Complexity, we refer the reader to recent books such as \cite{DBLP:series/txcs/DowneyF13,DBLP:books/sp/CyganFKLMPPS15,fomin2019kernelization}.

 \section{Problem definitions}
 In this section, we define the two key problems of interest, namely, \hitcycles and \createcycle along with their global variants, namely, \killcycles and \createcycleg, respectively.
 Also, we observe that all four problems lie in \NP and in \XP with respect to the solution size as the parameter. 
\subsection{\hitcycles} \label{sub:probdefone}


\begin{problem}[\hitcycles]
{Given a $d$-dimensional simplicial complex $\complex$,  a natural number $k$,  a natural number $r<d$ and a non-bounding cycle $\zeta \in \cycr(\complex)$.}
{$k$.}
{Does there exists a set $\SCC$ of $	r$-dimensional simplices with $|\SCC| \leq k$ such that $\SCC$ meets every cycle homologous to $\zeta$? 
}
\end{problem}


Let $\delcomplex$ denote the complex obtained from $\complex$ upon removal of the set of $r$-simplices $\CCC$ along with all the cofaces of the simplices in $\CCC$.
 In particular, the homology class $[\zeta]$ does not survive in $\delcomplex$.



Let $\{ \alpha_i \}$ for $i \in [\beta_r(\delcomplex)]$ be a homology basis for $\delcomplex$.
The inclusion map $\iota: \delcomplex \inclusion \complex$ induces a  map $\hat{\iota}:  \cycr(\delcomplex)  \to  \cycr(\complex)$ and also a map  $\tilde{\iota}:  \homr(\delcomplex)  \to  \homr(\complex)$.
Let $\hat{\alpha_i} = \hat{\iota}(\alpha_i)$.
Let $\boldA$ denote the matrix with nontrivial $r$-cycles $\hat{\alpha_i}$ as its columns. Let $\boldM$ denote the matrix $ [\boldA \,\, | \,\, \partial_{r+1}(\complex)]$ and $C(\boldM)$ the column space of $\boldM$.
The following lemma ensures polynomial time verification for the decision variant of  \hitcycles.

\begin{lemma}\label{lem:colspace}  $\zeta \notin$ column space of $\boldM$ if and only if $\SCC$ meets every cycle homologous to $\zeta$.
 \end{lemma}
 \begin{proof}
  ($\Longrightarrow$) Let $\rho$ be a cycle homologous to $\zeta$ such that $\SCC$ does not meet $\rho$.  
 Then, $\rho \in C(\boldM)$ since it survives in $\delcomplex$. The claim follows from observing that $\zeta$ is homologous to $\rho$.
 
 ($\Longleftarrow$) Suppose that $\SCC$ meets every cycle that is homologous to $\zeta$. Thus, at least one simplex is removed from every cycle homologous to $\zeta$. Then, a cycle homologous to $\zeta$ (in $\complex$) is not present in $\delcomplex$.
 The claim follows.
 \end{proof}
 
 \Cref{lem:colspace} provides an easy way to check if a set constitutes a feasible solution. 
 
 \begin{theorem} \label{thm:easycheck} Checking if a set $\SCC$ is a feasible solution to \hitcycles amounts to solving a linear system of equations, and can be done in $O(n^{\omega})$ time, where $\omega$ is the exponent of matrix multiplication, and $n$ is the size of the complex. 
 \end{theorem}
 
 \begin{corollary} \hitcycles is in \NP, and is in \XP with respect to the  solution size $k$ as the parameter.
 \end{corollary}

We now define the global variant of  \hitcycles. 

\begin{problem}[\killcycles]
{Given a $d$-dimensional simplicial complex $\complex$,  a natural number $k$,  a natural number $r<d$. }
{$k$.}
{ Does there exists a set $\SCC$ of $	r$-dimensional simplices with $|\SCC| \leq k$ such that the induced map on homology  $\tilde{\iota}:  \homr(\delcomplextwo)  \to  \homr(\complex)$ is non-surjective? 
}
\end{problem}


For a complex $\altcomplex$, let $\HCC_r(\altcomplex)$ denote an $r$-th homology basis of  $\altcomplex$. It is well-known that such a basis can always be computed in polynomial time.

 \begin{theorem} \killcycles is in  \NP, and in \XP with respect to the  solution size $k$ as the parameter.
 \end{theorem}
 \begin{proof}   $\SCC$ is a  solution for \killcycles if and only if one of the two conditions is satisfied:
 \begin{compactitem}
 \item $\beta_r(\delcomplextwo) < \beta_r(\complex)$ or
\item $\beta_r(\delcomplextwo) \geq \beta_r(\complex)$ and in the column rank profile of the matrix $[ \partial_{r+1}(\complex) \, | \, \HCC_r(\delcomplextwo) ]$, of the last $\beta_r(\delcomplextwo) $ columns, exactly $\beta_r(\complex)$ columns are nonzero.
 \end{compactitem}
The two conditions can be verified in polynomial time, proving the claim.
 \end{proof}

 \subsection{\createcycle}

\begin{problem}[\createcycle]
{Given a $d$-dimensional simplicial complex $\complex$,  a natural number $k$, a natural number $r  < d$ and a bounding cycle $\zeta \in \boundr(\complex)$.}
{$k$.}
{Does there exists a set $\SCC$ of $	r+1$-dimensional simplices with $|\SCC| \leq k$ such that removal of $\SCC$ from the $\complex$ makes $\zeta$ non-bounding? 
}
\end{problem}

 \begin{theorem} \createcycle is in \NP, and is in \XP with respect to the  solution size $k$ as the parameter.
 \end{theorem}
\begin{proof} A set $\SCC$ is a solution if and only if the system of equations $\partial_{r+1} (\delcomplextwo)\cdot \boldx = \zeta $ has no solution, which can be checked in polynomial time.
\end{proof}
 
 The global variant of  \createcycle can be described as follows.
 
 \begin{problem}[\createcycleg]
{Given a $d$-dimensional simplicial complex $\complex$,  a natural number $k$, a natural number $r  < d$ and a bounding cycle $\zeta \in \boundr(\complex)$.}
{$k$.}
{Does there exists a set $\SCC$ of $	r+1$-dimensional simplices with $|\SCC| \leq k$ such that the column space of $\partial_{r+1}(\complex_{\SCC})$ is a strictly smaller subspace of the column space of $\partial_{r+1}(\complex)$?
}
\end{problem}

 \begin{theorem} \createcycleg is in \NP, and is in \XP with respect to the  solution size $k$ as the parameter.
 \end{theorem}
\begin{proof} It is easy to check in polynomial time if  the column space of $\partial_{r+1}(\complex_{\SCC})$ is a strictly smaller  subspace of $\partial_{r+1}(\complex)$.
\end{proof}

\section{\hitcycles on surfaces} \label{sec:surfaces}

In this section we describe a polynomial time algorithm for \hitcycles on surfaces.
Let $\zeta$ be a nontrivial $1$-cycle in a  triangulated closed surface $\complex$.
The algorithm for surfaces has a very simple high-level description as detailed in \Cref{alg:surfacekiller}.

\begin{notation} \label{not:greatnot}
Note that if we evaluate the $r$-cocycle $\eta$ at an $r$-cycle $\zeta$, then by linearity,
\[\eta(\zeta)=\eta(\sum_{\sigma_{i}\in\zeta}\sigma_{i})=\sum_{\sigma_{i}\in\zeta}\eta(\sigma_{i}).\]
Because of $\mathbb{Z}_2$ addition, $\eta(\zeta)$ is either $0$ or $1$.
\end{notation}

 \begin{algorithm}[H]  
\caption{ The algorithm for \hitcycles on surfaces with input cycle $\zeta$  }\label{alg:surfacekiller}
\begin{algorithmic}[1]

\State{Find the optimal cohomology basis of $\complex$ with unit weights on edges.}
\State{Arrange the cocycles in the basis in ascending order of weight.}
\State{Pick the smallest weight cocycle $\eta$ with $\eta(\zeta) = 1$.}
\end{algorithmic}
\label{alg:surface}
\end{algorithm}

In what follows, we will establish a series of structural results about the solution set for \hitcycles on surfaces in order to prove the correctness of \Cref{alg:surfacekiller}.
We begin with a few definitions.


\begin{definition}[Connected cocycles]
A cocycle $\eta $ is said to be \emph{connected} if it induces a connected component in the dual graph, else we say that it is \emph{disconnected}.
\end{definition}

In \Cref{lem:yescocycle}, we show that a minimal solution is, in fact, a connected cocycle. 
Since cocycles can be potential solutions for \hitcycles, we make the following definitions.

\begin{definition}
We say that a cocycle  $\eta$ is  said to be a \emph{feasible  set} if every cycle $\zeta'\in [\zeta]$ meets $\eta$ in an edge. 
A cocycle that is not a feasible set, is said to be an \emph{infeasible set}. 
\end{definition}

Next, we provide a  useful characterization of cocycles that constitute feasible sets.

\begin{lemma} \label{lem:ctob}
If there exists a cycle in $[\zeta]$ that intersects a connected cocycle $\eta $ in an odd number of edges, then every cycle in $[\zeta]$ intersects  $\eta $ in an odd number of edges.
\end{lemma}
\begin{proof}
Suppose that a cycle $\gamma \in [\zeta]$ intersects $\eta $ in an odd number of edges. Let $\gamma' = \gamma + \partial \sigma$. We claim that $\gamma'$ intersects $\eta$ in an odd number of edges. We have four cases to consider.
\begin{enumerate}
\item The simplex boundary  $\partial \sigma$ is not incident on $\eta $. Then, the homologous cycle obtained by addition of $\partial \sigma$  maintains odd incidence. 
\item The simplex boundary  $\partial \sigma$  intersects $\eta $ in two edges and both edges also belong to  $\gamma$. Then, addition of $\partial \sigma$ to $\gamma$ reduces the number of incident edges on $\eta $ by two, and the number stays odd.
\item The simplex boundary  $\partial \sigma$  intersects $\eta $ in two edges and none of the edges belong to  $\gamma$. Then, addition of $\partial \sigma$ to $\gamma$ increases the number of incident edges on $\eta $ by two, and  the number stays odd.
\item The simplex boundary  $\partial \sigma$  intersects $\eta $ in two edges one of which belongs to  $\gamma$. Then, upon addition of $\partial \sigma$ to $\gamma$,  the incident edge is exchanged with the non-incident one, and the  incidence number stays the same.
\end{enumerate}
\Cref{fig:4case} illustrates the four cases. Any cycle in $[\zeta]$ can be obtained by adding simplex boundaries $\sum_i \partial \sigma_i$ to $\gamma $. 
So, applying the  four cases  inductively, we see that every cycle in  $[\zeta]$ has odd incidence on $\eta $. 
\end{proof}

\begin{figure}
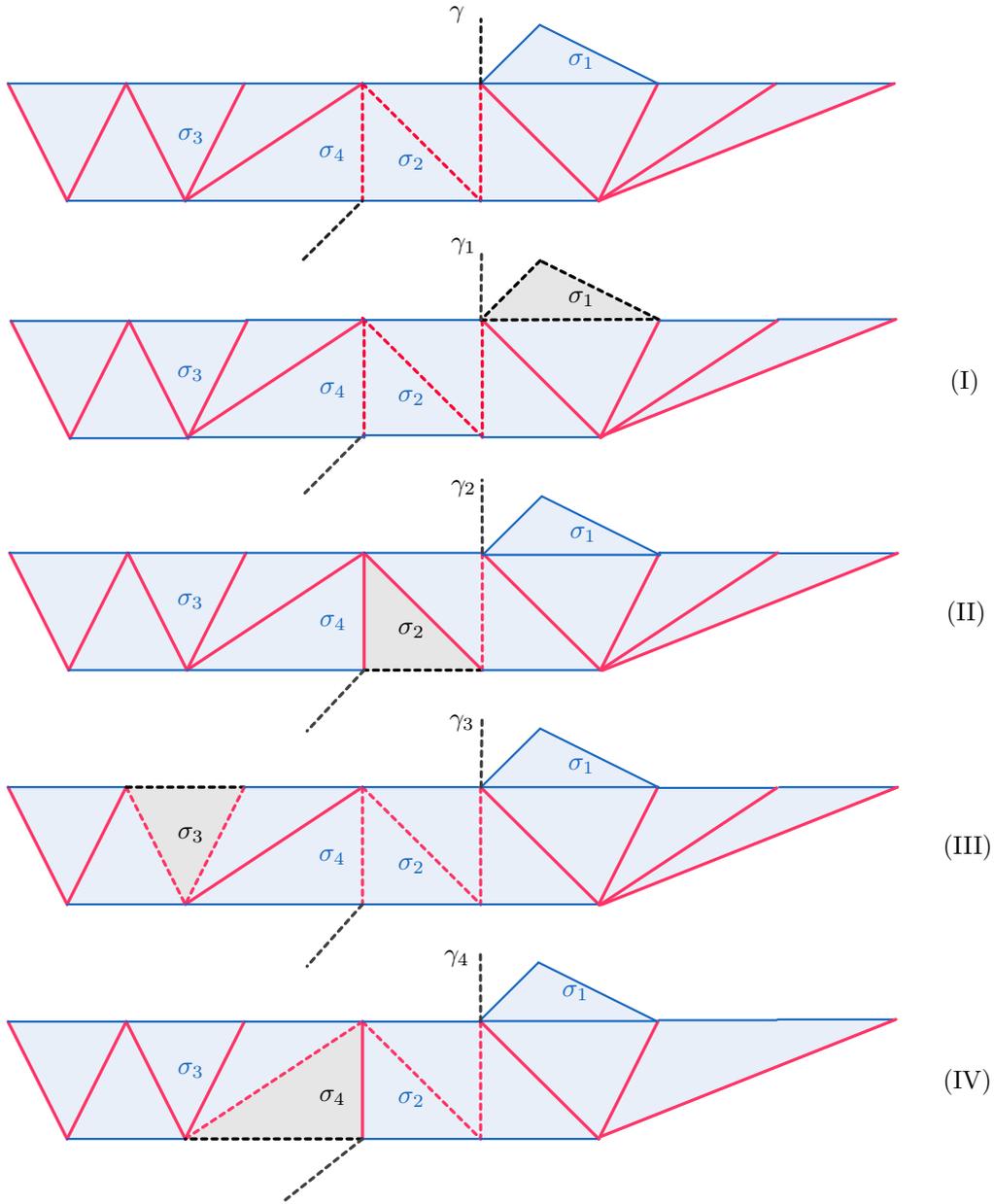

\definecolor{uququq}{rgb}{0.25098039215686274,0.25098039215686274,0.25098039215686274}
\definecolor{sqsqsq}{rgb}{0.12549019607843137,0.12549019607843137,0.12549019607843137}
\definecolor{ffqqtt}{rgb}{1,0,0.2}
\definecolor{ffttww}{rgb}{1,0.2,0.4}
\definecolor{rvwvcq}{rgb}{0.08235294117647059,0.396078431372549,0.7529411764705882}

\caption{In this figure, we illustrate the four cases discussed in \Cref{lem:ctob}. Let $\zeta$ be the input cycle and $\gamma \in [\zeta]$. The dotted edges in the topmost figure belong to $\gamma$.
Let $\eta$ be a cocycle that meets $\gamma$ in an odd number of edges. The edges of $\eta$ in each of the 5 figures are shown in pink. The part of $\gamma$ that does not intersect $\eta$ is shown in black and the part of  $\gamma$
that intersects $\eta$ is shown in pink. As shown in the topmost figure, $\gamma$ intersects $\eta$ in three (pink-dotted) edges. 
In the four cases labelled (I-IV), $\gamma_i = \gamma + \partial \sigma_i$, where $i \in [4]$. As before, for every $i\in [4]$, the edges of $\gamma_i$ that intersect the edges of $\eta$ are shown as pink-dotted edges and the edges of  $\gamma_i$  that do not 
intersect $\eta$ are shown as black-dotted edges. Note that in each of the four cases, the number of pink-dotted edges is odd.}
\label{fig:4case}
\end{figure}

Note that any connected cocycle $\eta $ induces a cycle graph  which is a subgraph of the dual graph $\dualgraph$ of the surface $\surface$. We denote the cycle graph by  $C_{\eta }$.

\begin{lemma} \label{lem:yescocycle} A minimal solution set is a  cocycle $\eta$ that induces a circle subgraph  $C_{\eta }$ in the dual graph $D_\complex$. 
\end{lemma}
\begin{proof} 
Let $e_1$ be an edge in the minimal solution set $\SCC$. Let $\sigma$ be a $2$-simplex incident on $e_1$, and let $e_2$ and $e_3$  be the other two edges incident on $\sigma$. 
Let $\gamma\in [\zeta]$ be a cycle with $e_1$ as the \emph{unique} edge incident on $\SCC$. We know that such a cycle exists because of minimality of $\SCC$. Then, there exists a cycle $\gamma' = \gamma  + \partial \sigma$ with $e_2$ and $e_3$  incident on it. 
Since $\gamma$ and $\gamma'$ differ only by a boundary  $\partial \sigma$, using the fact that $e_1$ is the unique edge incident on $\SCC$, either $e_2$ or $e_3$ must be incident on $\SCC$. 
Without loss of generality, assume that $e_2$ is incident on $\SCC$. Now, consider the $2$-simplex $\tau \neq \sigma$ incident on $e_2$. Using the same argument as before, and proceeding by induction, 
we obtain a sequence of edges  in $\SCC$ starting from $e_1$,  each connected by a $2$-simplex. Then, there must exist a sequence starting at $e_1$ and  ending at an edge $e'$ such that both $e'$ and $e_1$ are  incident on a common cofacet $\rho \neq \sigma$, for if this is not the case, then we can find a a cycle $\gamma' \in[\zeta]$ which is not incident on $\SCC$.
The sequence of edges from $e_1$ to $e'$ forms a cocycle, say $\eta $, where $\eta  \subset \SCC$.

Targeting a contradiction, assume that $\eta  \neq \SCC$.  By \Cref{lem:connected}, $\SCC$ induces a connected subgraph in the Hasse graph, which implies that there exists an edge $e' \in  \SCC \setminus \eta $  and a $2$-simplex $\tau$  such that $e' \prec \tau$ and the other two edges $e_1,e_2  \prec \tau$ belong to $\eta $. 
Let  $\gamma$ be a cycle with $e_1$ as the unique edge from $\SCC$ incident on it, and  let $\gamma'$ be a cycle with $e'$ as the unique edge from $\SCC$ incident on it.  Let $\BCC$ be the set of boundaries of $2$-simplices added to $\gamma$ in order to obtain $\gamma'$ from $\gamma$.
Since $\gamma$ intersects $\eta$ in an odd number of edges, using \Cref{lem:ctob}, any cycle homologous to $\gamma$ will also  intersect $\eta$ in an odd number of edges.
 Hence,  $\gamma'$  is incident on at least one of the edges of $\eta$. 
But this contradicts the existence of $\gamma'$ since by assumption, $\gamma'$ is not incident on $\eta$. Therefore, an edge $e'\in  \SCC \setminus \eta $  that shares a cofacet incident on two of the edges of $\eta $ does not exist.

 Finally, using \Cref{lem:connected} from \Cref{sub:fpthit}, for any edge $f' \in  \SCC \setminus \eta $, and  a cofacet  $\rho$ of $f'$, there must be a path in the dual graph from $\rho$  to a $2$-simplex $\varrho$ incident on  two of the edges of $\eta $. But for such a path to exist, there must exist an edge $e'\in  \SCC \setminus \eta $ and a $2$-simplex $\tau$  such that $e'\prec \tau$ and the other two edges $e_1,e_2  \prec \tau$ belong to $\eta $.  But we showed that such an edge $e'$ does not exist. 
 So, the set   $\SCC \setminus \eta $ is empty.  
 Since $\eta  = \SCC$, the claim follows. Please refer to \Cref{fig:yescocycle} for an example. Note that the final part of the argument is specific to surfaces.
\end{proof}

\begin{figure}

\begin{tikzpicture}[scale = 0.9,  line cap=round,line join=round,>=triangle 45,x=1cm,y=1cm]
\clip(-9.007213290378077,0) rectangle (6.007555799346402,7.3);
\fill[line width=1.2pt,color=rvwvcq,fill=rvwvcq,fill opacity=0.10000000149011612] (-8,3) -- (-7,1) -- (-6,3) -- cycle;
\fill[line width=1.2pt,color=rvwvcq,fill=rvwvcq,fill opacity=0.10000000149011612] (-6,3) -- (-5,1) -- (-7,1) -- cycle;
\fill[line width=1.2pt,color=rvwvcq,fill=rvwvcq,fill opacity=0.10000000149011612] (-6,3) -- (-4,3) -- (-5,1) -- cycle;
\fill[line width=1.2pt,color=rvwvcq,fill=rvwvcq,fill opacity=0.10000000149011612] (-4,3) -- (-1.9835183654338917,3.0659265382644323) -- (-5,1) -- cycle;
\fill[line width=1.2pt,color=rvwvcq,fill=rvwvcq,fill opacity=0.10000000149011612] (-1.9835183654338917,3.0659265382644323) -- (-3,1) -- (-5,1) -- cycle;
\fill[line width=1.2pt,color=rvwvcq,fill=rvwvcq,fill opacity=0.10000000149011612] (-1.9835183654338917,3.0659265382644323) -- (-1,1) -- (-3,1) -- cycle;
\fill[line width=1.2pt,color=rvwvcq,fill=rvwvcq,fill opacity=0.10000000149011612] (-1.9835183654338917,3.0659265382644323) -- (1,1) -- (-1,1) -- cycle;
\fill[line width=1.2pt,color=rvwvcq,fill=rvwvcq,fill opacity=0.10000000149011612] (-1.9835183654338917,3.0659265382644323) -- (0,3) -- (1,1) -- cycle;
\fill[line width=1.2pt,color=rvwvcq,fill=rvwvcq,fill opacity=0.10000000149011612] (0,3) -- (2,3) -- (1,1) -- cycle;
\fill[line width=1.2pt,color=rvwvcq,fill=rvwvcq,fill opacity=0.10000000149011612] (2,3) -- (3,1) -- (1,1) -- cycle;
\fill[line width=1.2pt,color=rvwvcq,fill=rvwvcq,fill opacity=0.10000000149011612] (2,3) -- (4,3) -- (3,1) -- cycle;
\fill[line width=1.2pt,color=rvwvcq,fill=rvwvcq,fill opacity=0.10000000149011612] (4,3) -- (5,1) -- (3,1) -- cycle;
\fill[line width=1.2pt,color=rvwvcq,fill=rvwvcq,fill opacity=0.10000000149011612] (-4,3) -- (-3,5) -- (-1.9835183654338917,3.0659265382644323) -- cycle;
\fill[line width=1.2pt,color=rvwvcq,fill=rvwvcq,fill opacity=0.10000000149011612] (-3,5) -- (-1,5) -- (-1.9835183654338917,3.0659265382644323) -- cycle;
\fill[line width=1.2pt,color=rvwvcq,fill=rvwvcq,fill opacity=0.10000000149011612] (-3,5) -- (-1,7) -- (-1,5) -- cycle;
\draw [line width=2.2pt,color=ffttww] (-8,3)-- (-7,1);
\draw [line width=2.2pt,color=ffttww] (-7,1)-- (-6,3);
\draw [line width=1.2pt,color=rvwvcq] (-6,3)-- (-8,3);
\draw [line width=2.2pt,color=ffttww] (-6,3)-- (-5,1);
\draw [line width=1.2pt,color=rvwvcq] (-5,1)-- (-7,1);
\draw [line width=1.2pt,color=rvwvcq] (-6,3)-- (-4,3);
\draw [line width=2.2pt,color=ffttww] (-4,3)-- (-5,1);

\draw [line width=2.2pt,color=ffttww] (-1.9835183654338917,3.0659265382644323)-- (-5,1);
\draw [line width=2.2pt,color=ffttww] (-1.9835183654338917,3.0659265382644323)-- (-3,1);
\draw [line width=1.2pt,color=rvwvcq] (-3,1)-- (-5,1);
\draw [line width=2.2pt,color=ffttww] (-1.9835183654338917,3.0659265382644323)-- (-1,1);
\draw [line width=1.2pt,color=rvwvcq] (-1,1)-- (-3,1);
\draw [line width=2.2pt,color=ffttww] (-1.9835183654338917,3.0659265382644323)-- (1,1);
\draw [line width=1.2pt,color=rvwvcq] (1,1)-- (-1,1);
\draw [line width=1.2pt,color=rvwvcq] (-1.9835183654338917,3.0659265382644323)-- (0,3);
\draw [line width=2.2pt,color=ffttww] (0,3)-- (1,1);
\draw [line width=1.2pt,color=rvwvcq] (0,3)-- (2,3);
\draw [line width=2.2pt,color=ffttww] (2,3)-- (1,1);
\draw [line width=2.2pt,color=ffttww] (2,3)-- (3,1);
\draw [line width=1.2pt,color=rvwvcq] (3,1)-- (1,1);
\draw [line width=1.2pt,color=rvwvcq] (2,3)-- (4,3);
\draw [line width=2.2pt,color=ffttww] (4,3)-- (3,1);
\draw [line width=2.2pt,color=ffttww] (4,3)-- (5,1);
\draw [line width=1.2pt,color=rvwvcq] (5,1)-- (3,1);
\draw [line width=1.2pt,color=rvwvcq] (-4,3)-- (-3,5);

\draw [line width=1.2pt,color=rvwvcq] (-1.9835183654338917,3.0659265382644323)-- (-4,3);

\draw [line width=1.2pt,color=rvwvcq] (-1,5)-- (-1.9835183654338917,3.0659265382644323);
\draw [line width=1.2pt,color=rvwvcq] (-1.9835183654338917,3.0659265382644323)-- (-3,5);
\draw [line width=1.2pt,color=rvwvcq] (-3,5)-- (-1,7);

\draw [line width=1.2pt,color=rvwvcq] (-1,5)-- (-3,5);
\draw [line width=2.2pt,color=yqqqyq] (-4,3)-- (-1.9835183654338917,3.0659265382644323);
\draw [line width=2.2pt,color=yqqqyq] (-3,5)-- (-1.9835183654338917,3.0659265382644323);
\draw [line width=2.2pt,color=yqqqyq] (-3,5)-- (-1,5);
\draw [line width=2.2pt,color=yqqqyq] (-1,7)-- (-1,5);

\begin{scriptsize}
\draw [fill=rvwvcq] (-8,3) circle (2.5pt);
\draw[color=rvwvcq] (-8.034796850977697,3.4281799897504857) node { \large  \textsf{ A }};
\draw [fill=rvwvcq] (-7,1) circle (2.5pt);
\draw[color=rvwvcq] (-7.029417142445105,0.6757470172104323) node { \large  \textsf{ B }};
\draw [fill=rvwvcq] (-6,3) circle (2.5pt);
\draw[color=rvwvcq] (-6.02403743391251,3.444661624316594) node { \large  \textsf{ C }};
\draw [fill=rvwvcq] (-5,1) circle (2.5pt);
\draw[color=rvwvcq] (-5.018657725379916,0.6592653826443241) node { \large  \textsf{ D }};
\draw [fill=rvwvcq] (-4,3) circle (2.5pt);
\draw[color=rvwvcq] (-4.02975965141343,3.494106528014918) node { \large  \textsf{ E }};
\draw[color=ffttww] (-4.787914841454403,2.5) node {\large $e_1$};
\draw [fill=rvwvcq] (-1.9835183654338917,3.0659265382644323) circle (2.5pt);
\draw[color=rvwvcq] (-2.0354818689143506,3.6094779699776747) node { \large  \textsf{ F }};
\draw[color=rvwvcq] (-3.667163690959052,2.3) node {\large $\tau$};
\draw[color=yqqqyq] (-3.0078983083147284,3.5105881625810262) node {\large $e'$};
\draw[color=ffttww] (-3.667163690959052,1.5) node {\large $e_2$};
\draw [fill=rvwvcq] (-3,1) circle (2.5pt);
\draw[color=rvwvcq] (-3.057343212013053,0.6592653826443241) node { \large  \textsf{ G }};
\draw [fill=rvwvcq] (-1,1) circle (2.5pt);
\draw[color=rvwvcq] (-1.0630654295139723,0.6757470172104323) node { \large  \textsf{ H }};
\draw [fill=rvwvcq] (1,1) circle (2.5pt);
\draw[color=rvwvcq] (0.99713889124954,0.6922286517765404) node { \large  \textsf{ I }};
\draw [fill=rvwvcq] (0,3) circle (2.5pt);
\draw[color=rvwvcq] (-0.02472245184916215,3.494106528014918) node { \large  \textsf{ J }};
\draw [fill=rvwvcq] (2,3) circle (2.5pt);
\draw[color=rvwvcq] (2.002518599782134,3.47762489344881) node { \large  \textsf{ K }};
\draw [fill=rvwvcq] (3,1) circle (2.5pt);
\draw[color=rvwvcq] (3.0243799428808362,0.7087102863426484) node { \large  \textsf{ L }};
\draw [fill=rvwvcq] (4,3) circle (2.5pt);
\draw[color=rvwvcq] (3.980314747715106,3.5105881625810262) node { \large  \textsf{ M }};
\draw [fill=rvwvcq] (5,1) circle (2.5pt);
\draw[color=rvwvcq] (4.9856944562477,0.6592653826443241) node { \large  \textsf{ N }};
\draw [fill=rvwvcq] (-3,5) circle (2.5pt);
\draw[color=rvwvcq] (-3.4693840761657553,5.076343446361296) node { \large  \textsf{ O }};
\draw[color=yqqqyq] (-2.150853310877107,4.417078063716972) node {\large $f_1$};
\draw [fill=rvwvcq] (-1,5) circle (2.5pt);
\draw[color=rvwvcq] (-0.6345429307951619,5.109306715493513) node { \large  \textsf{ P }};
\draw[color=yqqqyq] (-1.7717757158566207,5.537829214212323) node {\large $f_2$};
\draw [fill=rvwvcq] (-1,7) circle (2.5pt);
\draw[color=rvwvcq] (-1.4751062936666748,7.070621228860377) node { \large  \textsf{ Q }};
\draw[color=yqqqyq] (-0.7663960073240267,6.180612962290539) node {\large $f'$};
\end{scriptsize}
\end{tikzpicture}
\caption{The cocycle $\eta $ is shown in red. $e_1,e_2,e', f'$ and $\tau$ are as in \Cref{lem:yescocycle}. Note that there exists a path from $e'$ to $f'$ in the dual graph.}
\label{fig:yescocycle}
\end{figure}
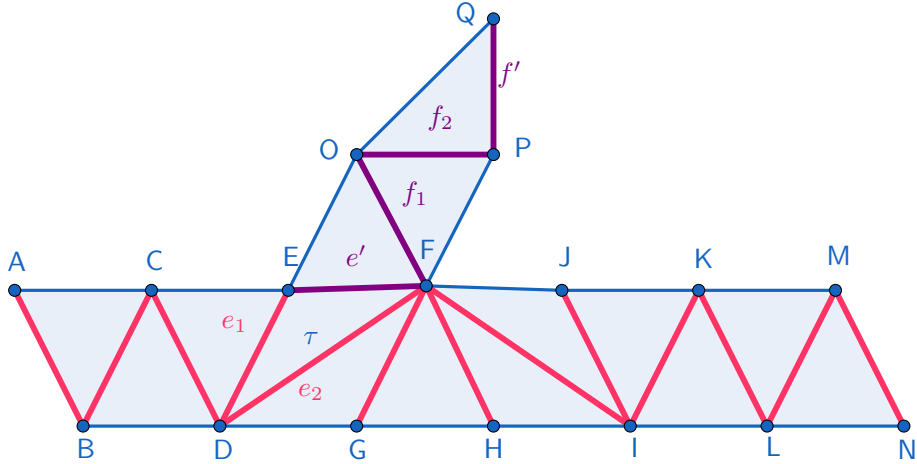

\begin{lemma} \label{lem:notrivial} A trivial cocycle is not a minimal solution set.
\end{lemma}
\begin{proof}  
Let $\eta $ be a trivial cocycle. Then, $\eta  = \delta (\SCC)$, where $\SCC$ is a collection of points. Let $e$ be an edge of  $\eta $. 
By the assumption on minimality of the solution set, there exists a cycle $\zeta' \in [\zeta]$  such that $e$ is the only edge of $\eta $ incident on $\zeta'$. One of the vertices of  $e$, say $v_1$, belongs to $\SCC$. 
We write the cycle $\zeta'$ as a sequence of vertices $v_1,v_2,\dots,v_q =v_1$, for some $q$, such that an edge connects subsequent vertices,  the sequence  starts and ends at $v_1$, and the edge $\{v_{q-1},v_1\} = e$.
Then, the path from $v_1$ to $v_{q-1}$ must pass through a vertex $v'$ such that $v' \in \SCC$. But this is only possible if $\zeta'$ also contains an edge of  $\eta $ other than $e$, which in turn, contradicts the minimality of the solution set.
Hence, a trivial cocycle is not a minimal solution set. See \Cref{fig:trivialco} for an example.
\end{proof}

\begin{figure}
\begin{tikzpicture}[scale = 1.8, line cap=round,line join=round,>=triangle 45,x=1cm,y=1cm]
\clip(-5.698968397031091,-0.4) rectangle (1.0632403725804274,4.4);
\fill[line width=1.2pt,color=rvwvcq,fill=rvwvcq,fill opacity=0.10000000149011612] (-5,3) -- (-4,3) -- (-4,2) -- cycle;
\fill[line width=1.2pt,color=rvwvcq,fill=rvwvcq,fill opacity=0.10000000149011612] (-4,3) -- (-3,3) -- (-4,4) -- cycle;
\fill[line width=1.2pt,color=rvwvcq,fill=rvwvcq,fill opacity=0.10000000149011612] (-3,3) -- (-2,3) -- (-2,4) -- cycle;
\fill[line width=1.2pt,color=rvwvcq,fill=rvwvcq,fill opacity=0.10000000149011612] (-2,3) -- (-2,2) -- (-1,3) -- cycle;
\fill[line width=1.2pt,color=rvwvcq,fill=rvwvcq,fill opacity=0.10000000149011612] (-2,2) -- (-2,1) -- (-1,1) -- cycle;
\fill[line width=1.2pt,color=rvwvcq,fill=rvwvcq,fill opacity=0.10000000149011612] (-3,1) -- (-2,0) -- (-2,1) -- cycle;
\fill[line width=1.2pt,color=rvwvcq,fill=rvwvcq,fill opacity=0.10000000149011612] (-4,1) -- (-4,0) -- (-3,1) -- cycle;
\fill[line width=1.2pt,color=rvwvcq,fill=rvwvcq,fill opacity=0.10000000149011612] (-4,2) -- (-5,1) -- (-4,1) -- cycle;
\fill[line width=1.2pt,color=rvwvcq,fill=rvwvcq,fill opacity=0.10000000149011612] (-4,2) -- (-3,3) -- (-4,3) -- cycle;
\fill[line width=1.2pt,color=rvwvcq,fill=rvwvcq,fill opacity=0.10000000149011612] (-4,2) -- (-3,2) -- (-3,3) -- cycle;
\fill[line width=1.2pt,color=rvwvcq,fill=rvwvcq,fill opacity=0.10000000149011612] (-3,3) -- (-3,2) -- (-2,2) -- cycle;
\fill[line width=1.2pt,color=rvwvcq,fill=rvwvcq,fill opacity=0.10000000149011612] (-3,3) -- (-2,3) -- (-2,2) -- cycle;
\fill[line width=1.2pt,color=rvwvcq,fill=rvwvcq,fill opacity=0.10000000149011612] (-3,2) -- (-4,1) -- (-3,1) -- cycle;
\fill[line width=1.2pt,color=rvwvcq,fill=rvwvcq,fill opacity=0.10000000149011612] (-4,2) -- (-4,1) -- (-3,2) -- cycle;
\fill[line width=1.2pt,color=rvwvcq,fill=rvwvcq,fill opacity=0.10000000149011612] (-3,2) -- (-3,1) -- (-2,2) -- cycle;
\fill[line width=1.2pt,color=rvwvcq,fill=rvwvcq,fill opacity=0.10000000149011612] (-3,1) -- (-2,1) -- (-2,2) -- cycle;
\draw [line width=2.2pt,color=yqqqyq] (-4,3)-- (-4,4);
\draw [line width=2.2pt,color=yqqqyq] (-3,3)-- (-4,4);
\draw [line width=2.2pt,color=ffqqqq] (-3,3)-- (-2,4);
\draw [line width=2.2pt,color=yqqqyq] (-2,4)-- (-2,3);
\draw [line width=2.2pt,color=yqqqyq] (-2,3)-- (-1,3);
\draw [line width=2.2pt,color=yqqqyq] (-1,3)-- (-2,2);
\draw [line width=2.2pt,color=yqqqyq] (-2,2)-- (-1,1);
\draw [line width=2.2pt,color=yqqqyq] (-1,1)-- (-2,1);
\draw [line width=2.2pt,color=yqqqyq] (-2,1)-- (-2,0);
\draw [line width=2.2pt,color=yqqqyq] (-2,0)-- (-3,1);
\draw [line width=2.2pt,color=yqqqyq] (-3,1)-- (-4,0);
\draw [line width=2.2pt,color=yqqqyq] (-4,0)-- (-4,1);
\draw [line width=2.2pt,color=yqqqyq] (-4,1)-- (-5,1);
\draw [line width=2.2pt,color=yqqqyq] (-5,1)-- (-4,2);
\draw [line width=2.2pt,color=yqqqyq] (-4,2)-- (-5,3);
\draw [line width=2.2pt,color=yqqqyq] (-5,3)-- (-4,3);
\draw [line width=1.2pt,color=rvwvcq] (-4,3)-- (-4,2);
\draw [line width=1.2pt,color=rvwvcq] (-2,3)-- (-2,2);
\draw [line width=1.2pt,color=rvwvcq] (-2,1)-- (-3,1);
\draw [line width=1.2pt,color=rvwvcq] (-3,1)-- (-4,1);
\draw [line width=1.2pt,color=rvwvcq] (-4,1)-- (-4,2);
\draw [line width=1.2pt,color=rvwvcq] (-4,2)-- (-3,3);
\draw [line width=1.2pt,color=rvwvcq] (-3,3)-- (-4,3);
\draw [line width=1.2pt,color=rvwvcq] (-4,2)-- (-3,2);
\draw [line width=1.2pt,dash pattern=on 2pt off 2pt,color=rvwvcq] (-3,2)-- (-3,3);
\draw [line width=1.2pt,color=rvwvcq] (-3,3)-- (-2,3);
\draw [line width=1.2pt,color=rvwvcq] (-2,2)-- (-3,3);
\draw [line width=1.2pt,dash pattern=on 2pt off 2pt,color=rvwvcq] (-3,2)-- (-4,1);
\draw [line width=1.2pt,color=rvwvcq] (-3,2)-- (-3,1);
\draw [line width=1.2pt,color=rvwvcq] (-3,1)-- (-2,2);
\draw [line width=1.2pt,color=rvwvcq] (-2,2)-- (-3,2);
\draw [line width=1.2pt,color=rvwvcq] (-2,1)-- (-2,2);

\draw [line width=2.2pt,color=yqqqyq] (-3,3)-- (-3,4);

\draw [line width=2.2pt,color=yqqqyq] (-2,2)-- (-1,2.3);
\draw [line width=2.2pt,color=yqqqyq] (-2,2)-- (-1,1.7);

\begin{scriptsize}


\draw [fill=rvwvcq] (-4,3) circle (2.5pt);
\draw[color=rvwvcq] (-3.838297745015343,3.242883136370304)node {\large  \textsf{ A }};

\draw [fill=rvwvcq] (-4,4) circle (2.5pt);
\draw[color=rvwvcq] (-4.040313415805624,4.2635938940475135)node {\large  \textsf{ B }};

\draw [fill=rvwvcq] (-3,3) circle (2.5pt);
\draw[color=rvwvcq] (-2.719602658128414,3.1641479438219125)node {\large  \textsf{ C }};

\draw [fill=rvwvcq] (-2,4) circle (2.5pt);
\draw[color=rvwvcq] (-1.977627092999595,4.2742262977733185)node {\large  \textsf{ D }};

\draw [fill=rvwvcq] (-3,4) circle (2.5pt);
\draw[color=rvwvcq] (-3,4.2742262977733185)node {\large  \textsf{ X }};

\draw [fill=rvwvcq] (-2,3) circle (2.5pt);
\draw[color=rvwvcq] (-1.8500382482899438,3.242883136370304)node {\large  \textsf{ E }};

\draw [fill=rvwvcq] (-1,3) circle (2.5pt);
\draw[color=rvwvcq] (-0.8186950868869292,3.1046618879348484)node {\large  \textsf{ F }};

\draw [fill=rvwvcq] (-2,2) circle (2.5pt);
\draw[color=rvwvcq] (-1.6049790184835097,1.9826863228060293)node {\large  \textsf{ G }};

\draw [fill=rvwvcq] (-1,1) circle (2.5pt);
\draw[color=rvwvcq] (-0.8399598943385378,1.0738727763062317)node {\large  \textsf{ H }};

\draw [fill=rvwvcq] (-1,2.3) circle (2.5pt);
\draw[color=rvwvcq] (-0.8399598943385378,2.3)node {\large  \textsf{ Y }};

\draw [fill=rvwvcq] (-1,1.7) circle (2.5pt);
\draw[color=rvwvcq] (-0.8399598943385378,1.7)node {\large  \textsf{ Z }};

\draw [fill=rvwvcq] (-2,1) circle (2.5pt);
\draw[color=rvwvcq] (-1.9138326706447695,1.2333588321932958)node {\large  \textsf{ I }};
\draw [fill=rvwvcq] (-2,0) circle (2.5pt);
\draw[color=rvwvcq] (-1.8287734408383354,0.10632403725804275)node {\large  \textsf{ J }};
\draw [fill=rvwvcq] (-3,1) circle (2.5pt);
\draw[color=rvwvcq] (-3.1259266953864566,1.201461621015883)node {\large  \textsf{ K }};
\draw [fill=rvwvcq] (-4,0) circle (2.5pt);
\draw[color=rvwvcq] (-4.37269856789471,0.06379442235482564)node {\large  \textsf{ L }};
\draw [fill=rvwvcq] (-4,1) circle (2.5pt);
\draw[color=rvwvcq] (-4.17853466424108,1.201461621015883)node {\large  \textsf{ M }};
\draw [fill=rvwvcq] (-5,1) circle (2.5pt);
\draw[color=rvwvcq] (-5.231142633095703,1.0951375837578403)node {\large  \textsf{ N }};
\draw [fill=rvwvcq] (-4,2) circle (2.5pt);
\draw[color=rvwvcq] (-4.2529614903217094,2.083951130257638)node {\large  \textsf{ O }};
\draw [fill=rvwvcq] (-5,3) circle (2.5pt);
\draw[color=rvwvcq] (-5.018494558579618,3.317309962450934)node {\large  \textsf{ P }};
\draw [fill=rvwvcq] (-3,2) circle (2.5pt);
\draw[color=rvwvcq] (-2.913278620870371,2.232804782418898)node {\large  \textsf{ Q }};
\end{scriptsize}
\end{tikzpicture}
\caption{The edges in purple and red together form a trivial cocycle $\eta $ given by $\delta(A+C+E+O+Q+G+M+K+I)$. If  $\eta $ forms a minimal solution set, then it intersects a cycle   $\zeta' \in [\zeta]$ in a unique edge, say, edge {DC}. But any path that passes through {DC} must pass through another edge of $\eta $. In this example, DC-CQ-QM-MN is one such path.}
\label{fig:trivialco}
\end{figure}

\Cref{lem:yescocycle,lem:notrivial}  combine to give the following theorem.

\begin{theorem}\label{thm:charcminsol} A minimal solution set is a nontrivial cocycle.
\end{theorem}



\begin{lemma} \label{lem:makesingle} If a connected cocycle $\eta $ intersects a cycle $\zeta_0 \in [\zeta]$ in $m$ edges, then  there exists another cycle $\gamma \in [\zeta]$ such that $\gamma$ also intersects $\eta $ in $m$ edges, and the intersection of $\gamma$ and $\eta $ induces a connected component in the dual graph. 
\end{lemma}
\begin{proof} 
To begin with, note that the intersection  of  $\zeta_0 \in [\zeta]$ with $\eta $ induces a (possibly disconnected) subgraph of the cycle graph $C_{\eta }$, which we denote by $C_{\eta }^{\zeta_0}$. 
Both $C_{\eta }$ and $C_{\eta }^{\zeta_0}$ are subgraphs of the dual graph $\dualgraph$.
Let $\CCC_1,\dots, \CCC_k$ be the $k$ connected components of $C_{\eta }^{\zeta_0}$. If $k=1$, then the lemma is already satisfied. So without loss of generality, assume $k>1$.
 We say that a component $\CCC_i$ is a \emph{neighbor} of a component $\CCC_j$ if there exists a vertex in $\CCC_i$ that has a path to a vertex in $\CCC_j$ that does not 
intersect  the edges of  $C_{\eta }^{\zeta_0}$. It is easy to check that every component of $C_{\eta }^{\zeta_0}$ has exactly two (possibly non-distinct) neighbors. 
Choose any two neighboring components $\CCC_i$  and $\CCC_j$ of $C_{\eta }^{\zeta_0}$. Let $v \in \CCC_i$ and $v' \in \CCC_j$ be two vertices that have a simple path $\PCC$   with vertices $ v = u_1, u_2, \dots, u_\ell = v'$ such that $\PCC$ does not intersect the edges of $C_{\eta }$. Every vertex $u_t$ corresponds to a simplex $\sigma_t $ in   $\surface$. Now, adding the simplex boundaries $\sum\limits_{t=2}^{\ell}\partial\sigma_{i}$ to $\zeta_0$ gives rise to a cycle homologous to $\zeta_0$ such that $\CCC_i$ has one vertex (and one edge) more and $\CCC_j$ has one vertex (and one edge) less. We repeat this process inductively until all edges  are \enquote{transported} from $\CCC_j$ to $\CCC_i$ and $\CCC_j$ becomes empty. That is, the new cycle $\zeta_1$ we obtain is homologous to $\zeta_0$ and has $k-1$ components.
We denote the  subgraph induced by intersection of $\zeta_1$ and $\eta $ by $C_{\eta }^{\zeta_1}$. 

We apply the same procedure to $C_{\eta }^{\zeta_1}$ as above and get a a cycle $\zeta_2 \in [\zeta]$ whose induced subgraph $C_{\eta }^{\zeta_2}$ has $k-2 $ components. Proceeding inductively, we finally obtain a cycle $ \gamma  = \zeta_{k-1}\in [\zeta]$
whose induced subgraph $C_{\eta }^{\zeta_{k-1}}$ is a connected subgraph of the dual graph. Moreover, 
by design, the total number of edges in every induced graph $C_{\eta }^{\zeta_{i}}$ for $i\in [0,k-1]$ is $m$.
\end{proof}

\begin{figure}
\definecolor{uququq}{rgb}{0.25098039215686274,0.25098039215686274,0.25098039215686274}
\definecolor{sqsqsq}{rgb}{0.12549019607843137,0.12549019607843137,0.12549019607843137}
\definecolor{ffttww}{rgb}{1,0.2,0.4}
\definecolor{rvwvcq}{rgb}{0.08235294117647059,0.396078431372549,0.7529411764705882}

\caption{In this figure, we provide an illustrative example of the cycle modification technique from \Cref{lem:makesingle}. 
As in \Cref{fig:4case}, the edges of a cycle that intersect the edges of the cocycle $\eta$ are shown as pink-dotted edges and the edges of  a cycle  that do not  intersect $\eta$ are shown as black-dotted edges. 
The cocycle $\eta$ is shown in pink. The intersection of $\zeta_0$  and $\eta$ induces a disconnected subgraph $\CCC_{\eta}^{\zeta_0}$ of the dual graph with several components, only three of which are shown, namely, $\CCC_1$, $\CCC_2$ and $\CCC_3$. 
Here, $\CCC_1$ and $\CCC_2$ are neighbors, and $\CCC_2$ and $\CCC_3$ are  neighbors. The path from $\CCC_1$ to $\CCC_2$ in $\dualgraph$ consists of simplices $\sigma_1,\sigma_2,\sigma_3$ (which are vertices in the dual graph).
So we add $\partial \sigma_2 + \partial \sigma_3$ to $\zeta_0$ to obtain $\zeta'$. Next, we add $\partial \sigma_3 + \partial \sigma_4$ to $\zeta'$ to obtain $\zeta''$. Finally, we add $\partial \sigma_4 + \partial \sigma_5$ to $\zeta''$ to obtain $\zeta_1$.
The number of connected components of $\CCC_{\eta}^{\zeta_1}$ is one less than the number of connected components of $\CCC_{\eta}^{\zeta_0}$. This was achieved by \enquote{transporting} edges in $\CCC_2$ to $\CCC_1$.}
\label{fig:movingit}
\end{figure}

\begin{lemma} \label{lem:equieven}
The following are equivalent.
\vspace{-0.25cm}
\begin{enumerate}[(a.)]
\item A connected cocycle $\eta $ is a feasible set for  the input cycle $\zeta$.
\item Every cycle in $[\zeta]$ intersects a connected cocycle $\eta $ in an odd number of edges.
\item There exists a cycle in $[\zeta]$ that intersects a connected cocycle $\eta $ in an odd number of edges.
\end{enumerate}
\end{lemma}
\begin{proof}
\begin{description}

\item[(a.) $\Longrightarrow$ (b.)] Assume that there exists a cycle $\xi \in [\zeta]$ that intersects $\eta $ in an even number of edges. 
The intersection  of  $\xi$ with $\eta $ induces a (possibly disconnected) subgraph of $C_{\eta }$, which we denote by $C_{\eta }^{\xi}$.
Using~\Cref{lem:makesingle}, there exists another cycle $\gamma $ homologous to $\xi$ such that $\gamma$ intersects $\eta $ in  the same number of edges, and the intersection of $\gamma$ and $\eta $  induces a connected graph in the dual graph, which we denote by $C_{\eta }^{\gamma}$.   
There are two cases:
\begin{description}
\item[Case 1] $C_{\eta }^{\gamma}$ is not identical to $C_{\eta }$ (the cycle graph induced by the entire cocycle $\eta $), and $C_{\eta }^{\gamma}$ is a path graph $\PCC$  with vertices $u_1, u_2, \dots, u_{m+1}$, where each vertex $u_i$ is distinct and corresponds to a simplex $\sigma_i $ in   $\surface$.
\item[Case 2] $C_{\eta }^{\gamma} $ is the same as the entire cycle graph $C_{\eta } $ with vertices  $u_1, u_2, \dots, u_{m},  u_{m+1}$, where $u_{m+1} = u_1$.
\end{description}
  In either case, upon adding the simplex boundaries $\sum\limits_{i=1}^{\nicefrac{m}{2}}\partial\sigma_{2i}$ to   $\gamma$, we obtain a cycle  that has an empty intersection with $\eta $. That is, there exists a cycle in $ [\zeta]$ which does not meet $\eta $ in any of its edges. In other words,  $\eta $ is an infeasible set.

\item[(b.) $\Longrightarrow$ (a.)] This is  true by the definition of a feasible set.

\item[(c.) $\Longrightarrow$ (b.)]  This is the content of \Cref{lem:ctob}.

\item[(b.) $\Longrightarrow$ (c.)] This is trivially true.  \qedhere
\end{description}
This completes the proof.
\end{proof}

\begin{figure}
\definecolor{uququq}{rgb}{0.25098039215686274,0.25098039215686274,0.25098039215686274}
\definecolor{ffttww}{rgb}{1,0.2,0.4}
\definecolor{rvwvcq}{rgb}{0.08235294117647059,0.396078431372549,0.7529411764705882}
\begin{tikzpicture}[scale=0.8,line cap=round,line join=round,>=triangle 45,x=1cm,y=1cm]
\clip(-9.200992354464158,-6.3585021401347) rectangle (9.625953943348346,5.902512864429701);
\fill[line width=2pt,color=rvwvcq,fill=rvwvcq,fill opacity=0.10000000149011612] (-8,4) -- (-7,2) -- (-6,4) -- cycle;
\fill[line width=2pt,color=rvwvcq,fill=rvwvcq,fill opacity=0.10000000149011612] (-6,4) -- (-5,2) -- (-7,2) -- cycle;
\fill[line width=2pt,color=rvwvcq,fill=rvwvcq,fill opacity=0.10000000149011612] (-6,4) -- (-4,4) -- (-5,2) -- cycle;
\fill[line width=2pt,color=rvwvcq,fill=rvwvcq,fill opacity=0.10000000149011612] (-4,4) -- (-2,4) -- (-5,2) -- cycle;
\fill[line width=2pt,color=rvwvcq,fill=rvwvcq,fill opacity=0.10000000149011612] (-2,4) -- (-2,2) -- (-5,2) -- cycle;
\fill[line width=2pt,color=rvwvcq,fill=rvwvcq,fill opacity=0.10000000149011612] (-2,4) -- (0,2) -- (-2,2) -- cycle;
\fill[line width=2pt,color=rvwvcq,fill=rvwvcq,fill opacity=0.10000000149011612] (-2,4) -- (0,4) -- (0,2) -- cycle;
\fill[line width=2pt,color=rvwvcq,fill=rvwvcq,fill opacity=0.10000000149011612] (0,4) -- (2,2) -- (0,2) -- cycle;
\fill[line width=2pt,color=rvwvcq,fill=rvwvcq,fill opacity=0.10000000149011612] (0,4) -- (3,4) -- (2,2) -- cycle;
\fill[line width=2pt,color=rvwvcq,fill=rvwvcq,fill opacity=0.1] (3,4) -- (5,4) -- (2,2) -- cycle;
\fill[line width=2pt,color=rvwvcq,fill=rvwvcq,fill opacity=0.10000000149011612] (5,4) -- (7,4) -- (2,2) -- cycle;
\fill[line width=2pt,color=rvwvcq,fill=rvwvcq,fill opacity=0.10000000149011612] (3,4) -- (1,5) -- (0,4) -- cycle;
\fill[line width=2pt,color=rvwvcq,fill=rvwvcq,fill opacity=0.10000000149011612] (-7.949183386663638,-2.047670801093047) -- (-6.949183386663636,-4.047670801093046) -- (-5.949183386663636,-2.047670801093047) -- cycle;
\fill[line width=2pt,color=rvwvcq,fill=rvwvcq,fill opacity=0.10000000149011612] (-5.949183386663636,-2.047670801093047) -- (-4.949183386663636,-4.047670801093046) -- (-6.949183386663636,-4.047670801093046) -- cycle;
\fill[line width=2pt,color=rvwvcq,fill=rvwvcq,fill opacity=0.10000000149011612] (-5.949183386663636,-2.047670801093047) -- (-3.9491833866636377,-2.047670801093047) -- (-4.949183386663636,-4.047670801093046) -- cycle;
\fill[line width=2pt,color=rvwvcq,fill=rvwvcq,fill opacity=0.10000000149011612] (-3.9699759340211953,-2.0307937074097135) -- (-1.969975934021195,-2.0307937074097135) -- (-4.969975934021196,-4.030793707409716) -- cycle;
\fill[line width=2pt,color=rvwvcq,fill=rvwvcq,fill opacity=0.10000000149011612] (-1.969975934021195,-2.0307937074097135) -- (-1.969975934021195,-4.030793707409716) -- (-4.969975934021196,-4.030793707409716) -- cycle;
\fill[line width=2pt,color=rvwvcq,fill=rvwvcq,fill opacity=0.10000000149011612] (-1.9980725277965015,-1.9975552716811613) -- (0.001927472203498537,-3.9975552716811635) -- (-1.9980725277965015,-3.9975552716811635) -- cycle;
\fill[line width=2pt,color=rvwvcq,fill=rvwvcq,fill opacity=0.1] (-1.969975934021195,-2.0307937074097135) -- (0.030024065978804915,-2.0307937074097135) -- (0.030024065978804915,-4.030793707409716) -- cycle;
\fill[line width=2pt,color=rvwvcq,fill=rvwvcq,fill opacity=0.10000000149011612] (0,-2) -- (2,-4) -- (0,-4) -- cycle;
\fill[line width=2pt,color=rvwvcq,fill=rvwvcq,fill opacity=0.10000000149011612] (0,-2) -- (3,-2) -- (2,-4) -- cycle;
\fill[line width=2pt,color=rvwvcq,fill=rvwvcq,fill opacity=0.10000000149011612] (3.0300240659788047,-2.039839893744594) -- (5.030024065978804,-2.039839893744594) -- (2.0300240659788047,-4.039839893744594) -- cycle;
\fill[line width=2pt,color=rvwvcq,fill=rvwvcq,fill opacity=0.10000000149011612] (5.030024065978804,-2.0190473463870338) -- (7.030024065978804,-2.0190473463870338) -- (2.0300240659788047,-4.019047346387035) -- cycle;
\fill[line width=1.2pt,color=rvwvcq,fill=rvwvcq,fill opacity=0.10000000149011612] (-5,2) -- (-3.130628173328566,0.9350214788215719) -- (-2,2) -- cycle;
\fill[line width=2pt,color=rvwvcq,fill=rvwvcq,fill opacity=0.10000000149011612] (0.004803923010275006,5.096447288921803) -- (1,5) -- (0.013727492602621849,4.013727492602622) -- cycle;
\fill[line width=2pt,color=rvwvcq,fill=rvwvcq,fill opacity=0.10000000149011612] (3.001927472203498,-2.0122371071237564) -- (1.0019274722034985,-1.0122371071237564) -- (0.001927472203498537,-2.0122371071237564) -- cycle;
\fill[line width=2pt,color=rvwvcq,fill=rvwvcq,fill opacity=0.10000000149011612] (0.001927472203498537,-2.0122371071237564) -- (0.0033418807172166665,-0.8951473349809933) -- (1.0019274722034985,-1.0122371071237564) -- cycle;
\fill[line width=2pt,color=rvwvcq,fill=rvwvcq,fill opacity=0.10000000149011612] (-4.949183386663635,-4.047670801093046) -- (-2.8531280310651033,-4.860419041984603) -- (-1.969975934021195,-4.030793707409715) -- cycle;
\draw [line width=1.2pt,color=ffttww] (-8,4)-- (-7,2);
\draw [line width=1.2pt,color=ffttww] (-7,2)-- (-6,4);
\draw [line width=0.8pt,color=rvwvcq] (-6,4)-- (-8,4);
\draw [line width=1.2pt,color=ffttww] (-6,4)-- (-5,2);
\draw [line width=0.8pt,color=rvwvcq] (-5,2)-- (-7,2);
\draw [line width=0.8pt,color=rvwvcq] (-6,4)-- (-4,4);
\draw [line width=1.2pt,color=ffttww] (-4,4)-- (-5,2);
\draw [line width=0.8pt,color=rvwvcq] (-4,4)-- (-2,4);
\draw [line width=1.2pt,dash pattern=on 2pt off 2pt,color=ffttww] (-2,4)-- (-5,2);
\draw [line width=1.2pt,dash pattern=on 2pt off 2pt,color=ffttww] (-2,4)-- (-2,2);
\draw [line width=1.2pt,dash pattern=on 2pt off 2pt,color=ffttww] (-2,4)-- (0,2);
\draw [line width=0.8pt,color=rvwvcq] (0,2)-- (-2,2);
\draw [line width=1.2pt,dash pattern=on 2pt off 2pt] (-2,4)-- (0,4);
\draw [line width=1.2pt,dash pattern=on 2pt off 2pt,color=ffttww] (0,4)-- (0,2);
\draw [line width=1.2pt,dash pattern=on 2pt off 2pt,color=ffttww] (0,4)-- (2,2);
\draw [line width=0.8pt,color=rvwvcq] (2,2)-- (0,2);
\draw [line width=1.2pt,dash pattern=on 2pt off 2pt,color=ffttww] (3,4)-- (2,2);
\draw [line width=0.8pt,color=rvwvcq] (3,4)-- (5,4);
\draw [line width=1.2pt,color=ffttww] (5,4)-- (2,2);
\draw [line width=0.8pt,color=rvwvcq] (5,4)-- (7,4);
\draw [line width=1.2pt,color=ffttww] (7,4)-- (2,2);
\draw [line width=1.2pt,dash pattern=on 2pt off 2pt] (3,4)-- (1,5);
\draw [line width=0.8pt,color=rvwvcq] (0,4)-- (3,4);
\draw [line width=1.2pt,dash pattern=on 2pt off 2pt] (0,4)-- (0.004803923010275006,5.096447288921803);
\draw [line width=1.2pt,color=ffttww] (-7.949183386663638,-2.047670801093047)-- (-6.949183386663636,-4.047670801093046);
\draw [line width=1.2pt,color=ffttww] (-6.949183386663636,-4.047670801093046)-- (-5.949183386663636,-2.047670801093047);
\draw [line width=0.8pt,color=rvwvcq] (-5.949183386663636,-2.047670801093047)-- (-7.949183386663638,-2.047670801093047);
\draw [line width=1.2pt,color=ffttww] (-5.949183386663636,-2.047670801093047)-- (-4.949183386663636,-4.047670801093046);
\draw [line width=0.8pt,color=rvwvcq] (-4.949183386663636,-4.047670801093046)-- (-6.949183386663636,-4.047670801093046);
\draw [line width=0.8pt,color=rvwvcq] (-5.949183386663636,-2.047670801093047)-- (-3.9491833866636377,-2.047670801093047);
\draw [line width=1.2pt,color=ffttww] (-3.9491833866636377,-2.047670801093047)-- (-4.949183386663636,-4.047670801093046);
\draw [line width=0.8pt,color=rvwvcq] (-3.9699759340211953,-2.0307937074097135)-- (-1.969975934021195,-2.0307937074097135);
\draw [line width=1.2pt,color=ffttww] (-1.969975934021195,-2.0307937074097135)-- (-4.969975934021196,-4.030793707409716);
\draw [line width=1.2pt,color=ffttww] (-1.969975934021195,-2.0307937074097135)-- (-1.969975934021195,-4.030793707409716);
\draw [line width=1.2pt,dash pattern=on 2pt off 2pt] (-1.969975934021195,-4.030793707409716)-- (-4.969975934021196,-4.030793707409716);
\draw [line width=1.2pt,color=ffttww] (-1.9980725277965015,-1.9975552716811613)-- (0.001927472203498537,-3.9975552716811635);
\draw [line width=0.8pt,color=rvwvcq] (0.001927472203498537,-3.9975552716811635)-- (-1.9980725277965015,-3.9975552716811635);
\draw [line width=0.8pt,color=rvwvcq] (-1.969975934021195,-2.0307937074097135)-- (0.030024065978804915,-2.0307937074097135);
\draw [line width=1.2pt,color=ffttww] (0.030024065978804915,-2.0307937074097135)-- (0.030024065978804915,-4.030793707409716);
\draw [line width=1.2pt,color=ffttww] (0,-2)-- (2,-4);
\draw [line width=0.8pt,color=rvwvcq] (2,-4)-- (0,-4);
\draw [line width=1.2pt,color=ffttww] (3,-2)-- (2,-4);
\draw [line width=0.8pt,color=rvwvcq] (3.0300240659788047,-2.039839893744594)-- (5.030024065978804,-2.039839893744594);
\draw [line width=1.2pt,color=ffttww] (5.030024065978804,-2.039839893744594)-- (2.0300240659788047,-4.039839893744594);
\draw [line width=0.8pt,color=rvwvcq] (5.030024065978804,-2.0190473463870338)-- (7.030024065978804,-2.0190473463870338);
\draw [line width=1.2pt,color=ffttww] (7.030024065978804,-2.0190473463870338)-- (2.0300240659788047,-4.019047346387035);
\draw [line width=1.2pt,dash pattern=on 2pt off 2pt,color=uququq] (-2.9731352019452615,-4.861522104339344)-- (-1.9980725277965015,-3.9975552716811635);
\draw (1.8266341026449424,5.401004493298012) node[anchor=north west] {$\gamma$};
\draw (1.947687847400871,-0.7035629208222048) node[anchor=north west] {$\gamma'$};
\draw [line width=1.2pt,dash pattern=on 2pt off 2pt] (-5,2)-- (-3.130628173328566,0.9350214788215719);
\draw [line width=1.2pt,dash pattern=on 2pt off 2pt] (-3.130628173328566,0.9350214788215719)-- (-2,2);
\draw [line width=0.8pt,color=rvwvcq] (-2,2)-- (-5,2);
\draw [line width=1.2pt,dash pattern=on 2pt off 2pt] (0.004803923010275006,5.096447288921803)-- (1,5);
\draw [line width=0.8pt,color=rvwvcq] (1,5)-- (0.013727492602621849,4.013727492602622);
\draw [line width=1.2pt,dash pattern=on 2pt off 2pt] (3.001927472203498,-2.0122371071237564)-- (1.0019274722034985,-1.0122371071237564);
\draw [line width=1.2pt,dash pattern=on 2pt off 2pt] (0.001927472203498537,-2.0122371071237564)-- (3.001927472203498,-2.0122371071237564);
\draw [line width=1.2pt,dash pattern=on 2pt off 2pt] (0.001927472203498537,-2.0122371071237564)-- (0.0033418807172166665,-0.8951473349809933);
\draw [line width=1.2pt,dash pattern=on 2pt off 2pt] (0.0033418807172166665,-0.8951473349809933)-- (1.0019274722034985,-1.0122371071237564);
\draw [line width=0.8pt,color=rvwvcq] (1.0019274722034985,-1.0122371071237564)-- (0.001927472203498537,-2.0122371071237564);
\draw [line width=1.2pt,dash pattern=on 2pt off 2pt] (-4.949183386663635,-4.047670801093046)-- (-2.8531280310651033,-4.860419041984603);
\draw (-8.06518618312523,3.4122644008792444) node[anchor=north west] {\textcolor{ffttww}{$\eta$}};
\draw (-8.03059939890925,-2.605836052701026) node[anchor=north west] {\textcolor{ffttww}{$\eta$}};
\draw[color=rvwvcq] (-3.7132637060759694,3.2133903916373674) node {$\sigma_1$};
\draw[color=rvwvcq] (-2.7275403559205504,2.746468804721657) node {$\sigma_2$};
\draw[color=rvwvcq] (-1.2921888109573965,2.6600018441817106) node {$\sigma_3$};
\draw[color=rvwvcq] (-0.44481259766589583,3.2133903916373674) node {$\sigma_4$};
\draw[color=rvwvcq] (0.6446711051374621,2.815642373153614) node {$\sigma_5$};
\draw[color=rvwvcq] (1.75144820004881,3.196096999529378) node {$\sigma_6$};
\draw[color=rvwvcq] (3.1522129607959844,3.196096999529378) node {$\sigma_7$};
\draw[color=rvwvcq] (-3.661383529752,-2.752829885618935) node {$\sigma_1$};
\draw[color=rvwvcq] (-2.692953571704571,-3.2716316488586132) node {$\sigma_2$};
\draw[color=rvwvcq] (-1.2921888109573965,-3.3408052172905705) node {$\sigma_3$};
\draw[color=rvwvcq] (-0.39293242134192636,-2.7355364935109456) node {$\sigma_4$};
\draw[color=rvwvcq] (0.6446711051374621,-3.2370448646426344) node {$\sigma_5$};
\draw[color=rvwvcq] (1.820621768480769,-2.8047100619429024) node {$\sigma_6$};
\draw[color=rvwvcq] (3.1349195686879948,-2.8738836303748596) node {$\sigma_7$};
\end{tikzpicture}
\caption{In this figure, we provide an illustrative example of the cycle modification scheme in  (\textbf{(a.)} $ \Longrightarrow $ \textbf{(b.)})  from \Cref{lem:equieven}. 
As in \Cref{fig:4case,fig:movingit}, the edges of a cycle that intersect the edges of the cocycle $\eta$ are shown as pink-dotted edges and the edges of  a cycle  that do not  intersect $\eta$ are shown as black-dotted edges. 
The cocycle $\eta$ is shown in pink.  The intersection of $\gamma$ and $\eta $  induces a connected graph in the dual graph, namely $C_{\eta }^{\gamma}$. The number of edges in $C_{\eta }^{\gamma}$ is even.
Let $ \gamma' = \gamma + \partial \sigma_2 + \partial \sigma_4 + \partial \sigma_6$. The cycle $\gamma'$ does not intersect  $\eta$.
}
\end{figure}

Using \Cref{not:greatnot}, \Cref{lem:equieven} can be written as follows.

\begin{lemma} \label{lem:betterequieven}
The following are equivalent.
\vspace{-0.25cm}
\begin{enumerate}[(a.)]
\item A connected cocycle $\eta $ is a feasible set for  the input cycle $\zeta$.
\item For  a connected cocycle $\eta$, and any cycle  $\zeta' \in [\zeta]$,  $\eta(\zeta')=1 $.
\item For  a connected cocycle $\eta$, there exists a cycle $\zeta' \in [\zeta]$  such that $\eta(\zeta')=1 $.
\end{enumerate}
\end{lemma}

\begin{lemma}\label{lem:nbdlemma} A connected cocycle $\eta $ is a feasible set if and only if a connected cocycle cohomologous to it is a feasible set. 
\end{lemma}
\begin{proof}
A cocycle $\eta'$ cohomologous to $\eta$ can be written as $\eta' = \eta + \delta(S)$, where $S$ is a collection of vertices. 
Then, by linearity,
\[\eta'(\zeta) = \eta(\zeta) +\delta(S)(\zeta) = \eta(\zeta) +\sum_{v\in S} \delta(v)(\zeta).\]  
 $\delta(v)$ is a connected  trivial cocycle. 
So, using \Cref{lem:notrivial} and $\lnot$(\textbf{(c.)} $ \Longrightarrow $ \textbf{(a.)}) in  \Cref{lem:betterequieven}, $\delta(v)(\zeta)$ is $0$ for every $v\in S$.
 Hence, $\eta'(\zeta) = 1$ if and only if $\eta(\zeta) = 1$.
  So, the claim follows from (\textbf{(c.)} $ \Longrightarrow $ \textbf{(a.)})  in \Cref{lem:betterequieven}.
\end{proof}

Next, we prove an important generalization of \Cref{lem:nbdlemma}.

\begin{lemma}\label{lem:lastlemmaxtra} Let $k>1$ be an integer. Let  $\eta_i$  for $i\in[k]$ be connected cocycles. On the one hand, if    $\eta_i$  for $i\in[k]$    are infeasible sets for the input cycle $\zeta$, then any  cocycle $\vartheta$ cohomologous to $\sum\limits_{i=1}^{k} \eta_i$ is an infeasible set. On the other hand, if   $\eta_k $ is  a feasible set, and  $\eta_i $  for $i\in[k-1]$ are  infeasible sets for the input cycle $\zeta$, then any  cocycle $\vartheta$ cohomologous to $\sum\limits_{i=1}^{k} \eta_i$  is a feasible set. 
\end{lemma}
\begin{proof}
A  cocycle $\vartheta$ cohomologous to $\sum\limits_{i=1}^{k} \eta_i$ can be written as $\sum\limits_{i=1}^{k} \eta_i + \delta(S)$ where $S$ is a collection of vertices. 
Then, by linearity, 
\[
\vartheta(\zeta)=\sum_{i=1}^{k}\eta_{i}(\zeta)+\delta(S)(\zeta)=\sum_{i=1}^{k}\eta_{i}(\zeta)+\sum_{v\in S}\delta(v)(\zeta).
\]
Using \Cref{lem:betterequieven,lem:notrivial}, $\delta(v)(\zeta)=0$
for every $v\in S$, and $\eta_{i}(\zeta)=0$ for $i\in[k-1]$. Hence,
$\vartheta(\zeta)=1$ if and only if $\eta_{k}(\zeta)=1$.
So, the claim follows from \Cref{lem:betterequieven}.
\end{proof}

\begin{remark}[Computing optimal (co)homology basis for surfaces] \label{rem:optico}
For   simplicial complexes with $n$ vertices, $m$ edges and $N$ simplices in total, we recall some of the known results from literature.
For the special case when the input complex is an  surface, Erickson and Whittlesey~\cite{EW} devised a $O(N^{2}\log N+gN^{2}+g^{3}N)$-time algorithm for computing an  optimal homology basis.  Borradaile et al.~\cite{Borradaile} improved on this result by providing a $O({(h+c)}^3 n \log n+m)$-time algorithm for the same problem. Here $c$ denotes the number of boundary components, and $h$ denotes the genus of the surface. Dlotko~\cite{Borradaile}  
generalized the algorithm from~\cite{EW} for computing  an optimal cohomology basis for surfaces. 
For general complexes, Dey et al.~\cite{DeyApprox,DeyLatest},  Chen and Freedman~\cite{ChenFreedman},  Busaryev et al.~\cite{Busaryev}, and Rathod~\cite{Rathodbasis} provided progressively faster algorithms for computing an optimal homology basis.
\end{remark}

Although we expect this to be fairly well known, for the sake of completeness, we describe an algorithm for computing minimum cohomology basis of a triangulated surface that uses the minimum homology basis algorithm as a subroutine.
\begin{lemma}  \label{lem:easyalgo} The minimum cohomology basis problem on  surfaces can be solved in the same time as the minimum homology basis problem on surfaces.
\end{lemma}
\begin{proof} Let $\surface$ be a surface with a weight function $w$ on its edges. 
  Let $\hat{\surface}$ be the dual cell complex of $\surface$. 
Then, to every edge $e$ of $\surface$ there is a unique corresponding edge $\hat{e}$ in $\hat{\surface}$.
We now define a weight function on the edges of $\hat{\surface}$ in the obvious way: $w({\hat{e}}) = w(e)$.
Let $\hat{\surface}'$ be the simplicial complex obtained from the stellar subdivision of each of the $2$-cells of  $\hat{\surface}$.
The weight function on edges of $\hat{\surface}$ is extended to a weight function on edges of $\hat{\surface}'$ by assigning weight $\infty$ to every newly added edge during the stellar subdivision.
Such a complex  $\hat{\surface}'$ can be computed in linear time.
It is easy to check that the cocycles of  $\surface$ are in one-to-one correspondence with the cycles of $\hat{\surface}$, and the cycles of $\hat{\surface}$ are in one-to-one correspondence with finite weight cycles of $\hat{\surface}'$.
Moreover, if $\eta$ is a cocycle of $\surface$, and if $\hat{\eta}$ and $\hat{\eta}'$ are the corresponding cycles in $\hat{\surface}$ and $\hat{\surface}'$, respectively, then $w(\eta) = w(\hat{\eta}) = w(\hat{\eta}')$.
Hence, computing a minimum homology basis for $\hat{\complex}'$ gives a minimum cohomology basis for $\surface$.
\end{proof}

\begin{theorem}  \Cref{alg:surface} provides a polynomial time algorithm for computing  an optimal solution for \hitcycles on surfaces.
\end{theorem} 
\begin{proof} 
Let $\left\{ \nu_{i} \,\,\mid\,\,i\in[m]\right\} $ be an optimal
cohomology basis for $\surface$. Then, by \Cref{thm:charcminsol}, any optimal
solution set is a cocycle. So, we can let $k$ be the smallest
integer for which a cocycle cohomologous to some cocycle in the span
of $\left\{ \nu_{i} \,\,\mid\,\,i\in[k]\right\} $ is a feasible
solution set. Because the algorithm confirms that each $\nu_{i} ,\,i\in[k-1]$
is an infeasible set, by \Cref{lem:lastlemmaxtra}, any \emph{connected} cocycle cohomologous to $\sum\limits_{i=1}^{k-1}\nu_{i} $
is an infeasible set. 
On the other hand, since there exists a feasible set $\theta =\sum\limits_{j_{i}}\nu_{j_{i}} +\nu_{k} +\beta $
where $j_{i}\in[k-1]$, and $\beta $ is a coboundary, by   \Cref{lem:lastlemmaxtra},
$\nu_{k} =\sum\limits_{j_{i}}\nu_{j_{i}} +\beta +\theta $
is also a feasible set. Because $\left\{ \nu_{i} \,\,\mid\,\,i\in[m]\right\} $
is an optimal cohomology basis, $\nu_{k} $ is, in fact, a minimal
solution set. 

From \Cref{lem:easyalgo}, we know that \textsf{Step-1} of \Cref{alg:surfacekiller} can be computed in polynomial time. 
\textsf{Step-2} can be implemented by a simple sorting algorithm. Finally,  \textsf{Step-3} can  be executed in linear time.
\end{proof}

\begin{remark}
The algorithmic results in this section motivate severals questions: To what extent can this machinery  be extended from surfaces to general complexes? 
\begin{enumerate}
\item Are the optimal solutions sets for \hitcycles nontrivial cocycles for  general complexes? To the best of our knowledge, this question is open. 
\item Can the optimal solution sets for \hitcycles be computed efficiently for general complexes? We answer this question in the negative in \Cref{sub:wonetop} by showing that for general complexes \hitcycles is \NP-hard and \Wone-hard.
Intriguingly, for the gadgets used in the reduction the optimal solution sets for \hitcycles are cocycles! So they do not provide (a family of) counterexamples for the first question. 
\item We believe that it should be possible to dualize the hardness results of  Chen and Freedman~\cite{ChenFreedmanLocal} to show that computing an optimal cohomology basis for general complexes is \NP-hard. So, in general, knowing that the optimal solutions sets are cocycles is not enough to guarantee tractability. One also needs an efficient algorithm for computing an optimal cohomology basis. 
\end{enumerate}
\end{remark}

\section{\Wone-hardness results} 

In this section, we obtain \Wone-hardness results for \hitcycles and \createcycle with respect to the solution size $k$ as the parameter via parameterized reductions from \kmulticolorclique.
We begin this section by recalling some common notions from graph theory.

A \emph{$k$-clique} in a graph $G$ is a complete subgraph of $G$ with $k$ vertices. 
Next, a \emph{$k$-coloring} of a graph $G$ is  an assignment of one of $k$ possible colors to every vertex of $G$ (that is, a vertex coloring) such that no two  vertices that share an edge receive the same color.
A graph $G$ equipped with a $k$-coloring is called a \emph{$k$-colored graph}.
Then, a \emph{multicolored $k$-clique} in a colored graph is a $k$-clique with a $k$-coloring. 
\kmulticolorclique asks for the existence of a multicolored $k$-clique in a $k$-colored graph $G$.
We remark that  reducing from \kmulticolorclique is a highly effective tool for showing \Wone-hardness~\cite{multicoloredfellow}. Formally, \kmulticolorclique  is defined as follows:

\begin{problem}[\kmulticolorclique]
{Given a graph $G=(V,E)$, and a vertex coloring $c: V \to [k]$.}
{$k$.}
{Does there exist a multicolored $k$-clique $H$  in $G$? 
}
\end{problem}

\begin{theorem}[Fellows et al.~\cite{multicoloredfellow}]
\label{thm:goodfellow}
 \kmulticolorclique is \Wone-complete.
\end{theorem}

\subsection{\Wone-hardness for \hitcycles} \label{sub:wonetop}

For $i\in[k]$, the subset of vertices of color $i$ is denoted by $V_i$. Clearly, the vertex coloring $c$ induces a partition on $V$: 
\[
V = \bigcup\limits_{i=1}^{k} V_i, 
\quad
\text{and}
\quad
V_i \bigcap V_j = \emptyset$ for all $i,j \in [k].
\]
We now provide a parameterized reduction from \kmulticolorclique to \hitcycles. 
For $r = |V| - 1$, we  define an $(r+1)$-dimensional complex $\complex(G)$ associated to the given colored graph $G$ as follows. 

\paragraph*{Vertices. }
The set of vertices of $\complex(G)$ contains the disjoint union of the vertices $V$ in the graph $G$, the set of colors $[k]$, and an additional dummy vertex $d$.
Altogether, we have $r+k+2$ vertices in $\complex(G)$ so far.
In what follows, further vertices are added to $\complex(G)$.



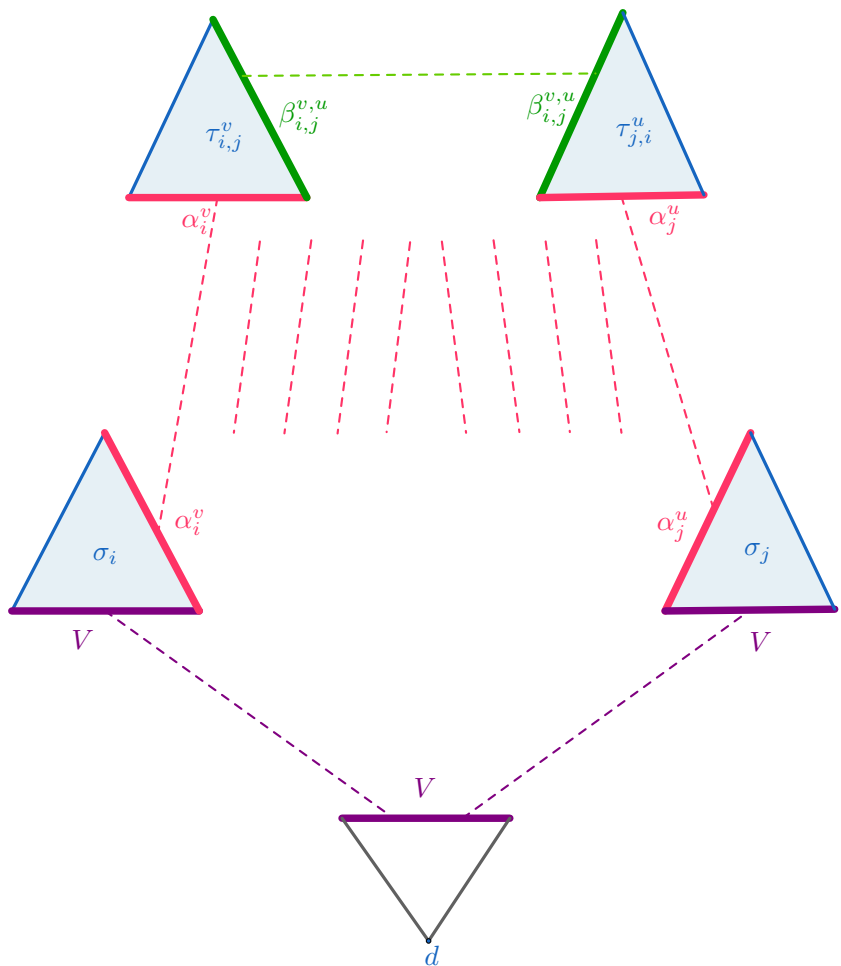
\begin{figure}
\begin{centering}
 \hspace*{-3em}{
\begin{tikzpicture}[scale = 1.7, line cap=round,line join=round,>=triangle 45,x=1cm,y=1cm]
\definecolor{qqwwzz}{rgb}{0,0.4,0.6}
\clip(-7.808458661160528,-3.285963647644749) rectangle (3.7551501623198424,4.437024194751816);
\fill[line width=2pt,color=qqwwzz,fill=qqwwzz,fill opacity=0.10000000149011612] (-5.16123,4.22327) -- (-5.81278077480143,2.834037509015928) -- (-4.434914365065875,2.834037509015928) -- cycle;
\fill[line width=2pt,color=qqwwzz,fill=qqwwzz,fill opacity=0.10000000149011612] (-1.9889344594112857,4.275990728506616) -- (-2.632937014746676,2.834037509015928) -- (-1.3618819546030247,2.8553997789343084) -- cycle;
\fill[line width=2pt,color=qqwwzz,fill=qqwwzz,fill opacity=0.10000000149011612] (-6,1) -- (-6.718539709673136,-0.3899284449513154) -- (-5.26790247164133,-0.3899284449513154) -- cycle;
\fill[line width=2pt,color=qqwwzz,fill=qqwwzz,fill opacity=0.10000000149011612] (-1,1) -- (-1.6620327656765546,-0.3899284449513154) -- (-0.34955145507634916,-0.37611285220815543) -- cycle;
\fill[line width=2pt,color=white] (-4.1626550521885255,-2.006352795901032) -- (-2.863989334331478,-2.006352795901032) -- (-3.49267,-2.96347) -- cycle;
\draw [line width=1.2pt,color=rvwvcq] (-5.16123,4.22327)-- (-5.81278077480143,2.834037509015928);
\draw [line width=2.8pt,color=ffttww] (-5.81278077480143,2.834037509015928)-- (-4.434914365065875,2.834037509015928);
\draw [line width=2.8pt,color=qqzzqq] (-4.434914365065875,2.834037509015928)-- (-5.16123,4.22327);
\draw [line width=2.8pt,color=qqzzqq] (-1.9889344594112857,4.275990728506616)-- (-2.632937014746676,2.834037509015928);
\draw [line width=2.8pt,color=ffttww] (-2.632937014746676,2.834037509015928)-- (-1.3618819546030247,2.8553997789343084);
\draw [line width=1.2pt,color=rvwvcq] (-1.3618819546030247,2.8553997789343084)-- (-1.9889344594112857,4.275990728506616);
\draw [line width=1.2pt,color=rvwvcq] (-6,1)-- (-6.718539709673136,-0.3899284449513154);
\draw [line width=2.8pt,color=yqqqyq] (-6.718539709673136,-0.3899284449513154)-- (-5.26790247164133,-0.3899284449513154);
\draw [line width=2.8pt,color=ffttww] (-5.26790247164133,-0.3899284449513154)-- (-6,1);
\draw [line width=2.8pt,color=ffttww] (-1,1)-- (-1.6620327656765546,-0.3899284449513154);
\draw [line width=2.8pt,color=yqqqyq] (-1.6620327656765546,-0.3899284449513154)-- (-0.34955145507634916,-0.37611285220815543);
\draw [line width=1.2pt,color=rvwvcq] (-0.34955145507634916,-0.37611285220815543)-- (-1,1);
\draw [line width=0.8pt,dash pattern=on 3pt off 3pt,color=ffttww] (-5.126951553792166,2.834037509015928)-- (-5.58708400270019,0.21605569232127153);
\draw [line width=0.8pt,dash pattern=on 3pt off 3pt,color=ffttww] (-2.000272849629603,2.844670520194366)-- (-1.29599,0.43503);
\draw [line width=0.8pt,dash pattern=on 3pt off 3pt,color=ffttww] (-3.6357263755474856,2.4861405234926)-- (-3.817540011267922,1.0048941383584689);
\draw [line width=0.8pt,dash pattern=on 3pt off 3pt,color=ffttww] (-3.3897432213374836,2.4914879833667305)-- (-3.202582125742917,0.9995466784843384);
\draw [line width=0.8pt,dash pattern=on 3pt off 3pt,color=ffttww] (-2.9898083287755632,2.501326870234779)-- (-2.7897863793150868,1.005150214550936);
\draw [line width=0.8pt,dash pattern=on 3pt off 3pt,color=ffttww] (-4,2.5)-- (-4.197209662331185,0.994199218610208);
\draw [line width=0.8pt,dash pattern=on 3pt off 3pt,color=ffttww] (-4.400502138025945,2.5013268702347795)-- (-4.608646805868699,0.9949494773570948);
\draw [line width=0.8pt,dash pattern=on 3pt off 3pt,color=ffttww] (-4.798294378781064,2.491810677843961)-- (-5,1);
\draw [line width=0.8pt,dash pattern=on 3pt off 3pt,color=ffttww] (-2.5865084656033566,2.4922886359156875)-- (-2.3992788082372893,1.005150214550936);
\draw [line width=0.8pt,dash pattern=on 3pt off 3pt,color=ffttww] (-2.196000894525559,2.502987473479463)-- (-2,1);
\draw [line width=2.8pt,color=yqqqyq] (-4.1626550521885255,-2.006352795901032)-- (-2.863989334331478,-2.006352795901032);
\draw [line width=1.2pt,color=wqwqwq] (-2.863989334331478,-2.006352795901032)-- (-3.49267,-2.96347);
\draw [line width=1.2pt,color=wqwqwq] (-3.49267,-2.96347)-- (-4.1626550521885255,-2.006352795901032);
\draw [line width=0.8pt,dash pattern=on 3pt off 3pt,color=yqqqyq] (-5.99648900875117,-0.3899284449513154)-- (-3.762258918880722,-2.006352795901032);
\draw [line width=0.8pt,dash pattern=on 3pt off 3pt,color=yqqqyq] (-3.229189074296303,-2.006352795901032)-- (-1.012699906748032,-0.3899284449513154);
\draw [line width=0.8pt,dash pattern=on 3pt off 3pt,color=wwccqq] (-4.931116281440997,3.7831287734211085)-- (-2.200579392962799,3.802107370754157);
\draw[color=rvwvcq] (-5.070592340032618,3.304667193535562) node {\large $\tau_{i,j}^v$};
\draw[color=rvwvcq] (-1.8979054426156213,3.3464130737647326) node {\large $\tau_{j,i}^u$};
\draw[color=ffttww] (-5.282832907295682,2.647169579926125) node {\large $\alpha_i^v$};
\draw[color=ffttww] (-5.338429395405021,0.2989638170352778) node {\large $\alpha_i^v$};
\draw[color=qqzzqq] (-4.456547675563783,3.461214244394952) node {\large $\beta_{i,j}^{v,u}$};
\draw[color=qqzzqq] (-2.5406018855948444,3.523833064738708) node {\large $\beta_{i,j}^{v,u}$};
\draw[color=ffttww] (-1.661013454098502,2.6576060499834178) node {\large $\alpha_j^u$};
\draw[color=ffttww] (-1.60029575925485,0.26765440686339986) node {\large $\alpha_j^u$};
\draw[color=rvwvcq] (-5.994219940103027,0.03805206560296144) node {\large $\sigma_i$};
\draw[color=rvwvcq] (-0.9429684323733342,0.05892500571754675) node {\large $\sigma_j$};
\draw[color=yqqqyq] (-6.169932862165566,-0.6150808474335493) node {\large $V$};
\draw[color=yqqqyq] (-0.920388423347311,-0.625517317490842) node {\large $V$};
\draw [fill=rvwvcq] (-3.49267,-2.96347) circle (0.5pt);
\draw[color=rvwvcq] (-3.466887117326743,-3.0720160114702497) node {\large $d$};
\draw[color=yqqqyq] (-3.519069467613207,-1.767457254308668) node {\large $V$};
\end{tikzpicture}
}
\end{centering}
\caption{The figure shows some of the attachments in complex $\complex(G)$. In particular,  $\alpha_i^v$ is the common face of $\tau_{i,j}^v$ and $\sigma_i$, $\alpha_j^u$ is the common face of $\tau_{j,i}^u$  and $\sigma_j$, and $\beta_{i,j}^{v,u}$ is the common face of $\tau_{i,j}^v$ and  $\tau_{j,i}^u$. The dashed lines indicate identifications along facets. The set of $r$-simplices supported by the vertices  $V\bigcup \{d\}$ forms a nontrivial $r$-cycle in $\complex(G)$. }
\label{fig:thsattach}
\end{figure}

\paragraph*{Simplices. } Below, we  describe the simplices that constitute the complex $\complex(G)$.
\begin{description}
\vspace*{0.25cm}
\item[The cycle $\zeta$.]

 First, add the $r$-simplex $V$ corresponding to vertex set $V$ of the graph $G$. 
Next, add the $r$-simplices $(V\setminus\left\{ {u}\right\})\bigcup{\{d\}} $ for every $u\in V$. 
The collection of these $r+2$ simplices of dimension $r$ forms a nontrivial $r$-cycle $\zeta$. 
\vspace*{0.25cm}
\item[The simplices in  $\XCC_1$.] $\XCC_1 =  \left\{\sigma_i  \mid i \in [k]   \right \}.$

\begin{itemize}
	\item For every color $i\in[k]$,
	\begin{itemize}
		\item add an $(r+1)$-simplex $\sigma_{i}=V\bigcup{ \{i \}}$. 
	\end{itemize}
\end{itemize}

\begin{definition}[Admissible and undesirable facets of $\sigma_i$]
A facet $(V\setminus\left\{ {v}\right\}) \bigcup{ \{i \}}$ of $\sigma_i$ is said to be \emph{undesirable} if and only if $v \not\in V_i$. All other facets of $\sigma_{i}$ are deemed \emph{admissible}.
In particular, $V$ is admissible.
\end{definition}
\vspace*{0.25cm}
 The idea here is that including  an admissible simplex of the form  $(V\setminus\left\{ {v}\right\}) \bigcup{ \{i \}}$ in $\SCC$ is akin to picking the vertex $v$ of color $i$ for constructing the colorful clique. Including  undesirable simplices in the solution will be made prohibitively expensive as the coloring specified by undesirable simplices is incompatible with the coloring $c$ that the graph $G$ comes equipped with.
\vspace*{0.25cm}
\item[The simplices in  $\XCC_2$.]  $\XCC_2 = \left\{\tau_{i,j}^{v} \mid i \in [k], \; v\in V_i, \; j \in [k]\setminus \{i\} \right\}.$

\begin{itemize}
\item For every color $i\in [k]$, 
\begin{itemize}
	\item for every vertex $v$ in $V_{i}$  and every color $j\in [k]\setminus \{i\}$, 
		\begin{itemize}
			\item add an $(r+1)$-simplex $\tau_{i,j}^{v}=(V\setminus\left\{ {v}\right\})\bigcup \{i,j\}$. 
		\end{itemize}
\end{itemize}

\end{itemize}
\begin{definition}[Admissible and undesirable facets of $\tau_{i,j}^{v}$]
The \emph{admissible} facets of $\tau_{i,j}^{v}$ are:
\begin{compactitem}
\item   $(V\setminus\{v,u\})\bigcup \{i,j\}$ with $u \in V_{j}$ and $\{u,v\}  \in E$, and
\item  $(V\setminus\left\{ {v}\right\})\bigcup \{ i\}$,
\end{compactitem}
A  facet  of $\tau_{i,j}^{v}$ that is not admissible is  \emph{undesirable}. 
\end{definition}

\vspace*{0.25cm}

The intuition here is that picking an admissible facet of the form $(V\setminus\{v,u\})\bigcup \{i,j\}$ is akin to picking the edge $\{u,v\}$ of color $\{i,j\}$ for constructing the colorful clique, whereas the admissible facet $(V\setminus\left\{ {v}\right\})\bigcup \{ i\}$ is common with $\sigma_i$.
Including undesirable simplices in the solution will be made  prohibitively expensive (as explained later). 
Undesirable simplices of $\tau_{i,j}^{v}$  correspond either to  coloring that is incompatible with $c$ or with edges that are not even present in $E$.
\vspace*{0.25cm}
\item[Undesirable and inadmissible simplices.] 

The  undesirability of certain $r$-simplices is implemented in the gadget as follows: 
Let  $m  = n^{3}$. Then, to every undesirable $r$-simplex $ \omega = \{v_1,v_2,\dots, v_{r+1}\}$, associate $m$ new vertices $\UCC^{\omega} =  \{u_1^\omega,u_2^\omega,\dots,u_m^\omega\}$. 
Now introduce $m$ new  $r+1$-simplices  
\[\Upsilon^{\omega}  = \{\mu_i(\omega) = \{v_1,v_2,\dots, v_{r+1}, u_i^\omega \} \mid i\in [m] \}\]
that are cofacets of $\omega$. See \Cref{fig:undesirableths} for an illustrative example. 

\begin{definition}[Set of inadmissible simplices associated to an undesirable simplex $\omega$] \label{def:undesirableinadmissable}
The set of $r$-simplices in $\{ \{ \text{facets of }\mu_i(\omega)  \}  \,\, |\,\, i \in [m] \}$  is denoted by $[\omega]$.
The simplices in the set $[\omega]$ are said to be \emph{inadmissible}.  In particular, $\omega$ itself is inadmissible.
\end{definition}
\vspace*{0.25cm}

Further, note that the set of vertices  in $\UCC^{\omega}$ and $r$-simplices in  $\Upsilon^{\omega} $  are unique to $\omega$.
As we  observe later, introducing these new simplices makes inclusion of $\omega$ in the solution set prohibitively expensive.
Denote by $\YCC$ the set of all $r+1$-simplices added in this step.

\end{description}

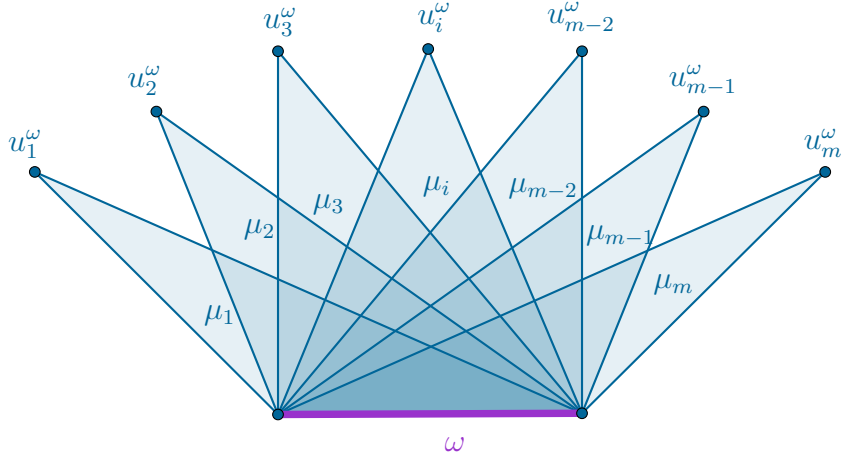
\begin{figure}
\definecolor{ffttww}{rgb}{1,0.2,0.4}
\definecolor{qqwwzz}{rgb}{0,0.4,0.6}
 \hspace*{-10em}{
\begin{tikzpicture}[scale =0.8,line cap=round,line join=round,>=triangle 45,x=1cm,y=1cm]
\clip(-15.427534520845159,-3.76562449915875) rectangle (7.394464297765075,3.908059712160069);
\fill[line width=2pt,color=qqwwzz,fill=qqwwzz,fill opacity=0.1] (-5,-3.02) -- (-9,1) -- (0,-3) -- cycle;
\fill[line width=2pt,color=qqwwzz,fill=qqwwzz,fill opacity=0.1] (-5,-3.02) -- (-7,2) -- (0,-3) -- cycle;
\fill[line width=2pt,color=qqwwzz,fill=qqwwzz,fill opacity=0.1] (-5,-3.02) -- (-5,3) -- (0,-3) -- cycle;
\fill[line width=2pt,color=qqwwzz,fill=qqwwzz,fill opacity=0.1] (-5,-3.02) -- (-2.530096126046699,3.037225251130051) -- (0,-3) -- cycle;
\fill[line width=2pt,color=qqwwzz,fill=qqwwzz,fill opacity=0.1] (-5,-3.02) -- (0,3) -- (0,-3) -- cycle;
\fill[line width=2pt,color=qqwwzz,fill=qqwwzz,fill opacity=0.1] (-5,-3.02) -- (2,2) -- (0,-3) -- cycle;
\fill[line width=2pt,color=qqwwzz,fill=qqwwzz,fill opacity=0.1] (-5,-3.02) -- (4,1) -- (0,-3) -- cycle;
\draw [line width=0.8pt,color=qqwwzz] (-5,-3.02)-- (-9,1);
\draw [line width=0.8pt,color=qqwwzz] (-9,1)-- (0,-3);
\draw [line width=0.8pt,color=qqwwzz] (0,-3)-- (-5,-3.02);
\draw [line width=0.8pt,color=qqwwzz] (-5,-3.02)-- (-7,2);
\draw [line width=0.8pt,color=qqwwzz] (-7,2)-- (0,-3);
\draw [line width=0.8pt,color=qqwwzz] (0,-3)-- (-5,-3.02);
\draw [line width=0.8pt,color=qqwwzz] (-5,-3.02)-- (-5,3);
\draw [line width=0.8pt,color=qqwwzz] (-5,3)-- (0,-3);
\draw [line width=0.8pt,color=qqwwzz] (0,-3)-- (-5,-3.02);
\draw [line width=0.8pt,color=qqwwzz] (-5,-3.02)-- (-2.530096126046699,3.037225251130051);
\draw [line width=0.8pt,color=qqwwzz] (-2.530096126046699,3.037225251130051)-- (0,-3);
\draw [line width=0.8pt,color=qqwwzz] (0,-3)-- (-5,-3.02);
\draw [line width=0.8pt,color=qqwwzz] (-5,-3.02)-- (0,3);
\draw [line width=0.8pt,color=qqwwzz] (0,3)-- (0,-3);
\draw [line width=0.8pt,color=qqwwzz] (0,-3)-- (-5,-3.02);
\draw [line width=0.8pt,color=qqwwzz] (-5,-3.02)-- (2,2);
\draw [line width=0.8pt,color=qqwwzz] (2,2)-- (0,-3);
\draw [line width=0.8pt,color=qqwwzz] (0,-3)-- (-5,-3.02);
\draw [line width=0.8pt,color=qqwwzz] (-5,-3.02)-- (4,1);
\draw [line width=0.8pt,color=qqwwzz] (4,1)-- (0,-3);
\draw [line width=0.8pt,color=qqwwzz] (0,-3)-- (-5,-3.02);
\begin{scriptsize}
\draw [line width=2.8pt,color=DarkOrchid] (-5,-3.02)-- (0,-3);
\draw [fill=qqwwzz] (-5,-3.02) circle (2.5pt);
\draw [fill=qqwwzz] (0,-3) circle (2.5pt);
\draw[color=DarkOrchid] (-2.0863163059859424,-3.5172507555617347) node {\Large $\omega$};
\draw [fill=qqwwzz] (-9,1) circle (2.5pt);
\draw[color=qqwwzz] (-9.155624408104944,1.4596859685630539) node {\Large $u_1^{\omega}$};
\draw[color=qqwwzz] (-5.957516779294719,-1.3788710984282702) node {\Large $\mu_1$};
\draw [fill=qqwwzz] (-7,2) circle (2.5pt);
\draw[color=qqwwzz] (-7.195719026570323,2.5194139402398146) node {\Large $u_2^\omega$};
\draw[color=qqwwzz] (-5.300568228097033,0.0462276921108151) node {\Large $\mu_2$};

\draw [fill=qqwwzz] (-5,3) circle (2.5pt);
\draw[color=qqwwzz] (-4.9627208005371495,3.484523343016865) node {\Large $u_3^\omega$};
\draw[color=qqwwzz] (-4.165145401300505,0.4567291382261194) node {\Large $\mu_3$};
\draw [fill=qqwwzz] (-2.530096126046699,3.037225251130051) circle (2.5pt);
\draw[color=qqwwzz] (-2.4269431540249014,3.5602181981366336) node {\Large $u_i^\omega$};
\draw[color=qqwwzz] (-2.3890957264650168,0.759508558705194) node {\Large $\mu_i$};
\draw [fill=qqwwzz] (0,3) circle (2.5pt);
\draw[color=qqwwzz] (-0.014169647082276804,3.5602181981366336) node {\Large $u_{m-2}^\omega$};
\draw[color=qqwwzz] (-0.6386522018203679,0.7216611311453096) node {\Large $\mu_{m-2}$};
\draw [fill=qqwwzz] (2,2) circle (2.5pt);
\draw[color=qqwwzz] (1.9917440135915911,2.5194139402398146) node {\Large $u_{m-1}^\omega$};
\draw[color=qqwwzz] (0.6370840489956406,-0.05421113383231896) node {\Large $\mu_{m-1}$};
\draw [fill=qqwwzz] (4,1) circle (2.5pt);
\draw[color=qqwwzz] (3.9692721035955456,1.4786096823429962) node {\Large $u_m^\omega$};
\draw[color=qqwwzz] (1.5091893122030664,-0.8679308263698318) node {\Large $\mu_m$};

\end{scriptsize}
\end{tikzpicture}
}
\caption{ For every undesirable simplex $ \omega = \{v_1,v_2,\dots, v_{r+1}\}$ $m$ new vertices $\UCC^{\omega} =  \{u_1^\omega,u_2^\omega,\dots,u_m^\omega\}$ are added to $\complex(G)$.
Moreover, $m$ new  $r+1$-simplices $\Upsilon^{\omega}  = \{\mu_i(\omega) = \{v_1,v_2,\dots, v_{r+1}, u_i^\omega \} \mid i\in [m] \}$, where  $\omega \prec \mu_i(\omega) $ for every $i \in [m]$ are also added to $\complex(G)$.
The facets of $\mu_i(\omega)$ for every $i\in[m]$ are the inadmissible simplices associated to $\omega$ and denoted by $[\omega]$.}
\label{fig:undesirableths}
\end{figure}

This completes the construction of complex $\complex(G)$. 
It is easy to check that the inadmissible and admissible  simplices of $\complex(G)$ partition the set of  $r$-simplices of $\complex(G)$.

\begin{notation}
The admissible facets $(V\setminus\left\{ {v}\right\}) \bigcup{\{i\}}$ and $(V\setminus\{v,u\})\bigcup \{i,j\}$   are denoted by $\alpha_i^v$ and  $\beta_{i,j}^{v,u}$, respectively.
For every vertex $v\in V$ of color $i$, there is a facet $\alpha_i^v$. For every edge $\{u,v\} \in E$, there is  a facet $\beta_{i,j}^{v,u}$, where $i$ is the color of $v$ and $j$ is the color of $u$.
\end{notation}

\begin{remark}[Meaning of superscripts and subscripts of simplices]
A simple mnemonic for remembering the meaning of the notation for simplices is as follows: the indices in the subscript are the included colors, and the vertices in the superscript indicate the  vertices excluded from $V$.
For instance, $\beta_{i,j}^{v,u}$ is the full simplex on the vertex set $(V\setminus\{v,u\})\bigcup \{i,j\}$. In this case, colors $i$ and $j$ are included and vertices $u$ and $v$ are excluded. The same notational rule applies for $\alpha_i^v$ , $\sigma_i$ and $\tau_{i,j}^{v}$.
\end{remark}

\begin{remark}[Correspondence between colors and vertices in $\alpha_i^v$ $\beta_{i,j}^{v,u}$ and $\tau_{i,j}^{v}$]
In our notation, the first color corresponds to the first vertex, the second color to the second vertex, and so on. For instance,
\begin{compactitem}
\item In $\alpha_i^v$, vertex $v$ is of color $i$.
\item In $\beta_{i,j}^{v,u}$, $v$ is of color $i$ and $u$ is of color $j$.
\item In $\tau_{i,j}^{v}$, $v$ is of color $i$ and the vertex associated to color $j$ is not specified. It is, in fact, chosen through a facet $\beta_{i,j}^{v,u} \prec \tau_{i,j}^{v}$.
\end{compactitem}
\end{remark}

\paragraph*{Choice of parameter. } Let $(k + \binom{k}{2} + 1 = \binom{k+1}{2} + 1)$ be the parameter  for \hitcycles on the complex $\complex(G)$.

\begin{remark}[Size of $\complex(G)$]
We note that every subset of vertices of $G$ is a simplex in $\complex(G)$. However, $\complex(G)$ is represented implicitly, and the simplices of dimensions other than $r$ and $r+1$ are not
used in the reduction. Thus, although $\complex(G)$ as a simplicial complex is exponential in the size of $G$, the reduction itself is polynomial in the size of $G$ because the number of $r$ and $r+1$ dimensional simplices of $\complex(G)$ are polynomial in size of $G$, even after inadmissible simplices.
\end{remark}

\begin{lemma}\label{lem:forwardclique} If there exists a multicolored $k$-clique $H=(V_H,E_H)$ of $G$, then there exists a topological hitting set $\SCC$ for  $\zeta$  consisting of  $\binom{k+1}{2} + 1$ $r$-simplices. 
\end{lemma}
\begin{proof} 
We construct a  set $\SCC$ of $r$-simplices that mimics the graphical structure of $H$ as follows:
\[\SCC_{\alpha}=\left\{ \alpha_{i}^{v}\,\,|\,\,v\in V_{i}\medcap V_{H}\right\} \]
\[\SCC_{\beta}=\left\{ \beta_{i,j}^{v,u}\,\,|\,\,v\in V_{i},\,u\in V_{j},\,\left\{ i,j\right\} \in E_{H}\right\} \]
First, set $\SCC = \left\{ V\right\} \bigcup \SCC_{\alpha} \bigcup \SCC_{\beta}$.
Next, note that every cycle $\zeta' \in [\zeta]$ can be expressed as 
 \[\zeta' = \zeta + \sum\limits_{\nu_i \in \XCC'} \partial \nu_i +   \sum\limits_{\mu_j \in \YCC'} \partial \mu_j \]
for some $\XCC' \subset \XCC_1 \bigcup \XCC_2$ and $\YCC' \subset \YCC$.
Let $\XCC_1' = \XCC' \medcap \XCC_1 $, and $\XCC_2' = \XCC' \medcap \XCC_2 $.
Now, we claim that removing $\SCC$ from $\complex(G)$ destroys every cycle $\zeta' \in [\zeta]$.
We show this by establishing that the coefficient in  every $ \zeta' \in [\zeta]$ of at least one of the simplicies of  $\SCC$ is $1$. 
In other words, $\SCC \bigcap \zeta\ \neq \emptyset$ for every $ \zeta' \in [\zeta]$.

\begin{description}
\item [{Case~1:} $\XCC' =\emptyset$.  ] Then, $V\in\SCC$ has coefficient $1$ in cycle $\zeta'$. This is because simplices in $\YCC$ are not incident on $V$, and $V \in \zeta$. 

\item [{Case~2:} $\XCC_1' \neq \emptyset, \XCC_2' =\emptyset $.  ] Then, the  cycle $\zeta'$ can be written as 
\[\zeta'=\zeta+\sum\limits_{\sigma_j \in \XCC_1'}\partial\sigma_{j} +   \sum\limits_{\mu_\ell \in \YCC'} \partial \mu_\ell\]
Then, every $\alpha_{j}^{v}\in\SCC$ for $\sigma_j \in \XCC_1'$ and $v\in V_H \medcap V_j$ has coefficient $1$ in cycle $\zeta'$.
This is because $\alpha_{j}^{v} \in \partial \sigma_j $ for every $\sigma_j \in \XCC_1'$, but  $\alpha_{j}^{v}  \not \in \zeta$ and $\alpha_{j}^{v}  \not \in \partial \mu_\ell$ for any $\mu_\ell \in \YCC'$.


\item [{Case~3:} $\XCC_1' = \emptyset, \XCC_2' \neq \emptyset $.  ] This case is identical to Case-1, because  $V\in\SCC$ has coefficient $1$ in cycle $\zeta'$. 

\item [{Case~4:} $ \XCC_1' \neq \emptyset, \,\,  \XCC_2' \neq \emptyset$.  ]If every simplex $\tau_{p,q}^{v}\in\XCC_{2}'$ is such that $v\in V_{p}\setminus V_{H}$,
then this case becomes identical to \textsf{Case 2}. So we will assume without
loss of generality that the set $\XCC_{2}''=\left\{ \tau_{p,q}^{v}\,\,|\,\,p,q\in[k],v\in V_{p}\bigcap V_{H},\tau_{p,q}^{v}\in\XCC_{2}'\right\} $
is non-empty. For some $\left\{ u,v\right\} \in E_{H}$ and $u\in V_{q}$, $v\in V_{p}$,
  if  $\tau_{p,q}^{v}\in\XCC_{2}''$ and $\tau_{q,p}^{u}\not\in\XCC_{2}''$,
 then the coefficient of $\ensuremath{\beta_{p,q}^{v,u}}\in\SCC$ in $\zeta'$ is $1$ because
the only two $(r+1)$-simplices incident on $\ensuremath{\beta_{p,q}^{v,u}}$
are $\tau_{p,q}^{v}$ and $\tau_{q,p}^{u}$. So, without loss of generality
assume that the symmetric simplex $\tau_{q,p}^{u}$ is also in $\XCC_{2}''$.
In other words, $|\XCC_{2}''|$ is even. Note that for every $\tau_{p,q}^{v}\in\XCC_{2}''$,
exactly one facet of $\tau_{p,q}^{v}$ lies in $\SCC_{\alpha}$, namely
$\alpha_{p}^{v}$. Hence the cardinality of the multiset $\TCC=\left\{ \partial\tau_{i,j}^{v}\bigcap\SCC_{\alpha}\,\,|\,\,\tau_{i,j}^{v}\in\XCC_{2}''\right\} $
is even. Let $\sigma_i \in\ECC$ if and only if the cardinality
of the set $\left\{ \tau_{i,j}^{v}\in\XCC_{2}''\,\,|\,\,\alpha_{i}^{v}\in\partial\tau_{i,j}^{v} \bigcap \SCC_{\alpha} \right\} $
is even, and $\sigma_i \in\OCC$ if and only if the cardinality
of the set $\left\{ \tau_{i,j}^{v}\in\XCC_{2}''\,\,|\,\,\alpha_{i}^{v}\in\partial\tau_{i,j}^{v} \bigcap  \SCC_{\alpha} \right\} $
is odd. It is easy to check that $ \XCC_{1}' \subseteq \OCC \bigcup \ECC $.
Note that since $|\TCC|=|\OCC|+|\ECC|$, $|\OCC|$ must be even. 

Now, if $\sigma_i \in\ECC \bigcap\XCC_{1}'$, then the coefficient
of $\alpha_{i}^{v}\in\SCC$ in $\zeta'$  is $1$ because the
only $(r+1)$-simplices incident on $\alpha_{i}^{v}$ are  $\left\{ \tau_{i,j}^{v}\in\XCC_{2}''\,\,|\,\,\alpha_{i}^{v}\in\partial\tau_{i,j}^{v}\right\} \bigcup \{\sigma_{i}\}$, and $|\left\{ \tau_{i,j}^{v}\in\XCC_{2}''\,\,|\,\,\alpha_{i}^{v}\in\partial\tau_{i,j}^{v}\right\}|$ is even when $\sigma_i \in \ECC$. So, without loss of generality assume that $\ECC \bigcap\XCC_{1}'$
is empty. That is, we assume that $\XCC_{1}' \subseteq \OCC $. But if, $\sigma \in \OCC \setminus \XCC_{1}'$, then the coefficient of $\alpha_{i}^{v}\in\SCC_{\alpha}$ in $\zeta'$ is $1$ because  in that case the only $(r+1)$-simplices incident on $\alpha_{i}^{v}$ will be $\left\{ \tau_{i,j}^{v}\in\XCC_{2}''\,\,|\,\,\alpha_{i}^{v}\in\partial\tau_{i,j}^{v}\right\} $ which has odd cardinality. So, we assume that $\OCC = \XCC_{1}'$.
But if $\OCC = \XCC_{1}'$, then $V\in\SCC$ has coefficient $1$ in $\zeta'$ because  $\OCC$ is even and $V\in\zeta$. 
This completes the proof. Please see \Cref{fig:thsattachtwo} for the final part of the argument. \qedhere
\end{description}
\end{proof}

\begin{figure}
\begin{centering}
 \hspace*{-3em}{
\begin{tikzpicture}[scale = 1.7, line cap=round,line join=round,>=triangle 45,x=1cm,y=1cm]
\clip(-7.808458661160528,-3.285963647644749) rectangle (3.7551501623198424,4.437024194751816);
\fill[line width=2pt,color=rvwvcq,fill=rvwvcq,fill opacity=0.10000000149011612] (-5.16123,4.22327) -- (-5.81278077480143,2.834037509015928) -- (-4.434914365065875,2.834037509015928) -- cycle;
\fill[line width=2pt,color=rvwvcq,fill=rvwvcq,fill opacity=0.10000000149011612] (-1.9889344594112857,4.275990728506616) -- (-2.632937014746676,2.834037509015928) -- (-1.3618819546030247,2.8553997789343084) -- cycle;
\fill[line width=2pt,color=rvwvcq,fill=rvwvcq,fill opacity=0.10000000149011612] (-6,1) -- (-6.718539709673136,-0.3899284449513154) -- (-5.26790247164133,-0.3899284449513154) -- cycle;
\fill[line width=2pt,color=rvwvcq,fill=rvwvcq,fill opacity=0.10000000149011612] (-1,1) -- (-1.6620327656765546,-0.3899284449513154) -- (-0.34955145507634916,-0.37611285220815543) -- cycle;
\fill[line width=2pt,color=ffwwzz,fill=ffwwzz,fill opacity=0.1] (-7.032043404877571,2.275696165857004) -- (-7.037473159202634,1.5822388414569286) -- (-0.029920464383715706,1.5686054704553352) -- (-0.03204340487757129,2.275696165857004) -- cycle;

\fill[line width=2pt,color=yqqqyq,fill=yqqqyq,fill opacity=0.1] (-7.032043404877571, -1.03) -- (-7.037473159202634,-1.62) -- (-0.029920464383715706,-1.62) -- (-0.03204340487757129,-1.03) -- cycle;

\fill[line width=2pt,color=white] (-4.1626550521885255,-2.006352795901032) -- (-2.863989334331478,-2.006352795901032) -- (-3.49267,-2.96347) -- cycle;
\draw [line width=1.2pt,color=rvwvcq] (-5.16123,4.22327)-- (-5.81278077480143,2.834037509015928);
\draw [line width=2.8pt,color=ffttww] (-5.81278077480143,2.834037509015928)-- (-4.434914365065875,2.834037509015928);
\draw [line width=2.8pt,color=qqzzqq] (-4.434914365065875,2.834037509015928)-- (-5.16123,4.22327);
\draw [line width=2.8pt,color=qqzzqq] (-1.9889344594112857,4.275990728506616)-- (-2.632937014746676,2.834037509015928);
\draw [line width=2.8pt,color=ffttww] (-2.632937014746676,2.834037509015928)-- (-1.3618819546030247,2.8553997789343084);
\draw [line width=1.2pt,color=rvwvcq] (-1.3618819546030247,2.8553997789343084)-- (-1.9889344594112857,4.275990728506616);
\draw [line width=1.2pt,color=rvwvcq] (-6,1)-- (-6.718539709673136,-0.3899284449513154);
\draw [line width=2.8pt,color=yqqqyq] (-6.718539709673136,-0.3899284449513154)-- (-5.26790247164133,-0.3899284449513154);
\draw [line width=2.8pt,color=ffttww] (-5.26790247164133,-0.3899284449513154)-- (-6,1);
\draw [line width=2.8pt,color=ffttww] (-1,1)-- (-1.6620327656765546,-0.3899284449513154);
\draw [line width=2.8pt,color=yqqqyq] (-1.6620327656765546,-0.3899284449513154)-- (-0.34955145507634916,-0.37611285220815543);
\draw [line width=1.2pt,color=rvwvcq] (-0.34955145507634916,-0.37611285220815543)-- (-1,1);
\draw [line width=0.8pt,dash pattern=on 3pt off 3pt,color=ffttww] (-5.126951553792166,2.834037509015928)-- (-5.58708400270019,0.21605569232127153);
\draw [line width=0.8pt,dash pattern=on 3pt off 3pt,color=ffttww] (-2.000272849629603,2.844670520194366)-- (-1.29599,0.43503);
\draw [line width=0.8pt,dash pattern=on 3pt off 3pt,color=ffttww] (-3.6357263755474856,2.4861405234926)-- (-3.817540011267922,1.0048941383584689);
\draw [line width=0.8pt,dash pattern=on 3pt off 3pt,color=ffttww] (-3.3897432213374836,2.4914879833667305)-- (-3.202582125742917,0.9995466784843384);
\draw [line width=0.8pt,dash pattern=on 3pt off 3pt,color=ffttww] (-2.9898083287755632,2.501326870234779)-- (-2.7897863793150868,1.005150214550936);
\draw [line width=0.8pt,dash pattern=on 3pt off 3pt,color=ffttww] (-4,2.5)-- (-4.197209662331185,0.994199218610208);
\draw [line width=0.8pt,dash pattern=on 3pt off 3pt,color=ffttww] (-4.400502138025945,2.5013268702347795)-- (-4.608646805868699,0.9949494773570948);
\draw [line width=0.8pt,dash pattern=on 3pt off 3pt,color=ffttww] (-4.798294378781064,2.491810677843961)-- (-5,1);
\draw [line width=0.8pt,dash pattern=on 3pt off 3pt,color=ffttww] (-2.5865084656033566,2.4922886359156875)-- (-2.3992788082372893,1.005150214550936);
\draw [line width=0.8pt,dash pattern=on 3pt off 3pt,color=ffttww] (-2.196000894525559,2.502987473479463)-- (-2,1);
\draw [line width=2pt,color=ffwwzz] (-7.032043404877571,2.275696165857004)-- (-7.037473159202634,1.5822388414569286);
\draw [line width=2pt,color=ffwwzz] (-7.037473159202634,1.5822388414569286)-- (-0.029920464383715706,1.5686054704553352);
\draw [line width=2pt,color=ffwwzz] (-0.029920464383715706,1.5686054704553352)-- (-0.03204340487757129,2.275696165857004);
\draw [line width=2pt,color=ffwwzz] (-0.03204340487757129,2.275696165857004)-- (-7.032043404877571,2.275696165857004);

\draw [line width=2pt,color=yqqqyq] (-7.032043404877571, -1.03) -- (-7.037473159202634,-1.62);
\draw [line width=2pt,color=yqqqyq] (-7.037473159202634,-1.62) -- (-0.029920464383715706,-1.62);
\draw [line width=2pt,color=yqqqyq] (-0.029920464383715706,-1.62)-- (-0.03204340487757129,-1.03);
\draw [line width=2pt,color=yqqqyq] (-0.03204340487757129,-1.03) -- (-7.032043404877571, -1.03);

\draw [line width=2.8pt,color=yqqqyq] (-4.1626550521885255,-2.006352795901032)-- (-2.863989334331478,-2.006352795901032);
\draw [line width=1.2pt,color=wqwqwq] (-2.863989334331478,-2.006352795901032)-- (-3.49267,-2.96347);
\draw [line width=1.2pt,color=wqwqwq] (-3.49267,-2.96347)-- (-4.1626550521885255,-2.006352795901032);
\draw [line width=0.8pt,dash pattern=on 3pt off 3pt,color=yqqqyq] (-5.99648900875117,-0.3899284449513154)-- (-3.762258918880722,-2.006352795901032);
\draw [line width=0.8pt,dash pattern=on 3pt off 3pt,color=yqqqyq] (-3.229189074296303,-2.006352795901032)-- (-1.012699906748032,-0.3899284449513154);
\draw [line width=0.8pt,dash pattern=on 3pt off 3pt,color=wwccqq] (-4.931116281440997,3.7831287734211085)-- (-2.200579392962799,3.802107370754157);
\draw[color=rvwvcq] (-5.070592340032618,3.304667193535562) node {\large $\tau_{i,j}^v$};
\draw[color=rvwvcq] (-1.8979054426156213,3.3464130737647326) node {\large $\tau_{j,i}^u$};
\draw[color=ffttww] (-5.282832907295682,2.647169579926125) node {\large $\alpha_i^v$};
\draw[color=ffttww] (-5.338429395405021,0.2989638170352778) node {\large $\alpha_i^v$};
\draw[color=qqzzqq] (-4.456547675563783,3.461214244394952) node {\large $\beta_{i,j}^{v,u}$};
\draw[color=qqzzqq] (-2.5406018855948444,3.523833064738708) node {\large $\beta_{i,j}^{v,u}$};
\draw[color=ffttww] (-1.661013454098502,2.6576060499834178) node {\large $\alpha_j^u$};
\draw[color=ffttww] (-1.60029575925485,0.26765440686339986) node {\large $\alpha_j^u$};
\draw[color=rvwvcq] (-5.994219940103027,0.03805206560296144) node {\large $\sigma_i$};
\draw[color=rvwvcq] (-0.9429684323733342,0.05892500571754675) node {\large $\sigma_j$};
\draw[color=yqqqyq] (-6.169932862165566,-0.6150808474335493) node {\large $V$};
\draw[color=yqqqyq] (-0.920388423347311,-0.625517317490842) node {\large $V$};
\draw[color=ffwwzz] (0.47809856432991793,1.8853072657437613) node {\Large $|\mathcal{T}|$};
\draw [fill=rvwvcq] (-3.49267,-2.96347) circle (0.5pt);
\draw[color=rvwvcq] (-3.466887117326743,-3.0720160114702497) node {\large $d$};
\draw[color=yqqqyq] (-3.519069467613207,-1.867457254308668) node {\large $V$};
\draw[color=yqqqyq] (0.47809856432991793, -1.3) node {\Large $|\mathcal{\XCC}_1'|$};

\end{tikzpicture}
}
\end{centering}
\caption{
Like in \Cref{fig:thsattachtwo}, the  dashed lines indicate identifications along facets.
Additionally, in this figure, $\tau_{i,j}^v$ and $\tau_{j,i}^u$ belong to $\XCC_{2}''$,
$\sigma_i$ and $\sigma_j$ belong to $\XCC_1'$,
and $\alpha_i^v$ and $\alpha_j^u$ belong to $\SCC_{\alpha}$. 
$\TCC$ accounts for all the incidences of  the boundaries of simplices in $\XCC_{2}''$ on simplices in $\SCC_{\alpha}$. The final part of the argument in \textsf{Case 4} of \Cref{lem:forwardclique} is depicted here.
 In \textsf{Case 4} of \Cref{lem:forwardclique}, we use the fact that $|\TCC|$ is even.
 Also, if $|\TCC|$ is even and if the coefficient of all the simplices of $\SCC \setminus \{V\}$ have coefficient $0$ in some cycle $\zeta' \in [\zeta]$, then $\XCC' = \OCC$, and $\OCC$ has even cardinality. But if this is so, then the coefficient of $V$ in $\zeta'$ is $(1 + |\OCC|) \text{ mod }2 = 1$.}
\label{fig:thsattachtwo}
\end{figure}

The next few lemmas provide a method to extract a multi-colored $k$-clique from $G$ given a solution set $\RCC$ for \hitcycles on $\complex(G)$. 

\begin{lemma}  If there exists a cycle $\zeta'\in[\zeta]$ such that only the inadmissible simplices of $\RCC$ have coefficient $1$ in $ \zeta'$, then the size of $\RCC$ is at least $m$. \label{lem:biginad}
\end{lemma}
\begin{proof} 
We consider two cases.

\begin{description}
\item[Case 1:] A cycle $\zeta'$ has a unique inadmissible simplex $ \nu$ with coefficient $1$ in $\RCC$.

Suppose that $ \nu = \{v_1,v_2,\dots, v_{r+1}\}$  is, in fact, an undesirable simplex. Assume that  $\nu$ is the unique simplex in  $\RCC$  with coefficient $1$ in  $\zeta'$. 
Let $\mu_i = \{v_1,v_2,\dots, v_{r+1}, u_i^\nu \}$, $i\in [m]$ be the inadmissible simplices in $[\nu]$.
Then, a simplex in $\RCC$ will have coefficient $1$ in the  cycle $\zeta'_i = \zeta' + \partial \mu_i$ only if one of the simplices in  $\partial \mu_i \setminus  \{ \nu \}$ for every $i \in [m]$ belongs to $\RCC$. Since  $\partial \mu_i \setminus  \{ \nu \}$ for $i\in [m]$ are disjoint sets, the size of $\RCC$ is at least $m$.

Next, suppose  $ \nu = \{v_1,v_2,\dots, v_{r+1}\} \in [\omega]$, where $\omega$ is a undesirable simplex, and $\nu \neq \omega$. 
Then, $\omega$ is a facet of  $ \mu_j(\omega)$ for some $\mu_j(\omega) = \{v_1,v_2,\dots, v_{r+1}, u_j^\nu \}$, $ j \in [m]$. 
 Since $\zeta'\in[\zeta]$ is a \emph{cycle},  all simplices in $\partial \mu_j(\omega)$ must have coefficient $1$ in $\zeta'$.
Then, a simplex in $\RCC$ will have coefficient $1$ in each of  the  cycles $\zeta'_i = \zeta' + \partial \mu_i(\omega) + \partial \mu_j(\omega)$, $i \in [m]\setminus \{j\}$ only if one of the simplices in  $\partial \mu_i(\omega) \setminus \{\nu \}$ for each $i$ belongs to $\RCC$,  where $\mu_i(\omega) = \{v_1,v_2,\dots, v_{r+1}, u_i^\nu \}$. 
Also, the cycle $\zeta'' = \zeta + \partial \mu_j(\omega)$  has a simplex in $\RCC$  with coefficient $1$ only if $\omega \in \RCC$.
Hence, in both cases, the size of $\RCC$ is at least $m$.

\item[Case 2:] A cycle $\zeta'$ has multiple inadmissible simplices with coefficient $1$ in $\RCC$.

More generally, suppose there exist more than one  inadmissible simplices  in $\RCC$  with coefficient $1$ in $ \zeta'$, for some cycle $\zeta' \in [\zeta]$.
For an undesirable simplex $\omega$, we say that  \emph{$[\omega]$ belongs to $\zeta'$} if it there exists a simplex in $[\omega]$ that has coefficient $1$ in $\zeta'$.
Let $J$ be an indexing set for the undesirable simplices of $\complex(G)$ whose classes belong to $\zeta'$. That is, for all $j\in J$,  $[\omega^j]$ belongs to $ \zeta'$. 
Define the sets $P$ and $Q$ as follows.
\[ P = \{ \mu_i(\omega^j) \,\, | \,\, j\in J, \,\, i\in [m],  \,\,  \text{ a simplex in } \partial \mu_i(\omega^j)\setminus \{\omega^j\} \text{ belongs to }\zeta' \text{ and }\RCC\}.  \]
and
\[ Q = \{ \mu_k(\omega^j) \,\, | \,\, j\in J, \,\, k\in [m],  \,\, \mu_k(\omega^j) \not\in P  \}.  \]
Define $\zeta_{\overline{Q}}$ as follows.
\[ \zeta_{\overline{Q}} = \zeta' + \sum_{\substack{\mu_i(\omega^j)\in P, \\ j\in J}} \partial \mu_i(\omega^j) +  \sum_{\substack{\mu_k(\omega^j)\in \overline{Q}, \\ j\in J}} \partial \mu_k(\omega^j).\]
for all  $\overline{Q} \subset Q$.

Then, a simplex in $\RCC$ will have coefficient $1$ in each of  the   cycles $\zeta_{\overline{Q}}$ if and only if one of the $r$-simplices in  $\partial \mu_k(\omega^j)$ for  every $ \mu_k(\omega^j) \in Q$ belongs to $\RCC$.
Clearly, $|P  + Q|  \geq m$, proving the claim. \qedhere
\end{description}
\end{proof}
 
\begin{lemma} \label{lem:admissiblew}
Let $\RCC$ be a solution set for \hitcycles on complex $\complex(G)$ such that $|\RCC|  \leq \binom{k+1}{2} + 1$.
Then,
\begin{enumerate}[(1.)]
\item $V \in \RCC$.
 \item For every $\sigma_{i}$, there is  at least one simplex  $\alpha_i^v $ (with $v \in V_i$)  that is included in $\RCC$. 
\item For every unordered pair $(i,j)$, where $i,j \in [k]$, there exists a simplex $ \beta_{i,j}^{v,u}$ for some $v,u$ that is included in $\RCC$.
\item  $|\RCC|  =  |A_\RCC| =\binom{k+1}{2} + 1$, where $A_\RCC$ denotes the set of admissible simplices of  $\RCC$.
\end{enumerate}
\end{lemma}
\begin{proof}
If $\zeta'\in[\zeta]$ is such that $\zeta' \bigcap A_\RCC = \emptyset$, then  we are forced to include some  simplex $\omega \in \zeta'$ in $\RCC$ such that $\omega$ is inadmissible. 
In that case, \Cref{lem:biginad} applies and we are forced to include at least $m$ simplices. But, if we include more than  $n^3$ facets in $\RCC$, we exceed the budget of  $\binom{k+1}{2} + 1$. So, going forward, we assume that  at least one simplex in $A_\RCC$ has coefficient $1$ for every $\zeta'\in[\zeta]$.

Note that if at least one simplex in $A_\RCC$ has coefficient $1$ for every cycle in $[\zeta]$, then we do not need inadmissible simplices in $\RCC$.
We now prove the four statements of the lemma.
\begin{enumerate}[(1.)]

\item Since  $V$ is the only admissible facet of $\zeta$, it must be included.

\item Since $\zeta' = \zeta + \partial \sigma_{i}$ is a cycle homologous to $\zeta$, and the coefficient of $V$ in $\zeta'$ is zero, the  admissible simplices in $\zeta'$ are given by the set $\{\alpha_i^v \mid v\in V_i\}$. One of the simplices in this set must be included in $\RCC$ for each $i$, for $\RCC$ to be a solution set. 

\item Note that for a fixed $i$ and $j \in [k] \setminus \{i\}$ unless some admissible facet $ \beta_{i,j}^{v,u}$ for some $v,u$ is included in $\RCC$, the coefficient of all admissible simplices in $\zeta'=\zeta+\partial\sigma_{i}+ \sum\limits_{v\in V_i} \partial\tau_{i,j}^{v}$ will be zero. The claim follows.

\item This follows from the first three parts of the lemma. By (1.) we must include $V$ in $\RCC$, by (2.)  we must include at least $k$ $\alpha$ faces in $\RCC$, and by (3.), we must include at least $\binom{k}{2}$ faces n $\RCC$. Since $\binom{k}{2} + k + 1 = \binom{k+1}{2} + 1$, the claim follows.\qedhere
\end{enumerate}

\end{proof}

\begin{lemma} \label{lem:mainreverse}
If $|\RCC|   =\binom{k+1}{2} + 1$, then one can obtain a $k$-clique $H$ of $G$ from $\RCC$.
\end{lemma}
\begin{proof}  
If $|\RCC|   =\binom{k+1}{2} + 1$, then using  \label{lem:admissiblew}~(4.), $|\RCC|  =  |A_\RCC|$. Therefore, $\RCC$ consists entirely of admissible simplices. 
We now provide four conditions that characterize a solution of size $\binom{k+1}{2} + 1$.

As noted in \Cref{lem:admissiblew}~(1.), $V$ is part of any solution set. Using \Cref{lem:admissiblew}~(2.),
 for $\RCC$ to be a solution set, at least one facet
(other than $V$) of $\sigma_{i}$ for every $i$ must belong to $\RCC$.

\begin{description}
\item[Condition 1.]  \label{condone}  For every $i \in [k]$, the \emph{only}  facet of $\sigma_{i}$ (other than $V$) that belongs $\RCC$ is an admissible simplex $ \alpha_i^v$, for some $v \in V_i$.

\NoIndent{Now, $ \alpha_i^v$
is incident on $k-1$ simplices, namely, $\tau_{i,j}^{v}$ for all
$j\neq i$. Using \Cref{lem:admissiblew}~(3.), for $\RCC$ to be a solution, we must include in $\RCC$ at least
one  admissible facet   $ \beta_{i,j}^{v,u}$ of $\tau_{i,j}^{v}$ (for all $j\neq i$).} 

\item[Condition 2.] For every  $v \in V_i$  such that $ \alpha_i^v\in \RCC$, and every  $ j\in [k] \setminus \{i\}$,   the \emph{only}  facet of $\tau_{i,j}^{v}$ (other than $ \alpha_i^v$) that belongs $\RCC$ is  an admissible simplex $ \beta_{i,j}^{v,u}$, for some $u \in V_j$. 

\NoIndent{Note that $ \beta_{i,j}^{v,u}$ is also incident on both $\tau_{j,i}^{u}$ and $\tau_{j,i}^{u}$.
Then, since there must exist an admissible simplex in $\RCC$ with coefficient $1$ in $\zeta'=\zeta+\partial\sigma_{i}+\partial\tau_{i,j}^{v} + \partial \tau_{j,i}^{u}$, it is necessary that at least one  facet of $\tau_{i,j}^{u}$ other than $ \beta_{i,j}^{v,u}$ is included in $\RCC$. That is, $\alpha_{j}^{u}$ must be included in $\RCC$. 

Repeating the same argument for every $\sigma_{i}$ and every $\tau_{i,j}^{v}$, it is easy to check that the only way to construct $\RCC$ without exceeding the budget of $\binom{k+1}{2}+1$ is by making these choices consistent. Thus, we obtain two additional conditions.
}

\item[Condition 3.]  for every  $i$ and $v$  such that $ \alpha_i^v\in \RCC$, and every  $ j\in [k] \setminus \{i\}$,  if $ \alpha_i^v$ is  in $\RCC$, and  $ \beta_{i,j}^{v,u}$ is in $\RCC$, then  $ \alpha_j^u$ is in  $\RCC$.

\item[Condition 4.] for every  $i$ and $v$  if $ \alpha_i^v \not \in \RCC$ (from choices made for $\sigma_i$'s in \textsf{Condition 1.}), then  for every $ j\in [k] \setminus \{i\}$, the facet $\beta_{i,j}^{u,v} $ of $\tau_{i,j}^{v}$ in not included in $\RCC$.

\end{description}


The fact that such a set $\RCC$ is indeed a solution set follows the same argument as in Claim~\ref{lem:forwardclique}. 
It is clear that such a solution set $\RCC$  satisfies \textsf{Conditions 1-4}  if and only if $|\RCC| = \binom{k+1}{2}+1$. 
If any of the conditions are not satisfied, then either we are forced to choose more than one vertices per color, or we have that the choice of vertices $u,v$ in  $ \beta_{i,j}^{v,u}$ that is included in $\RCC$ as per \Cref{lem:admissiblew}~(3.) for  pairs $(i,j)$ and $(j,i)$  is inconsistent. In both cases, $|\RCC| \geq \binom{k+1}{2}+1$. 


Finally, the graph $H$ is constructed from $\RCC$  by first including one
vertex $v$ per color $i$  for every
$  \alpha_i^v  \in \RCC$, and the edges $\left\{ u,v\right\} $ for every simplex $ \beta_{i,j}^{v,u} \in \RCC$.
\end{proof}

\Cref{lem:forwardclique} and \Cref{lem:mainreverse} together provide a parameterized reduction from \kmulticolorclique to \hitcycles. 
Using~\Cref{thm:goodfellow}, we obtain the following result.

\begin{theorem} \hitcycles is $\Wone$-hard.
\end{theorem}

\subsection{\Wone-hardness for \createcycle} \label{sub:createhard}

Next, we provide a parameterized reduction from \kmulticolorclique to \createcycle. This reduction bears some  similarities  with reduction  from \kmulticolorclique to \hitcycles. 
So towards the end, we skip some of the details that are common to both the reductions.

\paragraph*{Recalling some notation from \Cref{sub:wonetop}}
In \Cref{sub:wonetop}, given a $k$-colored graph $G=(V,E)$, the $(r+1)$-dimensional complex $\complex(G)$ was built out of types of $(r+1)$-simplices, namely, $\XCC_1$ and $\XCC_2$:
\[ \XCC_1 = \left\{\sigma_i  \mid i \in [k]   \right \},\]
where $\sigma_{i} = V\bigcup{\{i\}}$, and
\[\XCC_2 = \left\{\tau_{i,j}^{v} \mid i \in [k], \; v\in V_i, \; j \in [k]\setminus \{i\} \right\},\]
where $\tau_{i,j}^v = (V \setminus{v}) \bigcup{\{i, j \}} $.

\noindent{Furthermore, recall from \Cref{sub:wonetop}, that for every $i \in [k]$, if $v\in V_i$, then }
\[\alpha_i^v = (V\setminus\left\{ {v}\right\}) \bigcup{\{i\}}\]
 is an \emph{admissible} facet of $\sigma_i$  and $\tau_{i,j}^v$ respectively.
 
 \noindent{Also, for every $i \in [k]$, and $j\in [k] \setminus \{i\}$,  if $v\in V_i$,  $u \in V_j$, and $\{u,v\} \in E$, then }
 \[\beta_{i,j}^{v,u}  = (V\setminus\{v,u\})\bigcup \{i,j\}\] 
 is an \emph{admissible} facet of $\tau_{i,j}^v$ and $\tau_{j,i}^u$ respectively.

\paragraph*{Overview of the reduction for \createcycle}
In this section, given a $k$-colored graph $G=(V,E)$, an $r$-dimensional complex $\altcomplex(G)$ is constructed, where $r = |V|  - 1$. Here, we provide an overview of the construction. 

Let $\sigma_{i} = V\bigcup{\{i\}}$ as in \Cref{sub:wonetop}, and let $ \hat{\partial}\sigma_{i}$ be the $r$-complex  $ \partial \sigma_{i} \setminus \{V\}$.
The complex $ \hat{\Delta}\sigma_{i}$ is formed from $ \hat{\partial}\sigma_{i}$ by  the so-called S-subdivision of some of the faces of $ \hat{\partial}\sigma_{i}$.
The  S-subdivision  of a simplex is described in \Cref{sec:subdivideit}. The construction of  $ \hat{\Delta}\sigma_{i}$ from $ \hat{\partial}\sigma_{i}$ is described in \Cref{alg:deltasigma}.
The lexicographically highest simplex of an S-subdivided face of $ \hat{\partial}\sigma_{i}$ is a \emph{distinguished simplex} in   $ \hat{\Delta}\sigma_{i}$. The simplices in $ \hat{\partial}\sigma_{i}$ 
that are not distinguished are called \emph{undesirable}. We wish to exclude undesirable simplices from solutions of small size. To implement the undesirability of simplices,
we add further simplices to $ \hat{\Delta}\sigma_{i}$. The  newly added simplices  and the undesirable simplices are together called \emph{inadmissible simplices} of $ \hat{\Delta}\sigma_{i}$. 
That completes the high-level description of $ \hat{\Delta}\sigma_{i}$.
Next, a subcomplex $\ZCC_1$  is built out of the union of subcomplexes $ \hat{\Delta} \sigma_i $. That is,
\[\ZCC_1 =  \bigcup_{i \in [k]}  \hat{\Delta}\sigma_i. \] 

Let $T = |V| \cdot (k-1)$. It is easy to check that $|\XCC_2|  = T$. 
 Let $t \in [T]$ be an indexing variable such that there is a unique  $t$ that corresponds to a triple $(i,j,v)$, where $i \in [k]$, $j \in [k]\setminus \{i\}$ and $ v\in V_i$.
Now, for every $t \in [T]$, add a \emph{new} set of vertices $V_t $. Here the vertex set $V_t $ is in  one-to-one correspondence with the vertex set $(V\setminus\left\{ {v}\right\})\bigcup\{i,j\}$.
Let $\tau_{i,j}^{v}$ be the full simplex on the vertex set $V_t $.
The complex $ {\Delta}\tau_{i,j}^{v}$ is formed from $  {\partial} \tau_{i,j}^{v}$  following an S-subdivision of some of the faces of $  {\partial} \tau_{i,j}^{v}$.
The construction of  $  {\Delta}\tau_{i,j}^{v}$ from $ {\partial} \tau_{i,j}^{v}$ is described in \Cref{alg:deltatau}. 
The \emph{distinguished, undesirable,} and \emph{inadmissible} simplices of $  {\Delta}\tau_{i,j}^{v}$ are built in a manner analogous to the distinguished, undesirable and inadmissible simplices of $ \hat{\Delta}\sigma_{i}$.
For further details, please refer to \Cref{sec:descgadget}.
Next, a subcomplex $\ZCC_2$  is built out of the union of subcomplexes $ {\Delta}\tau_{i,j}^{v}$. That is,

\[\ZCC_2  =  \bigcup_{\substack{i\in[k]\\ v\in V_{i}}} \bigcup_{j \in [k]\setminus \{i\}} \Delta\tau _{i,j}^v. \]

The complex $\altcomplex(G)$ is obtained from $G$ by identifying  the distinguished faces of $\ZCC_1 \bigcup \ZCC_2$ as per the procedure described in  \Cref{alg:attachments}. The distinguished faces upon identifications are called the \emph{admissible simplices} of  $\altcomplex(G)$. 
In what should remind the reader of the notation used in \Cref{sub:wonetop}, for every $i \in [k]$, and $v\in V_i$, there is an admissible simplex denoted by $\alpha_i^v$ that belongs to $\altcomplex(G)$.
 And for every $i \in [k]$, $j\in [k] \setminus \{i\}$  with $v\in V_i$,  $u \in V_j$ and $\{u,v\} \in E$, there is an admissible simplex $\beta_{i,j}^{v,u}$ that belongs to $\altcomplex(G)$.
 The admissible simplices of $\altcomplex(G)$ encode the connectivity and coloring information of $G$.
Analogous to the construction of the complex $\complex(G)$ described in \Cref{sub:wonetop}, in complex $\altcomplex(G)$,  $\alpha_i^v$ is an admissible $r$-simplex belonging to the simplicial manifolds $ \hat{\Delta} \sigma_i$  and ${\Delta} \tau_{i,j}^v$, 
where $ \hat{\Delta} \sigma_i$  and ${\Delta} \tau_{i,j}^v$ are subcomplexes of $\altcomplex(G)$.
 Also, $\beta_{i,j}^{v,u}$ of  is an admissible $r$-simplex belonging to the simplicial manifolds ${\Delta}\tau_{i,j}^v$ and ${\Delta}\tau_{j,i}^u$.
 We ask the reader to compare \Cref{fig:mainfigbnt,fig:thsattach}.

Finally, \Cref{prop:forwardcliquetwo,lem:mainreversetwo} combine to show that the $k$-multicolored cliques of $G$ are in one-to-one correspondence with $\binom{k+1}{2}$-sized solutions for \createcycle with $\altcomplex(G)$ as the instance.
In fact, the  reduction is a parameterized reduction that establishes the \Wone-hardness of \createcycle as a consequence.

\begin{remark}
Note that we use the notation $\hat{\partial}\sigma_{i}$ and $ \hat{\Delta}\sigma_{i}$ for the complexes associated to $\sigma_i$, and $  {\partial} \tau_{i,j}^{v}$ and $ {\Delta}\tau_{i,j}^{v}$ for the complexes associated to $\tau_{i,j}^{v}$.
This disparity in notation (that is the use of \; $\hat{}$ \; for $\sigma_i$) is to remind the reader that in the case of $\hat{\partial}\sigma_{i}$, a face is deleted from simplex boundary of  $\sigma_i$, whereas in the case of $  {\partial} \tau_{i,j}^{v}$, the full simplex boundary is used.
\end{remark}

\subsubsection{S-subdivisions of simplices} \label{sec:subdivideit}

Next, we recall a  lemma from Munkres~\cite[Lemma 3.2]{MR755006} that will be used to provide a guarantee that the complex described in \Cref{sec:descgadget}  is, in fact, a simplicial complex. 

\begin{lemma}[Munkres, \protect{\cite[Lemma 3.2]{MR755006}}]
\label{lem:pasting}
Let $\LCC$ be a finite set of labels.  Let $\complex$ be a simplicial complex defined on a set of vertices $V$. Also, let $f: V \to \LCC$ be a surjective map associating to each vertex of $\complex$ a label from $\LCC$. The labeling $f$ extends to a simplicial map $g: \complex \to  \complex_f$ where $\complex_f$ has vertex set $V$ and is obtained from $\complex$ by identifying vertices with the same label. 
 
  If for all pairs $v,w \in V$, $f(v) = f(w)$ implies that their stars $\str_{\complex} (v)$ and $\str_{\complex} (w)$ are vertex disjoint, then, for all faces $\face,\smallface \in \complex$ we have that 
\begin{compactitem}
\item
$\face$ and $g(\face)$ have the same dimension, and 
\item
$g(\face) = g(\smallface)$ implies that either $\face = \smallface$ or $\face$ and $\smallface$ are vertex disjoint in $\complex$. 
\end{compactitem}
\end{lemma} 
\Cref{lem:pasting} provides a way of gluing faces of a simplicial complex by a simplicial quotient map obtained from vertex identifications. In particular, \Cref{lem:pasting}  provides conditions under which the gluing does not create unwanted identifications, and the resulting complex thus obtained is also a simplicial complex. Now, we describe a special kind of subdivision, which we call an \emph{S-subdivision} of a $d$-simplex, with a later application of \Cref{lem:pasting}  in mind. 

Let $\nu$ be an $d$-simplex, and let $\VCC$ be the vertex
set of $\nu$ equipped with an ordering $\succ_{\nu}$. We construct a complex $\altaltcomplex_{\nu}$ obtained from a subdivision of $\nu$
such that an $r$-simplex $\Omega\in \altaltcomplex_{\nu}$  has the following property: for every
vertex $v\in\Omega$, $\left(\str_{\altaltcomplex_{\nu}}v\right)\bigcap \VCC=\emptyset$.
The construction of the complex $\altaltcomplex_{\nu}$ is described in \Cref{alg:ssubdivision}.

 \begin{algorithm}[H]
\caption{ S-subdivision of simplex $\nu$ }\label{alg:ssubdivision}
\begin{algorithmic}[1]

\Procedure{S-subdivide}{$\nu,\succ_{\nu}$}

\State{Let $\VCC$ denote the vertices of $\nu$;}
\State{Let $\altaltcomplex_{0}\gets\left\{ \nu\right\} $; \quad $\Omega_{1}\gets\nu$; \quad $\VCC_{0}\gets \VCC$;}

\For{$i=1$ to $2(d+1)$}
\State{Perform a stellar subdivision of $\Omega_{i}$ to obtain $\altaltcomplex_{i}$ from $\altaltcomplex_{i-1}$;}
\State{Let $v_{i}$ be the new vertex introduced during the stellar subdivision;}
\State{$\VCC_{i}\gets \VCC_{i-1}\bigcup\left\{ v_{i}\right\} ;$ }
\State{Extend $\succ_{\nu}$ as follows: Set $v_{i}\succ_{\nu} v$ for all $v\in \VCC_{i-1}$; }
\State{Let $\Omega_{i+1}$ be the lexicographically highest $d$-simplex of $\altaltcomplex_{i}$;}
\EndFor
\State{$\altaltcomplex_{\nu}\gets \altaltcomplex_{i};$}
\State{\textbf{return }$\altaltcomplex_{\nu},\, \Omega_{2(d+1)+1}$;}

\EndProcedure
\end{algorithmic}
\label{alg:s-subdivide}
\end{algorithm}

Please refer to  \Cref{fig:subdivision} for an illustrative
example. In \Cref{fig:subdivision}, $\VCC=\left\{ A,B,C\right\} $
and $\Omega=\left\{ G,H,I\right\} $, and the stars of $G$, $H$ and
$I$  do not intersect $\VCC$.

\definecolor{ffttww}{rgb}{1,0.2,0.4}
\definecolor{rvwvcq}{rgb}{0.08235294117647059,0.396078431372549,0.7529411764705882}
\begin{figure}
\begin{tikzpicture}[scale = 0.9, line cap=round,line join=round,>=triangle 45,x=1cm,y=1cm]
\clip(-8.097418778800077,-4.968733469922368) rectangle (10.757267831146908,6.930202704633709);
\fill[line width=1pt,color=rvwvcq,fill=rvwvcq,fill opacity=0.10000000149011612] (-1,6) -- (-7.98392035652087,-4.016079643479129) -- (7.01607964347913,-4.016079643479129) -- cycle;
\fill[line width=1pt,color=rvwvcq,fill=rvwvcq,fill opacity=0.10000000149011612] (-1,6) -- (-0.988911730357444,3.3283625653086304) -- (-7.98392035652087,-4.016079643479129) -- cycle;
\fill[line width=1pt,color=rvwvcq,fill=rvwvcq,fill opacity=0.10000000149011612] (-1,6) -- (-0.988911730357444,3.3283625653086304) -- (7.01607964347913,-4.016079643479129) -- cycle;
\fill[line width=1pt,color=rvwvcq,fill=rvwvcq,fill opacity=0.10000000149011612] (-0.988911730357444,3.3283625653086304) -- (-7.98392035652087,-4.016079643479129) -- (7.01607964347913,-4.016079643479129) -- cycle;
\fill[line width=1pt,color=rvwvcq,fill=rvwvcq,fill opacity=0.10000000149011612] (-0.988911730357444,3.3283625653086304) -- (-0.988911730357444,-2.219114434991162) -- (-7.98392035652087,-4.016079643479129) -- cycle;
\fill[line width=1pt,color=rvwvcq,fill=rvwvcq,fill opacity=0.10000000149011612] (-0.988911730357444,3.3283625653086304) -- (-0.988911730357444,-2.219114434991162) -- (7.01607964347913,-4.016079643479129) -- cycle;
\fill[line width=1pt,color=rvwvcq,fill=rvwvcq,fill opacity=0.10000000149011612] (-0.988911730357444,-2.219114434991162) -- (-7.98392035652087,-4.016079643479129) -- (7.01607964347913,-4.016079643479129) -- cycle;
\fill[line width=1pt,color=rvwvcq,fill=rvwvcq,fill opacity=0.10000000149011612] (-0.988911730357444,3.3283625653086304) -- (-0.988911730357444,-2.219114434991162) -- (3.7063441655484732,-2.315592295865941) -- cycle;
\fill[line width=1pt,color=rvwvcq,fill=rvwvcq,fill opacity=0.10000000149011612] (-0.988911730357444,-2.219114434991162) -- (3.7063441655484732,-2.315592295865941) -- (7.01607964347913,-4.016079643479129) -- cycle;
\fill[line width=1pt,color=rvwvcq,fill=rvwvcq,fill opacity=0.10000000149011612] (-0.988911730357444,3.3283625653086304) -- (3.7063441655484732,-2.315592295865941) -- (7.01607964347913,-4.016079643479129) -- cycle;
\fill[line width=1pt,color=rvwvcq,fill=rvwvcq,fill opacity=0.10000000149011612] (-0.988911730357444,3.3283625653086304) -- (-0.988911730357444,-2.219114434991162) -- (0.18490224361903554,-0.19307935662080328) -- cycle;
\fill[line width=1pt,color=rvwvcq,fill=rvwvcq,fill opacity=0.10000000149011612] (-0.988911730357444,3.3283625653086304) -- (0.18490224361903554,-0.19307935662080328) -- (3.7063441655484732,-2.315592295865941) -- cycle;
\fill[line width=1pt,color=rvwvcq,fill=rvwvcq,fill opacity=0.10000000149011612] (0.18490224361903554,-0.19307935662080328) -- (-0.988911730357444,-2.219114434991162) -- (3.7063441655484732,-2.315592295865941) -- cycle;
\fill[line width=1pt,color=rvwvcq,fill=rvwvcq,fill opacity=0.10000000149011612] (0.18490224361903554,-0.19307935662080328) -- (0.5064951132016323,-1.4472915479929305) -- (-0.988911730357444,-2.219114434991162) -- cycle;
\fill[line width=1pt,color=rvwvcq,fill=rvwvcq,fill opacity=0.10000000149011612] (0.18490224361903554,-0.19307935662080328) -- (0.5064951132016323,-1.4472915479929305) -- (3.7063441655484732,-2.315592295865941) -- cycle;
\fill[line width=1pt,color=rvwvcq,fill=rvwvcq,fill opacity=0.10000000149011612] (0.5064951132016323,-1.4472915479929305) -- (-0.988911730357444,-2.219114434991162) -- (3.7063441655484732,-2.315592295865941) -- cycle;
\fill[line width=1pt,color=rvwvcq,fill=rvwvcq,fill opacity=0.10000000149011612] (0.9084862001798789,-1.1417783218894635) -- (3.7063441655484732,-2.315592295865941) -- (0.5064951132016323,-1.4472915479929305) -- cycle;
\fill[line width=1pt,color=rvwvcq,fill=rvwvcq,fill opacity=0.10000000149011612] (0.18490224361903554,-0.19307935662080328) -- (0.9084862001798789,-1.1417783218894635) -- (3.7063441655484732,-2.315592295865941) -- cycle;
\fill[line width=2pt,color=ffttww,fill=ffttww,fill opacity=0.81] (0.18490224361903554,-0.19307935662080328) -- (0.9084862001798789,-1.1417783218894635) -- (0.5064951132016323,-1.4472915479929305) -- cycle;
\draw [line width=1pt,color=rvwvcq] (-1,6)-- (-7.98392035652087,-4.016079643479129);
\draw [line width=1pt,color=rvwvcq] (-7.98392035652087,-4.016079643479129)-- (7.01607964347913,-4.016079643479129);
\draw [line width=1pt,color=rvwvcq] (7.01607964347913,-4.016079643479129)-- (-1,6);
\draw [line width=1pt,color=rvwvcq] (-1,6)-- (-0.988911730357444,3.3283625653086304);
\draw [line width=1pt,color=rvwvcq] (-0.988911730357444,3.3283625653086304)-- (-7.98392035652087,-4.016079643479129);
\draw [line width=1pt,color=rvwvcq] (-7.98392035652087,-4.016079643479129)-- (-1,6);
\draw [line width=1pt,color=rvwvcq] (-1,6)-- (-0.988911730357444,3.3283625653086304);
\draw [line width=1pt,color=rvwvcq] (-0.988911730357444,3.3283625653086304)-- (7.01607964347913,-4.016079643479129);
\draw [line width=1pt,color=rvwvcq] (7.01607964347913,-4.016079643479129)-- (-1,6);
\draw [line width=1pt,color=rvwvcq] (-0.988911730357444,3.3283625653086304)-- (-7.98392035652087,-4.016079643479129);
\draw [line width=1pt,color=rvwvcq] (-7.98392035652087,-4.016079643479129)-- (7.01607964347913,-4.016079643479129);
\draw [line width=1pt,color=rvwvcq] (7.01607964347913,-4.016079643479129)-- (-0.988911730357444,3.3283625653086304);
\draw [line width=1pt,color=rvwvcq] (-0.988911730357444,3.3283625653086304)-- (-0.988911730357444,-2.219114434991162);
\draw [line width=1pt,color=rvwvcq] (-0.988911730357444,-2.219114434991162)-- (-7.98392035652087,-4.016079643479129);
\draw [line width=1pt,color=rvwvcq] (-7.98392035652087,-4.016079643479129)-- (-0.988911730357444,3.3283625653086304);
\draw [line width=1pt,color=rvwvcq] (-0.988911730357444,3.3283625653086304)-- (-0.988911730357444,-2.219114434991162);
\draw [line width=1pt,color=rvwvcq] (-0.988911730357444,-2.219114434991162)-- (7.01607964347913,-4.016079643479129);
\draw [line width=1pt,color=rvwvcq] (7.01607964347913,-4.016079643479129)-- (-0.988911730357444,3.3283625653086304);
\draw [line width=1pt,color=rvwvcq] (-0.988911730357444,-2.219114434991162)-- (-7.98392035652087,-4.016079643479129);
\draw [line width=1pt,color=rvwvcq] (-7.98392035652087,-4.016079643479129)-- (7.01607964347913,-4.016079643479129);
\draw [line width=1pt,color=rvwvcq] (7.01607964347913,-4.016079643479129)-- (-0.988911730357444,-2.219114434991162);
\draw [line width=1pt,color=rvwvcq] (-0.988911730357444,3.3283625653086304)-- (-0.988911730357444,-2.219114434991162);
\draw [line width=1pt,color=rvwvcq] (-0.988911730357444,-2.219114434991162)-- (3.7063441655484732,-2.315592295865941);
\draw [line width=1pt,color=rvwvcq] (3.7063441655484732,-2.315592295865941)-- (-0.988911730357444,3.3283625653086304);
\draw [line width=1pt,color=rvwvcq] (-0.988911730357444,-2.219114434991162)-- (3.7063441655484732,-2.315592295865941);
\draw [line width=1pt,color=rvwvcq] (3.7063441655484732,-2.315592295865941)-- (7.01607964347913,-4.016079643479129);
\draw [line width=1pt,color=rvwvcq] (7.01607964347913,-4.016079643479129)-- (-0.988911730357444,-2.219114434991162);
\draw [line width=1pt,color=rvwvcq] (-0.988911730357444,3.3283625653086304)-- (3.7063441655484732,-2.315592295865941);
\draw [line width=1pt,color=rvwvcq] (3.7063441655484732,-2.315592295865941)-- (7.01607964347913,-4.016079643479129);
\draw [line width=1pt,color=rvwvcq] (7.01607964347913,-4.016079643479129)-- (-0.988911730357444,3.3283625653086304);
\draw [line width=1pt,color=rvwvcq] (-0.988911730357444,3.3283625653086304)-- (-0.988911730357444,-2.219114434991162);
\draw [line width=1pt,color=rvwvcq] (-0.988911730357444,-2.219114434991162)-- (0.18490224361903554,-0.19307935662080328);
\draw [line width=1pt,color=rvwvcq] (0.18490224361903554,-0.19307935662080328)-- (-0.988911730357444,3.3283625653086304);
\draw [line width=1pt,color=rvwvcq] (-0.988911730357444,3.3283625653086304)-- (0.18490224361903554,-0.19307935662080328);
\draw [line width=1pt,color=rvwvcq] (0.18490224361903554,-0.19307935662080328)-- (3.7063441655484732,-2.315592295865941);
\draw [line width=1pt,color=rvwvcq] (3.7063441655484732,-2.315592295865941)-- (-0.988911730357444,3.3283625653086304);
\draw [line width=1pt,color=rvwvcq] (0.18490224361903554,-0.19307935662080328)-- (-0.988911730357444,-2.219114434991162);
\draw [line width=1pt,color=rvwvcq] (-0.988911730357444,-2.219114434991162)-- (3.7063441655484732,-2.315592295865941);
\draw [line width=1pt,color=rvwvcq] (3.7063441655484732,-2.315592295865941)-- (0.18490224361903554,-0.19307935662080328);
\draw [line width=1pt,color=rvwvcq] (0.18490224361903554,-0.19307935662080328)-- (0.5064951132016323,-1.4472915479929305);
\draw [line width=1pt,color=rvwvcq] (0.5064951132016323,-1.4472915479929305)-- (-0.988911730357444,-2.219114434991162);
\draw [line width=1pt,color=rvwvcq] (-0.988911730357444,-2.219114434991162)-- (0.18490224361903554,-0.19307935662080328);
\draw [line width=1pt,color=rvwvcq] (0.18490224361903554,-0.19307935662080328)-- (0.5064951132016323,-1.4472915479929305);
\draw [line width=1pt,color=rvwvcq] (0.5064951132016323,-1.4472915479929305)-- (3.7063441655484732,-2.315592295865941);
\draw [line width=1pt,color=rvwvcq] (3.7063441655484732,-2.315592295865941)-- (0.18490224361903554,-0.19307935662080328);
\draw [line width=1pt,color=rvwvcq] (0.5064951132016323,-1.4472915479929305)-- (-0.988911730357444,-2.219114434991162);
\draw [line width=1pt,color=rvwvcq] (-0.988911730357444,-2.219114434991162)-- (3.7063441655484732,-2.315592295865941);
\draw [line width=1pt,color=rvwvcq] (3.7063441655484732,-2.315592295865941)-- (0.5064951132016323,-1.4472915479929305);
\draw [line width=1pt,color=rvwvcq] (0.9084862001798789,-1.1417783218894635)-- (3.7063441655484732,-2.315592295865941);
\draw [line width=1pt,color=rvwvcq] (3.7063441655484732,-2.315592295865941)-- (0.5064951132016323,-1.4472915479929305);
\draw [line width=1pt,color=rvwvcq] (0.5064951132016323,-1.4472915479929305)-- (0.9084862001798789,-1.1417783218894635);
\draw [line width=1pt,color=rvwvcq] (0.18490224361903554,-0.19307935662080328)-- (0.9084862001798789,-1.1417783218894635);
\draw [line width=1pt,color=rvwvcq] (0.9084862001798789,-1.1417783218894635)-- (3.7063441655484732,-2.315592295865941);
\draw [line width=1pt,color=rvwvcq] (3.7063441655484732,-2.315592295865941)-- (0.18490224361903554,-0.19307935662080328);
\draw [line width=1pt,color=rvwvcq] (0.18490224361903554,-0.19307935662080328)-- (0.9084862001798789,-1.1417783218894635);
\draw [line width=1pt,color=rvwvcq] (0.9084862001798789,-1.1417783218894635)-- (0.5064951132016323,-1.4472915479929305);
\draw [line width=1pt,color=rvwvcq] (0.5064951132016323,-1.4472915479929305)-- (0.18490224361903554,-0.19307935662080328);
\begin{scriptsize}
\draw [fill=rvwvcq] (-1,6) circle (2.5pt);
\draw[color=rvwvcq] (-1.0130311955761422,6.455853221999378) node {\large\textsf{A}};
\draw [fill=rvwvcq] (-7.98392035652087,-4.016079643479129) circle (2.5pt);
\draw[color=rvwvcq] (-7.975516822039367,-4.4675079090176105) node {\large\textsf{B}};
\draw [fill=rvwvcq] (7.01607964347913,-4.016079643479129) circle (2.5pt);
\draw[color=rvwvcq] (6.994631257030524,-4.469268978580221) node {\large\textsf{C}};
\draw [fill=rvwvcq] (-0.988911730357444,3.3283625653086304) circle (2.5pt);
\draw[color=rvwvcq] (-0.5306418912022467,3.5293581087977484) node {\large\textsf{D}};
\draw [fill=rvwvcq] (-0.988911730357444,-2.219114434991162) circle (2.5pt);
\draw[color=rvwvcq] (-1.4471815695126482,-1.8573224567107458) node {\large\textsf{E}};
\draw [fill=rvwvcq] (3.7063441655484732,-2.315592295865941) circle (2.5pt);
\draw[color=rvwvcq] (4.116375074266281,-2.082437465418564) node {\large\textsf{F}};
\draw [fill=rvwvcq] (0.18490224361903554,-0.19307935662080328) circle (2.5pt);
\draw[color=rvwvcq] (0.5466942218994533,0.10439404774309405) node {\large\textsf{G}};
\draw [fill=rvwvcq] (0.5064951132016323,-1.4472915479929305) circle (2.5pt);
\draw[color=rvwvcq] (0.0482252740464279,-1.2302163610246821) node {\large\textsf{H}};
\draw [fill=rvwvcq] (0.9084862001798789,-1.1417783218894635) circle (2.5pt);
\draw[color=rvwvcq] (1.0451631697524788,-0.9560659870881768) node {\large\textsf{I}};
\end{scriptsize}
\end{tikzpicture}
\caption{The  figure shows a specific subdivision of the $2$-simplex $ABC$ defined on $d+1=3$ vertices. The vertices that are higher in the alphabetical order are also higher with respect to the ordering $\succ$. For instance, $G\succ D \succ A$. 
The above triangulation is obtained as follows: First, the simplex $ABC$ is stellar subdivided by the introduction of the vertex $D$. Since $BCD$ is the lexicographically highest simplex, it is the only one that is stellar subdivided by the introduction of the vertex $E$. Then, $CDE$ which is the lexicographically highest simplex is subdivided by the introduction of the vertex $F$, and so on. Note that at each step the lexicographically highest vertices always span a simplex, and that simplex is the one that is subdivided. For instance, after the first subdivision, $BCD$ is a simplex, after the second subdivision $CDE$ is a simplex, after the third subdivision $DEF$ is a simplex, and so on. The process stops after $2(d+1)$ subdivisions. In this case, we perform six subdivisions. The total number of $d$-simplices introduced is $2d(d+1)+1$, which in this case is 13. Note that the vertices $A,B$ and $C$ do not lie in the respective stars of the vertices of the highlighted triangle $GHI$.}
\label{fig:subdivision}
\end{figure}
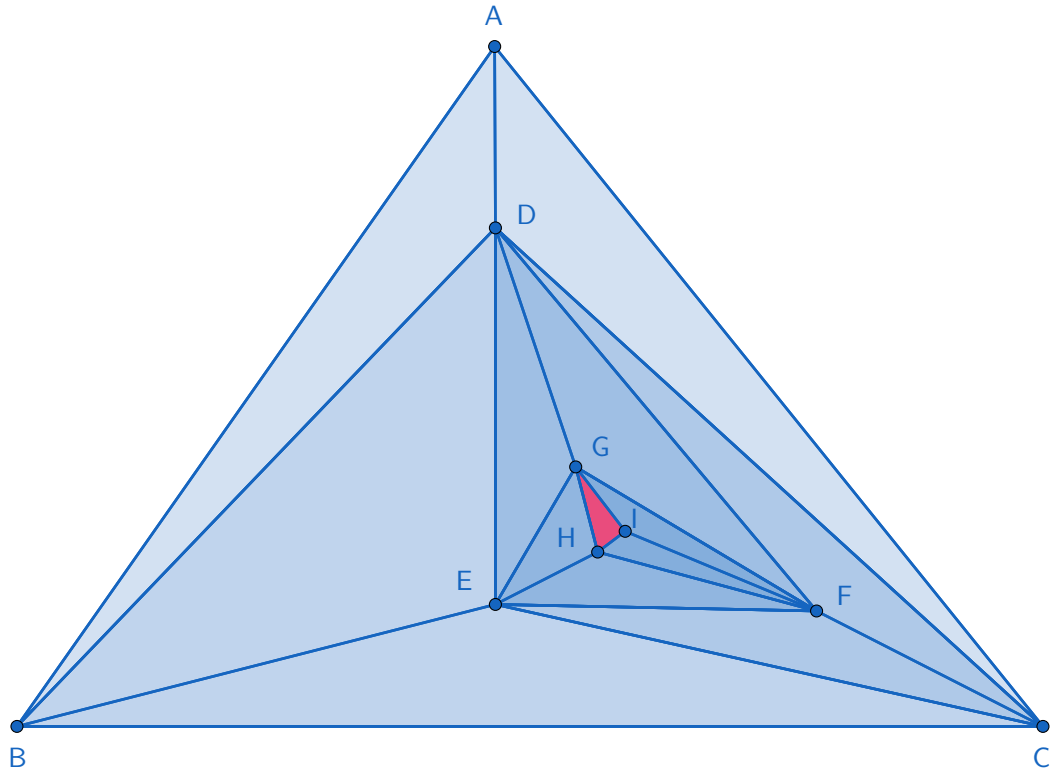

\begin{remark}
The total number of $d$-simplices in $\altaltcomplex_{i}$ for $i\in[0,d+1]$ are $2 i\cdot d+1$. So $\altaltcomplex_{\nu}$
has $2d(d+1)+1$ $d$-simplices. Also, by construction, $\Omega_{2(d+1)+1}$
is the lexicographically highest $d$-simplex of $\altaltcomplex_{\nu}$.
\end{remark}

\begin{lemma} For every $i\in[2(d+1)+1]$, $\Omega_{i}$ is the full simplex on  the $d+1$ lexicographically highest vertices of $\altaltcomplex_{i-1}$. \label{lem:allhigh}
\end{lemma}
\begin{proof}
This is trivially true for $i=1$ as $\altaltcomplex_{0}$ has only $d+1$ vertices.
Suppose that the statement of the lemma holds true for all $i\in[j]$
for some $j>1$. Let $\left\{ u_{0},u_{1},\dots,u_{d}\right\} $ be
the vertices of $\Omega_{j}$ where $u_{k}\succ_{\nu} u_{k-1}$ for $k\in[d]$.
Then, by construction, $\Omega_{j+1}=\left\{ u_{1},\dots,u_{d},v_{j}\right\} $,
which coincides with the set of lexicographically highest vertices of $\altaltcomplex_{j}$.
\end{proof}

\begin{lemma} \label{lem:highedges}
Let $v_{i}$ be the vertex introduced during the $i$-th iteration
of the algorithm. If $\left\{ u,v_{i}\right\} $ is an edge in $\altaltcomplex_{\nu}$,
then it has two types.
\begin{compactitem}
\item (type-1) $  v_{i} \succ_{\nu} u$, or 
\item (type-2) $u \succ_{\nu} v_{i}$ and there are at most $d$ vertices $w_{j}$,
$j\in[d]$ such that $u\succ_{\nu} w_{j}\succ_{\nu} v_{i}$. 
\end{compactitem}
\end{lemma}
\begin{proof}
In complex $\altaltcomplex_{i}$, $v_{i}$ has degree $d+1$. In particular,  denoting the vertices of $\Omega_{i}$ by $\left\{ u_{0},u_{1},\dots,u_{d}\right\} $, the edges $\left\{ v_{i},u_{k}\right\} $
for $k\in[0,d]$ belong to  $\altaltcomplex_{i}$. So all edges of $\altaltcomplex_{i}$ incident on  $v_{i}$ are of  \typeone.

Moreover, for every $i'>i$, every vertex $v'\neq v_{i}$ of $\Omega_{i'}$ satisifes
$v'\succ_{\nu} v_{i}$ by \Cref{lem:allhigh}. Hence, the newly added edges in $\altaltcomplex_{i'}$  for $i'>i$  that are incident on $v_{i}$ are of  \typetwo. 
Again, using \Cref{lem:allhigh} inductively, there can be at most $(d+1)$ such vertices $v'$ in the final complex $\altaltcomplex_{\nu}$. 
\end{proof}

\begin{proposition} \label{lem:satisfymunkres}
Let $v$ be a vertex of $\Omega_{2(d+1)+1}$. Then, $\left(\str_{\altaltcomplex_{\nu}}v\right)\bigcap \VCC=\emptyset$.
\end{proposition}
\begin{proof}
The complex $\altaltcomplex_{\nu}$ has $3(d+1)$ vertices totally ordered by $\succ_{\nu}$.
By \Cref{lem:allhigh}, the vertices of $\Omega_{2(d+1)+1}$  are the highest $d+1$
vertices ordered by $\succ_{\nu}$. By construction, the vertices in $\VCC$
are are the lowest $d+1$ vertices ordered by $\succ_{\nu}$. 

By \Cref{lem:highedges}, the vertices of $\Omega_{2(d+1)+1}$ do not have any edges
in common with vertices in $\VCC$. The claim follows.
\end{proof}

\subsubsection{Description of the reduction} \label{sec:descgadget}
We now give a detailed description of the reduction.
 As before, associated to a $k$-colored graph $G=(V,E)$, we  define an $r$-dimensional complex $\altcomplex(G)$  as follows. 


\paragraph*{Vertices. } 
Let  $V' = V \bigcup [k]$, and $r = |V| - 1$. Then, $|V'| = r+k+1$.
Include the vertex set $V'$ in $\altcomplex(G)$.
In what follows, we add further vertices to $\altcomplex(G)$.


\paragraph*{Ordering relation $\succ_{V'}$ on vertices of $\altcomplex(G)$. }
We now impose the following ordering relation on $V'$. 
Enumerate the vertices of $G$ according to a fixed total order $V=\{v_1,v_2,\dots, v_{r+1}$\}.
For every color $i\in[k]$ and $j\in[r+1]$, we have $i\succ_{V'} v_{j}$.
For $i_{2} \geq i_{1}$, we have $i_{2}\succ_{V'} i_{1}$, and for $j_{2} \geq j_{1}$,
we have $v_{j_{2}}\succ_{V'} v_{j_{1}}$.


\begin{remark}[Implementing undesirability]  \label{rem:inadsim}
The  undesirability of certain $r$-simplices is implemented in the gadget as follows: 
Let  $m  = n^{3}$. Then, to every undesirable $r$-simplex $ \omega = \{v_1,v_2,\dots, v_{r+1}\}$, associate $m$ new vertices $\UCC^{\omega} =  \{u_1^\omega,u_2^\omega,\dots,u_m^\omega\}$. 
For  every $\ell \in [m]$, let $\mu_\ell = \{v_1,v_2,\dots, v_{r+1}, u_\ell^\omega\}$.
Now introduce $m(r+1)$ \emph{new}  $r$-simplices  
\[\Upsilon^{\omega}  = \left \{ \{ \text{facets of  }\mu_\ell  \} \setminus \{\omega\} \,\, | \,\, \ell \in [m] \right \}.\]
Note that for any two undesirable simplices $\omega_1$ and $ \omega_2$ we have,  $\UCC^{\omega_1} \bigcap \UCC^{\omega_2} = \emptyset $ and $\Upsilon^{\omega_1}  \bigcap \Upsilon^{\omega_2} = \emptyset$.
As observed later, introducing these new simplices makes inclusion of $\omega$ in the solution set prohibitively expensive.
Please refer to \Cref{fig:undesirablebnt} for an illustrative example.
For undesirable simplices, we denote the set of $r$-simplices in $\Upsilon^{\omega} \bigcup \omega$ by $[\omega]$.
For admissible simplices, $[\omega] = \omega$.
\end{remark}

\begin{figure}
\definecolor{ffttww}{rgb}{1,0.2,0.4}
\definecolor{qqwwzz}{rgb}{0,0.4,0.6}
 \hspace*{-10em}{
\begin{tikzpicture}[scale =0.8,line cap=round,line join=round,>=triangle 45,x=1cm,y=1cm]
\clip(-15.427534520845159,-3.76562449915875) rectangle (7.394464297765075,3.708059712160069);
\fill[line width=2pt,color=qqwwzz,fill=qqwwzz,fill opacity=0.0] (-5,-3.02) -- (-9,1) -- (0,-3) -- cycle;
\fill[line width=2pt,color=qqwwzz,fill=qqwwzz,fill opacity=0.0] (-5,-3.02) -- (-7,2) -- (0,-3) -- cycle;
\fill[line width=2pt,color=qqwwzz,fill=qqwwzz,fill opacity=0.0] (-5,-3.02) -- (-5,3) -- (0,-3) -- cycle;
\fill[line width=2pt,color=qqwwzz,fill=qqwwzz,fill opacity=0.0] (-5,-3.02) -- (-2.530096126046699,3.037225251130051) -- (0,-3) -- cycle;
\fill[line width=2pt,color=qqwwzz,fill=qqwwzz,fill opacity=0.0] (-5,-3.02) -- (0,3) -- (0,-3) -- cycle;
\fill[line width=2pt,color=qqwwzz,fill=qqwwzz,fill opacity=0.0] (-5,-3.02) -- (2,2) -- (0,-3) -- cycle;
\fill[line width=2pt,color=qqwwzz,fill=qqwwzz,fill opacity=0.0] (-5,-3.02) -- (4,1) -- (0,-3) -- cycle;
\draw [line width=0.8pt,color=qqwwzz] (-5,-3.02)-- (-9,1);
\draw [line width=0.8pt,color=qqwwzz] (-9,1)-- (0,-3);
\draw [line width=0.8pt,color=qqwwzz] (0,-3)-- (-5,-3.02);
\draw [line width=0.8pt,color=qqwwzz] (-5,-3.02)-- (-7,2);
\draw [line width=0.8pt,color=qqwwzz] (-7,2)-- (0,-3);
\draw [line width=0.8pt,color=qqwwzz] (0,-3)-- (-5,-3.02);
\draw [line width=0.8pt,color=qqwwzz] (-5,-3.02)-- (-5,3);
\draw [line width=0.8pt,color=qqwwzz] (-5,3)-- (0,-3);
\draw [line width=0.8pt,color=qqwwzz] (0,-3)-- (-5,-3.02);
\draw [line width=0.8pt,color=qqwwzz] (-5,-3.02)-- (-2.530096126046699,3.037225251130051);
\draw [line width=0.8pt,color=qqwwzz] (-2.530096126046699,3.037225251130051)-- (0,-3);
\draw [line width=0.8pt,color=qqwwzz] (0,-3)-- (-5,-3.02);
\draw [line width=0.8pt,color=qqwwzz] (-5,-3.02)-- (0,3);
\draw [line width=0.8pt,color=qqwwzz] (0,3)-- (0,-3);
\draw [line width=0.8pt,color=qqwwzz] (0,-3)-- (-5,-3.02);
\draw [line width=0.8pt,color=qqwwzz] (-5,-3.02)-- (2,2);
\draw [line width=0.8pt,color=qqwwzz] (2,2)-- (0,-3);
\draw [line width=0.8pt,color=qqwwzz] (0,-3)-- (-5,-3.02);
\draw [line width=0.8pt,color=qqwwzz] (-5,-3.02)-- (4,1);
\draw [line width=0.8pt,color=qqwwzz] (4,1)-- (0,-3);
\draw [line width=0.8pt,color=qqwwzz] (0,-3)-- (-5,-3.02);
\draw [line width=2.8pt,color=DarkOrchid] (-5,-3.02)-- (0,-3);
\draw [fill=qqwwzz] (-5,-3.02) circle (2.5pt);
\draw [fill=qqwwzz] (0,-3) circle (2.5pt);
\draw[color=DarkOrchid] (-2.0863163059859424,-3.5172507555617347) node {\Large $\omega$};
\draw [fill=qqwwzz] (-9,1) circle (2.5pt);
\draw[color=qqwwzz] (-9.155624408104944,1.4596859685630539) node {\Large $u_1$};
\draw[color=qqwwzz] (-5.957516779294719,-1.3788710984282702) node { $\partial \mu_1$};
\draw [fill=qqwwzz] (-7,2) circle (2.5pt);
\draw[color=qqwwzz] (-7.195719026570323,2.5194139402398146) node {\Large $u_2$};
\draw[color=qqwwzz] (-5.300568228097033,0.0462276921108151) node { $\partial \mu_2$};

\draw [fill=qqwwzz] (-5,3) circle (2.5pt);
\draw[color=qqwwzz] (-4.9627208005371495,3.484523343016865) node {\Large $u_3$};
\draw[color=qqwwzz] (-4.165145401300505,0.4567291382261194) node { $\partial \mu_3$};
\draw [fill=qqwwzz] (-2.530096126046699,3.037225251130051) circle (2.5pt);
\draw[color=qqwwzz] (-2.4269431540249014,3.5602181981366336) node {\Large $u_i$};
\draw[color=qqwwzz] (-2.3890957264650168,0.759508558705194) node { $\partial \mu_i$};
\draw [fill=qqwwzz] (0,3) circle (2.5pt);
\draw[color=qqwwzz] (-0.014169647082276804,3.5602181981366336) node {\Large $u_{m-2}$};
\draw[color=qqwwzz] (-0.6386522018203679,0.7216611311453096) node { $\partial \mu_{m-2}$};
\draw [fill=qqwwzz] (2,2) circle (2.5pt);
\draw[color=qqwwzz] (1.9917440135915911,2.5194139402398146) node {\Large $u_{m-1}$};
\draw[color=qqwwzz] (0.6370840489956406,-0.05421113383231896) node { $\partial \mu_{m-1}$};
\draw [fill=qqwwzz] (4,1) circle (2.5pt);
\draw[color=qqwwzz] (3.9692721035955456,1.4786096823429962) node {\Large $u_m$};
\draw[color=qqwwzz] (1.5091893122030664,-0.8679308263698318) node { $\partial \mu_m$};

\end{tikzpicture}
}
\caption{ For every undesirable simplex $ \omega = \{v_1,v_2,\dots, v_{r+1}\}$ $m$ new vertices $\UCC^{\omega} =  \{u_1^\omega,u_2^\omega,\dots,u_m^\omega\}$ are added to $\altcomplex(G)$.
Furthermore, $m(r+1)$  new  $r$-simplices  $\Upsilon^{\omega}  = \left \{ \{ \text{facets of  }\mu_\ell  \} \setminus \{\omega\} \,\, | \,\, \ell \in [m] \right \}$ are  added to $\altcomplex(G)$.
The $r$-simplices in $\Upsilon^{\omega} \bigcup \omega$ are denoted by $[\omega]$.}
\label{fig:undesirablebnt}
\end{figure}
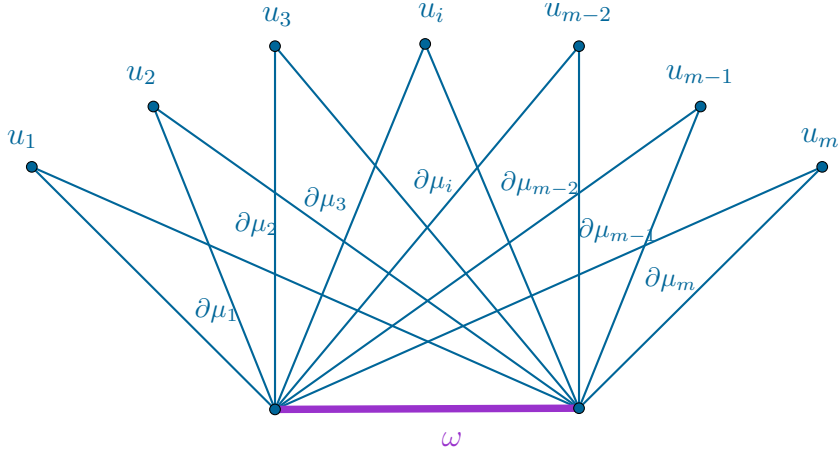



\definecolor{yqqqqq}{rgb}{0.5019607843137255,0,0}
\definecolor{ffqqtt}{rgb}{1,0,0.2}
\definecolor{ffttww}{rgb}{1,0.2,0.4}
\definecolor{zzwwff}{rgb}{0.6,0.4,1}
\definecolor{sexdts}{rgb}{0.1803921568627451,0.49019607843137253,0.19607843137254902}
\definecolor{rvwvcq}{rgb}{0.08235294117647059,0.396078431372549,0.7529411764705882}
\begin{figure}
\begin{centering}

 \hspace*{-3em}{
\begin{tikzpicture}[scale=0.9, line cap=round,line join=round,>=triangle 45,x=1cm,y=1cm]
\clip(-12.190233905111764,-7.386234566696782) rectangle (4.351172667925901,6.359621767949472);
\fill[line width=1pt,color=rvwvcq,fill=rvwvcq,fill opacity=0.10000000149011612] (-7,5.97) -- (-8,4) -- (-7,3) -- cycle;
\fill[line width=1pt,color=rvwvcq,fill=rvwvcq,fill opacity=0.10000000149011612] (-7,5.97) -- (-6,4) -- (-7,3) -- cycle;
\fill[line width=1pt,color=sexdts,fill=sexdts,fill opacity=0.1] (-4,0) -- (-5,-2) -- (-4,-3) -- cycle;
\fill[line width=1pt,color=sexdts,fill=sexdts,fill opacity=0.1] (-4,-3) -- (-3,-2) -- (-4,0) -- cycle;
\fill[line width=1pt,color=sexdts,fill=sexdts,fill opacity=0.1] (-10,0) -- (-11,-2) -- (-10,-3) -- cycle;
\fill[line width=1pt,color=sexdts,fill=sexdts,fill opacity=0.1] (-10,0) -- (-9,-2) -- (-10,-3) -- cycle;
\fill[line width=1pt,color=sexdts,fill=sexdts,fill opacity=0.1] (1,-2) -- (2,-3) -- (2,0) -- cycle;
\fill[line width=1pt,color=sexdts,fill=sexdts,fill opacity=0.1] (2,0) -- (3,-2) -- (2,-3) -- cycle;
\fill[line width=1pt,color=rvwvcq,fill=rvwvcq,fill opacity=0.10000000149011612] (-1,6) -- (-2,4) -- (-1,3) -- cycle;
\fill[line width=1pt,color=rvwvcq,fill=rvwvcq,fill opacity=0.10000000149011612] (-1,6) -- (0,4) -- (-1,3) -- cycle;
\fill[line width=1pt,color=ffttww,fill=ffttww,fill opacity=0.1] (-6.708124509329276,4.8371672633900875) -- (-6.795869737488912,4.3399443038188155) -- (-6.386392006077277,4.442313736671724) -- cycle;
\fill[line width=1pt,color=ffttww,fill=ffttww,fill opacity=0.1] (-7.447655493016648,4.77916483329696) -- (-7.535400721176284,4.281941873725688) -- (-7.125922989764649,4.384311306578597) -- cycle;
\fill[line width=1pt,color=ffttww,fill=ffttww,fill opacity=0.1] (-1.4754628147350501,4.803838208010506) -- (-1.5632080428946862,4.306615248439234) -- (-1.153730311483051,4.408984681292143) -- cycle;
\fill[line width=1pt,color=ffttww,fill=ffttww,fill opacity=0.1] (-0.7088945410046864,4.903825374149251) -- (-0.7966397691643224,4.4066024145779785) -- (-0.3871620377526873,4.508971847430887) -- cycle;
\fill[line width=1pt,color=ffttww,fill=ffttww,fill opacity=0.1] (-9.774100968384111,-1.1031621260286695) -- (-9.861846196543748,-1.6003850855999413) -- (-9.452368465132112,-1.498015652747033) -- cycle;
\fill[line width=1pt,color=ffttww,fill=ffttww,fill opacity=0.1] (1.5478301230951717,-1.1650892647678186) -- (1.4600848949355356,-1.6623122243390906) -- (1.869562626347171,-1.559942791486182) -- cycle;
\fill[line width=1pt,color=ffttww,fill=ffttww,fill opacity=0.1] (-4.453386611607825,-1.1512548326892225) -- (-4.541131839767461,-1.6484777922604945) -- (-4.131654108355826,-1.546108359407586) -- cycle;
\fill[line width=1pt,color=ffttww,fill=ffttww,fill opacity=0.1] (-3.7159473819406323,-1.1442703432762205) -- (-3.803692610100268,-1.6414933028474925) -- (-3.3942148786886333,-1.539123869994584) -- cycle;
\fill[line width=1pt,color=ffttww,fill=ffttww,fill opacity=0.1] (-10.446217959748253,-1.1656816656468323) -- (-10.53396318790789,-1.6629046252181043) -- (-10.124485456496254,-1.5605351923651958) -- cycle;
\fill[line width=1pt,color=ffttww,fill=ffttww,fill opacity=0.1] (2.2413966294521828,-1.082761694764035) -- (2.1536514012925467,-1.5799846543353069) -- (2.563129132704182,-1.4776152214823983) -- cycle;
\draw [line width=1pt,color=rvwvcq] (-7,5.97)-- (-8,4);
\draw [line width=1pt,color=rvwvcq] (-8,4)-- (-7,3);
\draw [line width=1pt,color=rvwvcq] (-7,3)-- (-7,5.97);
\draw [line width=1pt,color=rvwvcq] (-7,5.97)-- (-6,4);
\draw [line width=1pt,color=rvwvcq] (-6,4)-- (-7,3);
\draw [line width=1pt,color=rvwvcq] (-7,3)-- (-7,5.97);
\draw [line width=1pt,dash pattern=on 2pt off 2pt,color=RoyalBlue] (-8,4)-- (-6,4);
\draw [line width=1pt,color=sexdts] (-4,0)-- (-5,-2);
\draw [line width=1pt,color=zzwwff] (-5,-2)-- (-4,-3);
\draw [line width=1pt,color=sexdts] (-4,-3)-- (-4,0);
\draw [line width=1pt,color=zzwwff] (-4,-3)-- (-3,-2);
\draw [line width=1pt,color=sexdts] (-3,-2)-- (-4,0);
\draw [line width=1pt,color=sexdts] (-4,0)-- (-4,-3);
\draw [line width=1pt,dash pattern=on 2pt off 2pt,color=zzwwff] (-5,-2)-- (-3,-2);
\draw [line width=1pt,color=sexdts] (-10,0)-- (-11,-2);
\draw [line width=1pt,color=zzwwff] (-11,-2)-- (-10,-3);
\draw [line width=1pt,color=sexdts] (-10,-3)-- (-10,0);
\draw [line width=1pt,color=sexdts] (-10,0)-- (-9,-2);
\draw [line width=1pt,color=zzwwff] (-9,-2)-- (-10,-3);
\draw [line width=1pt,color=sexdts] (-10,-3)-- (-10,0);
\draw [line width=1pt,dash pattern=on 2pt off 2pt,color=zzwwff] (-11,-2)-- (-9,-2);
\draw [line width=1pt,color=zzwwff] (1,-2)-- (2,-3);
\draw [line width=1pt,color=sexdts] (2,-3)-- (2,0);
\draw [line width=1pt,color=sexdts] (2,0)-- (1,-2);
\draw [line width=1pt,color=sexdts] (2,0)-- (3,-2);
\draw [line width=1pt,color=zzwwff] (3,-2)-- (2,-3);
\draw [line width=1pt,color=sexdts] (2,-3)-- (2,0);
\draw [line width=1pt,dash pattern=on 2pt off 2pt,color=zzwwff] (1,-2)-- (3,-2);
\draw [line width=1pt,color=rvwvcq] (-1,6)-- (-2,4);
\draw [line width=1pt,color=rvwvcq] (-2,4)-- (-1,3);
\draw [line width=1pt,color=rvwvcq] (-1,3)-- (-1,6);
\draw [line width=1pt,color=rvwvcq] (-1,6)-- (0,4);
\draw [line width=1pt,color=rvwvcq] (0,4)-- (-1,3);
\draw [line width=1pt,color=rvwvcq] (-1,3)-- (-1,6);
\draw [line width=1pt,dash pattern=on 2pt off 2pt,color=RoyalBlue] (-2,4)-- (0,4);
\draw [line width=1pt,color=ffttww] (-6.708124509329276,4.8371672633900875)-- (-6.795869737488912,4.3399443038188155);
\draw [line width=1pt,color=ffttww] (-6.795869737488912,4.3399443038188155)-- (-6.386392006077277,4.442313736671724);
\draw [line width=1pt,color=ffttww] (-6.386392006077277,4.442313736671724)-- (-6.708124509329276,4.8371672633900875);
\draw [line width=1pt,color=ffttww] (-7.447655493016648,4.77916483329696)-- (-7.535400721176284,4.281941873725688);
\draw [line width=1pt,color=ffttww] (-7.535400721176284,4.281941873725688)-- (-7.125922989764649,4.384311306578597);
\draw [line width=1pt,color=ffttww] (-7.125922989764649,4.384311306578597)-- (-7.447655493016648,4.77916483329696);
\draw [line width=1pt,color=ffttww] (-1.4754628147350501,4.803838208010506)-- (-1.5632080428946862,4.306615248439234);
\draw [line width=1pt,color=ffttww] (-1.5632080428946862,4.306615248439234)-- (-1.153730311483051,4.408984681292143);
\draw [line width=1pt,color=ffttww] (-1.153730311483051,4.408984681292143)-- (-1.4754628147350501,4.803838208010506);
\draw [line width=1pt,color=ffttww] (-0.7088945410046864,4.903825374149251)-- (-0.7966397691643224,4.4066024145779785);
\draw [line width=1pt,color=ffttww] (-0.7966397691643224,4.4066024145779785)-- (-0.3871620377526873,4.508971847430887);
\draw [line width=1pt,color=ffttww] (-0.3871620377526873,4.508971847430887)-- (-0.7088945410046864,4.903825374149251);
\draw [line width=1pt,color=ffttww] (-9.774100968384111,-1.1031621260286695)-- (-9.861846196543748,-1.6003850855999413);
\draw [line width=1pt,color=ffttww] (-9.861846196543748,-1.6003850855999413)-- (-9.452368465132112,-1.498015652747033);
\draw [line width=1pt,color=ffttww] (-9.452368465132112,-1.498015652747033)-- (-9.774100968384111,-1.1031621260286695);
\draw [line width=1pt,color=ffttww] (1.5478301230951717,-1.1650892647678186)-- (1.4600848949355356,-1.6623122243390906);
\draw [line width=1pt,color=ffttww] (1.4600848949355356,-1.6623122243390906)-- (1.869562626347171,-1.559942791486182);
\draw [line width=1pt,color=ffttww] (1.869562626347171,-1.559942791486182)-- (1.5478301230951717,-1.1650892647678186);
\draw [line width=1pt,color=zzwwff] (-5,-6)-- (-3,-6);
\draw [line width=1pt,color=zzwwff] (-4,-7)-- (-5,-6);
\draw [line width=1pt,color=zzwwff] (-4,-7)-- (-3,-6);
\draw [line width=0.8pt,dash pattern=on 2pt off 2pt,color=ffqqtt] (-9.711959248618388,-1.3659497179250482)-- (-7.378441377708626,4.512968233096253);
\draw [line width=0.8pt,dash pattern=on 2pt off 2pt,color=Black] (-6.610998651150528,4.587269947586744)-- (-1.416014324417937,4.589830818850211);
\draw [line width=0.8pt,dash pattern=on 2pt off 2pt,color=ffqqtt] (-0.6320190206580366,4.633999568357811)-- (1.6095450168526677,-1.3950347394296219);
\draw [line width=1pt,color=ffttww] (-4.453386611607825,-1.1512548326892225)-- (-4.541131839767461,-1.6484777922604945);
\draw [line width=1pt,color=ffttww] (-4.541131839767461,-1.6484777922604945)-- (-4.131654108355826,-1.546108359407586);
\draw [line width=1pt,color=ffttww] (-4.131654108355826,-1.546108359407586)-- (-4.453386611607825,-1.1512548326892225);
\draw [line width=1pt,color=ffttww] (-3.7159473819406323,-1.1442703432762205)-- (-3.803692610100268,-1.6414933028474925);
\draw [line width=1pt,color=ffttww] (-3.803692610100268,-1.6414933028474925)-- (-3.3942148786886333,-1.539123869994584);
\draw [line width=1pt,color=ffttww] (-3.3942148786886333,-1.539123869994584)-- (-3.7159473819406323,-1.1442703432762205);
\draw [line width=1pt,color=ffttww] (-10.446217959748253,-1.1656816656468323)-- (-10.53396318790789,-1.6629046252181043);
\draw [line width=1pt,color=ffttww] (-10.53396318790789,-1.6629046252181043)-- (-10.124485456496254,-1.5605351923651958);
\draw [line width=1pt,color=ffttww] (-10.124485456496254,-1.5605351923651958)-- (-10.446217959748253,-1.1656816656468323);
\draw [line width=1pt,color=ffttww] (2.2413966294521828,-1.082761694764035)-- (2.1536514012925467,-1.5799846543353069);
\draw [line width=1pt,color=ffttww] (2.1536514012925467,-1.5799846543353069)-- (2.563129132704182,-1.4776152214823983);
\draw [line width=1pt,color=ffttww] (2.563129132704182,-1.4776152214823983)-- (2.2413966294521828,-1.082761694764035);
\draw [line width=0.8pt,dash pattern=on 2pt off 2pt,color=ffqqtt] (-10.412426270878699,-1.4204401828798467)-- (-9.192692232399477,1.4100486194876807);
\draw [line width=0.8pt,dash pattern=on 2pt off 2pt,color=ffqqtt] (-10.522880657754708,-0.3211586778680315)-- (-9.798473686528345,1.4772534769119874);
\draw [line width=0.8pt,dash pattern=on 2pt off 2pt,color=ffqqtt] (-8.765378312632537,-0.49178997351096093)-- (-8.014600611803647,1.4363436672541419);
\draw [line width=0.8pt,dash pattern=on 2pt off 2pt,color=zzwwff] (-9.211553878841203,-2.675677060436653)-- (-5,-5.5);
\draw [line width=0.8pt,dash pattern=on 2pt off 2pt,color=zzwwff] (-4.008204534611733,-3.449147908903197)-- (-4,-5.5);
\draw [line width=0.8pt,dash pattern=on 2pt off 2pt,color=zzwwff] (1.3422078328478801,-2.6186651966837657)-- (-3,-5.5);
\draw [line width=0.8pt,dash pattern=on 2pt off 2pt,color=ffqqtt] (2.3062209828086697,-1.3164511507729089)-- (1.2278455529293726,1.4468858882927969);
\draw [line width=0.8pt,dash pattern=on 2pt off 2pt,color=ffqqtt] (-0.01902603786856602,1.357021269136189)-- (0.6785224571463321,-0.5906919063302037);
\draw [line width=0.8pt,dash pattern=on 2pt off 2pt,color=ffqqtt] (1.946762506182236,1.4356528108982212)-- (2.8004763881700128,-0.7547972810441069);
\draw [line width=0.8pt,dash pattern=on 2pt off 2pt,color=ffqqtt] (-4,4)-- (-4,1.5);
\draw [line width=0.8pt,dash pattern=on 2pt off 2pt,color=ffqqtt] (-5,4)-- (-5.5,1.5);
\draw [line width=0.8pt,dash pattern=on 2pt off 2pt,color=ffqqtt] (-3,4)-- (-2.5,1.5);

\draw (-11.272629010314907,-2.4055527894910529) node[anchor=north west] {\textcolor{ForestGreen}{$\hat{\Delta} \sigma_i$}};
\draw (2.471973030263508,-2.4055527894910529) node[anchor=north west] {\textcolor{ForestGreen}{$\hat{\Delta} \sigma_j$}};

\draw (-2.471973030263508,3.5887990166884047) node[anchor=north west] {\textcolor{RoyalBlue}{$\Delta\tau_{j,i}^u$}};
\draw (-6.76129558585731,3.5887990166884047) node[anchor=north west] {\textcolor{RoyalBlue}{$\Delta\tau_{i,j}^v$}};

\draw (-5.46129558585731,-6.228402706034805) node[anchor=north west] {\textcolor{BlueViolet}{$\partial V$}};

\draw (-9.272629010314907,-0.9055527894910529) node[anchor=north west] {\textcolor{ffttww}{$\alpha_i^v(\sigma_i)$}};
\draw (-0.171973030263508,-0.9055527894910529) node[anchor=north west] {\textcolor{ffttww}{$\alpha_j^u(\sigma_j)$}};

\draw (-3.471973030263508,5.5887990166884047) node[anchor=north west] {\textcolor{ffttww}{$\beta_{i,j}^{v,u}(\tau_{j,i}^u)$}};
\draw (-6.26129558585731,5.5887990166884047) node[anchor=north west] {\textcolor{ffttww}{$\beta_{i,j}^{v,u}(\tau_{i,j}^v)$}};

\draw (-0.271973030263508,4.9887990166884047) node[anchor=north west] {\textcolor{ffttww}{$\alpha_j^u(\tau_{j,i}^u)$}};
\draw (-9.56129558585731,4.9887990166884047) node[anchor=north west] {\textcolor{ffttww}{$\alpha_i^v(\tau_{i,j}^v)$}};

\end{tikzpicture}
}
\end{centering} 
\caption{The figure depicts the construction of $\altcomplex(G)$ via  identifications of various gadgets as described in \Cref{alg:attachments}. In particular, the dashed red lines show identifications of the (red) congruent faces  of  \typeone gadgets (shown in green) and  \typetwo gadgets (shown in blue).  The dashed black line shows identifications of the (red) congruent faces  of two distinct  \typetwo gadgets.
Note that some of the red dashed lines are only partially drawn.
The red faces are the lexicographically highest distinguished faces obtained by S-subdivisions described in \Cref{sec:subdivideit}.  The construction of the \typeone green gadgets is described in \Cref{alg:deltasigma}, and the construction of the  \typetwo blue gadgets is described in \Cref{alg:deltatau}. Note that the full simplex with vertex set  $V$ (or its subdivision) does not appear as a simplex in any of the \typeone gadgets. In fact the \typeone gadget are $r$-manifolds with $\partial V$ as their common boundary, and the dashed purple lines depict precisely that. The  \typetwo gadgets are topological $r$-spheres.}
\label{fig:mainfigbnt}
\end{figure}

\paragraph*{Gadgets. } 
The complex $\altcomplex(G)$ is constructed by gluing the distinguished faces of two types of gadgets.
Next, we describe these two types of  gadgets.

 \paragraph*{Gadgets of  \typeone.}
	 
The construction of gadgets of   \typeone is explained in detail in the pseudocode of \Cref{alg:deltasigma}. Below, we provide a high-level sketch.

First, we describe the subroutine \textsc{SubdivideDelta1}. In this subroutine,
given an index $i \in [k]$, we begin our construction with the complex $\hat{\partial}\sigma_{i}=\{\text{facets of }\sigma_{i}\}\setminus\{V\}$,
where as in \Cref{sub:wonetop}, $\sigma_{i}=V\bigcup{\{i\}}$.
The vertices of $\hat{\partial}\sigma_{i}$ inherit an order from
$\succ_{V'}$. 

\begin{definition}[Pre-admissible and non-pre-admissible simplices of $\hat{\partial}\sigma_{i}$]
For every $v \in V_i$, the simplices $\ensuremath{a_{i}^{v}=(V\setminus\left\{ {v}\right\}) \bigcup{\{i\}}}$
are called the \emph{pre-admissible simplices} of $\hat{\partial}\sigma_{i}$,
and all other simplices of $\hat{\partial}\sigma_{i}$ are called
\emph{non-pre-admissible}. 
\end{definition}

The procedure \textsc{S-Subdivide} described in \Cref{sec:subdivideit} takes an $r$-simplex
$\nu$ as input and returns a subdivision of  $\nu$ along
with the (lexicographically highest) distinguished simplex from within the subdivided simplex. For every pre-admissible simplex $a_{i}^{v}$,
its subdivision is denoted by $\ensuremath{\altaltcomplex_{i}^{v}}$,
and the  distinguished simplex  of $\ensuremath{\altaltcomplex_{i}^{v}}$
is denoted by $\alpha_{i}^{v}(\sigma_{i})$. The complex $\altaltcomplex$
is formed by taking the union of the subdivided pre-admissible simplices.

Let $A$ denote the collection of  \emph{distinguished simplices} in $\altaltcomplex$,
and let $W$ denote the set of non-pre-admissible $r$-simplices of $ \hat{\partial}\sigma_{i}$. 
 Since there are $V_{i}$ pre-admissible simplices for color
$i$, $|A|=|V_{i}|$. Finally, the complex $\hat{\Delta}\sigma_{i}$
is formed by taking the union of the non-pre-admissible simplices, namely $W$, with the
collection of subdivisions of the pre-admissible simplices, namely $\altaltcomplex$. We end the description of \textsc{SubdivideDelta1} with one last definition.

\begin{definition}[Undesirable simplices of $\hat{\Delta}\sigma_{i}$]
At the end of the procedure \textsc{SubdivideDelta1}, the simplices in $\hat{\Delta}\sigma_{i} \setminus A$   are called the \emph{undesirable simplices} of $\hat{\Delta}\sigma_{i}$. 
\end{definition}

In  procedure \textsc{TypeZ1}, the complex $\ZCC_{1}$ is constructed. To begin with, the subroutine \textsc{SubdivideDelta1} is invoked for  every $i\in[k]$,
which returns the complex $\hat{\Delta}\sigma_{i}$ along with its
\emph{distinguished simplices} $A_{\sigma_{i}}$. 
Next, we add further simplices to
$\hat{\Delta}\sigma_{i}$ in order to implement undesirability of simplices as per \Cref{rem:inadsim}.   As per the notation used in  \textsc{TypeZ1},  $\hat{\Delta}\sigma_{i}\setminus A_{\sigma_{i}}$  are the undesirable simplices of $\hat{\Delta}\sigma_{i}$. Then,
to every undesirable simplex $\omega\in\hat{\Delta}\sigma_{i}\setminus A_{\sigma_{i}}$,
we add $(r+1)m$ simplices $\Upsilon^{\omega}$ to $\hat{\Delta}\sigma_{i}$,
completing the construction of $\hat{\Delta}\sigma_{i}$. 
The complex $\ZCC_{1}$ is then given by the union of all simplices in $\hat{\Delta}\sigma_{i}$ for every $i$.
We end the description of  \textsc{TypeZ1} with a definition.
\begin{definition}[Inadmissible simplices of $\hat{\Delta}\sigma_{i}$]
At the end of the procedure \textsc{TypeZ1}, the simplices in $\hat{\Delta}\sigma_{i} \setminus A_{\sigma_{i}}$ are the \emph{inadmissible simplices} of $\hat{\Delta}\sigma_{i}$.
\end{definition}
	
 \begin{algorithm}
\caption{ Construction of complex $\hat{\Delta} \sigma_i$ }\label{alg:deltasigma}
\begin{algorithmic}[1]

\Procedure{SubdivideDelta1}{$i$}
\State{ $\sigma_{i} \gets V\bigcup{\{i\}}$; }
\State{ Let $ \hat{\partial}\sigma_{i}$ be the $r$-complex  $ \partial \sigma_{i} \setminus \{V\}$;}
\State{The vertex set $U$ of $\hat{\partial}\sigma_{i}$ is ordered by $\succ$, obtained by restricting $\succ_{V'}$ to $U$;}
\State{Every $r$-simplex $a_i^v =(V\setminus\left\{ {v}\right\}) \bigcup{\{i\}}$ of $\hat{\partial}\sigma_{i}$ with $v \in V_{i}$ is deemed \emph{pre-admissible};}
\State{Let $W$ denote the set of $r$-simplices of $ \hat{\partial}\sigma_{i}$ that are not pre-admissible;}
\State{$A \gets \emptyset$;\quad$\altaltcomplex \gets \emptyset$;}

\For{\textbf{each }pre-admissible simplex $a_i^v$ of $\hat{\partial}\sigma_{i}$}

\State{$\altaltcomplex_i^v,\, \alpha_i^v(\sigma_i) \gets $ \textsc{S-subdivide}($a_i^v,\succ$);}

\State{$A \gets  A \bigcup \{ \alpha_i^v(\sigma_i)\}$;}

\State{$\altaltcomplex \gets  \altaltcomplex \bigcup\altaltcomplex_i^v$;}

\EndFor
\State{$\hat{\Delta} \sigma_i \gets\altaltcomplex\bigcup W ;$}
\State{\textbf{return }$\hat{\Delta}\sigma_i, \,A$;} \Comment{The simplices in $\hat{\Delta}\sigma_i \setminus A$ are undesirable.}

\EndProcedure

\Statex

\Procedure{TypeZ1}{}
\State{$\ZCC_1\gets \emptyset$;}
\For{ $i =1$ to $k$}
\State{ $\hat{\Delta}\sigma_i, \,A_{\sigma_{i}} \gets \textsc{SubdivideDelta1}(i)$ ;}
\State{The simplices in $A_{\sigma_{i}}$ are  the distinguished simplices of $\hat{\Delta}\sigma_i$;}
\State{The simplices in $\hat{\Delta}\sigma_i \setminus A_{\sigma_{i}}$ are deemed the undesirable simplices of $\hat{\Delta}\sigma_i$;}
\For{\textbf{every }undesirable simplex $\omega$ in $\hat{\Delta}\sigma_i$} \Comment{as described in \Cref{rem:inadsim}}
\State{ Add  $(r+1)m$ simplices $\Upsilon^{\omega}$ to $\hat{\Delta}\sigma_i$;}
\EndFor  \Comment{The simplices in $\hat{\Delta}\sigma_i \setminus A_{\sigma_{i}}$ are inadmissible.}
\State{$\ZCC_1 \gets \ZCC_1 \bigcup  \hat{\Delta}\sigma_i$;} 
\EndFor 

\EndProcedure
\end{algorithmic}
\end{algorithm}

\paragraph*{Gadgets of  \typetwo.}

We now provide a high-level description of  gadgets of   \typetwo, the pseudocode of which is provided in \Cref{alg:deltatau}.
The \typetwo gadgets are indexed by $t$.

First, we describe the subroutine \textsc{SubdivideDelta2}. In this subroutine, 
given a vertex $v\in V_i$ a color $j\neq i$, and an index $t$, we introduce a vertex set $V_t$ whose vertices are in one-to-one correspondence with the vertices $(V\setminus\left\{ {v}\right\})\bigcup\{i,j\}$.
Let $\tau_{i,j}^{v}$ be the full $(r+1)$-simplex on $V_t$, and $\partial\tau_{i,j}^{v}$ be the complex induced by the facets of $\tau_{i,j}^{v}$.
The vertices of $\partial\tau_{i,j}^{v}$  are ordered according to the same rules as
$\succ_{V'}$. 

\begin{definition}[Pre-admissible and non-pre-admissible simplices of $\partial\tau_{i,j}^{v}$]
The simplices $a_{i}^{v}= V_t \setminus \{ j_t\}$  for every $v\in V_i$, and the simplices $b_{i,j}^{v,u} =  V_t \setminus \{u_t\}, \text{ where } u \in V_{j} \text{ and }\{u,v\}  \in E $  are also said to be the \emph{pre-admissible simplices}  of $\partial\tau_{i,j}^{v}$.
All other $r$-simplices of $\partial\tau_{i,j}^{v}$ are deemed \emph{non-pre-admissible}.
\end{definition}

 We invoke the procedure \mbox{\textsc{S-Subdivide}} described in \Cref{sec:subdivideit} to subdivide the pre-admissible simplices of  $\partial\tau_{i,j}^{v}$.
The subdivision of a pre-admissible simplex $b_{i,j}^{v,u}$  is denoted by $\altaltcomplex_{i,j}^{v,u}$ and the distinguished simplex of $\altaltcomplex_{i,j}^{v,u}$ is denoted by  $\beta_{i,j}^{v,u}(\tau_{j,i}^u)$.
The subdivision of a pre-admissible simplex $\alpha_i^v(\tau _{i,j}^v)$ is denoted by $\altaltcomplex_i^v$ and the distinguished simplex of $\altaltcomplex_i^v$  is denoted by $a_i^v$.
 The complex $\altaltcomplex$ is formed by taking the union of the subdivided pre-admissible simplices. Furthermore, the collection of all the distinguished simplices of the subdivided pre-admissible simplices is denoted by
$A$. It is easy to check that, $|A|= k $. Finally, the complex $\Delta\tau_{i,j}^{v}$ is formed by taking the union of the non-pre-admissible simplices, namely $W$, with the collection of subdivisions of the pre-admissible simplices, namely $\altaltcomplex$. 
We conclude the description of \textsc{SubdivideDelta2} with a definition.
\begin{definition}[Undesirable simplices of $\Delta\tau_{i,j}^{v}$]
At the end of procedure \textsc{SubdivideDelta2}, the simplices in $\Delta\tau_{i,j}^{v} \setminus A$   are said to be the \emph{undesirable simplices} of $\Delta\tau_{i,j}^{v}$. 
\end{definition}

In the procedure \textsc{TypeZ2}, the complex $\ZCC_{2}$ is constructed. To do
this, the subroutine \textsc{SubdivideDelta2} is invoked for every color $i$, every vertex $v$ in $V_{i}$,  and every color $j$ where $j\neq i$,
which returns the complex ${\Delta}\tau_{i,j}^{v}$, along with its set of
distinguished simplices $A_{\tau_{i,j}^{v}}$. 
Next, we add further simplices to
${\Delta}\tau_{i,j}^{v}$, for every $v \in V_i$ and $j\in [k]\setminus \{i\}$, in order to implement undesirability of simplices as per \Cref{rem:inadsim}. 
We start the construction with the undesirable simplices of $ {\Delta}\tau_{i,j}^{v} $, namely $ {\Delta}\tau_{i,j}^{v} \setminus A_{\tau_{i,j}^{v}}$.
To every undesirable simplex $\omega\in {\Delta}\tau_{i,j}^{v} \setminus A_{\tau_{i,j}^{v}}$,
we add $(r+1)m$ simplices $\Upsilon^{\omega}$ to ${\Delta}\tau_{i,j}^{v}$,
completing the construction of ${\Delta}\tau_{i,j}^{v}$. 
The complex $\ZCC_{2}$ is then given by the union of all simplices in ${\Delta}\tau_{i,j}^{v}$ for every  $i \in [k]$, every vertex $v \in V_{i}$ and  every $j \in [k] \setminus \{ i \}$.
We conclude the description of \textsc{TypeZ2} with a definition.
\begin{definition}[Inadmissible simplices of ${\Delta}\tau_{i,j}^{v} $]
 At the end of procedure \textsc{TypeZ2}, the simplices in ${\Delta}\tau_{i,j}^{v}  \setminus A_{\tau_{i,j}^{v}}$ are the \emph{inadmissible simplices} of ${\Delta}\tau_{i,j}^{v} $.
\end{definition}

\begin{algorithm}
\caption{ Construction of complex $\Delta\tau _{i,j}^v$ }\label{alg:deltatau}
\begin{algorithmic}[1]

\Procedure{SubdivideDelta2}{$v,i,j,t$}

\State{Let $V_t \gets \emptyset$;}
\For{\textbf{every } $u \in V \setminus v$}
	\State{Add a vertex $u_t$ to $V_t$;}
\EndFor
\State{$V_t \gets V_t \bigcup \{i_t,j_t\}$;}
\State{Let $\tau _{i,j}^v $ be the full $(r+1)$-simplex on $V_t$;}
\State{Let $ \partial\tau _{i,j}^v$ be the $r$-complex induced by the facets of $\tau _{i,j}^v$;}
\State{The vertices $V_t$ of $\partial\tau _{i,j}^v$  are in a natural one-to-one correspondence to a subset of vertices in $V'$. The ordering $\succ_{t}$ on $V_t$ is defined using the same rules as for $\succ_{V'}$;}
\State{The $r$-simplex $a_i^v = V_t \setminus \{j_t\}$  is deemed \emph{pre-admissible};}
\State{The $r$-simplices $\left \{ b_{i,j}^{v,u} \, | \, b_{i,j}^{v,u} = V_t \setminus \{u_t\}, \text{ where } u \in V_{j} \text{ and }\{u,v\}  \in E \right\}$ are also deemed \emph{pre-admissible};}

\State{Let $W$ denote the set of $r$-simplices of $ \partial\tau _{i,j}^v$ that are not pre-admissible;}

\State{$\altaltcomplex_i^v,\, \alpha_i^v(\tau _{i,j}^v) \gets $ \textsc{S-subdivide}($a_i^v,\succ_{t}$);}
 
\State{$A \gets   \{ \alpha_i^v(\tau _{i,j}^v)\}$;\quad$\altaltcomplex \gets \altaltcomplex_i^v$;}

\For{\textbf{each }pre-admissible simplex $b_{i,j}^{v,u}$ of $\partial\tau _{i,j}^v$}

\State{$\altaltcomplex_{i,j}^{v,u},\, \beta_{i,j}^{v,u}(\tau_{j,i}^u) \gets $ \textsc{S-subdivide}($b_{i,j}^{v,u},\succ_{t}$);}
 
\State{$A \gets  A \bigcup \{  \beta_{i,j}^{v,u}(\tau_{j,i}^u)\}$;}

\State{$\altaltcomplex \gets  \altaltcomplex \bigcup\altaltcomplex_{i,j}^{v,u}$;}

\EndFor
\State{$\Delta\tau _{i,j}^v \gets\altaltcomplex\bigcup W ;$} 
\State{\textbf{return }$\Delta\tau _{i,j}^v, \,A$;}   \Comment{The simplices in $\Delta\tau _{i,j}^v \setminus A$ are undesirable.}

\EndProcedure

\Statex

\Procedure{TypeZ2}{}
\State{$\ZCC_2\gets \emptyset$; $t=0$;}
\For{ $i =1$ to $k$}
\For{\textbf{every }vertex $v$ in $V_{i}$,  and a color $j \in [k] \setminus \{i\}$}
\State{$t = t+1$;}
\State{ $ \Delta\tau _{i,j}^v, \,A_{\tau _{i,j}^v} \gets \textsc{SubdivideDelta2}(v,i,j,t)$ ;}
\State{The simplices in $A_{\tau _{i,j}^v}$ are  the distinguished simplices of $\Delta\tau _{i,j}^v$;}
\State{The simplices in $\Delta\tau _{i,j}^v \setminus A_{\tau _{i,j}^v}$ are deemed the undesirable simplices of $\Delta\tau _{i,j}^v$;}
\For{\textbf{every }undesirable simplex $\omega$ in $\Delta\tau _{i,j}^v$} \Comment{as described in \Cref{rem:inadsim}}
\State{ Add another $(r+1)m$ simplices $\Upsilon^{\omega}$ to $\Delta\tau _{i,j}^v$;}
\EndFor   \Comment{The simplices in $\Delta\tau _{i,j}^v  \setminus A_{\tau _{i,j}^v}$ are inadmissible.}
\State{$\ZCC_2 \gets \ZCC_2 \bigcup \Delta\tau _{i,j}^v$;}
\EndFor
\EndFor
\EndProcedure
\end{algorithmic}
\label{alg:fptths}
\end{algorithm}

\paragraph*{Attachments.} 

Let $\complex' = \ZCC_1 \bigcup \ZCC_2$.  Then, complex  $\altcomplex(G)$ is formed from $\complex'$ after making the attachments described in \Cref{alg:attachments}. 
 \begin{algorithm}[H]
\caption{ Construction of complex $\altcomplex(G)$ }\label{alg:attachments}
\begin{algorithmic}[1]

\For{\textbf{every }$v\in V_i$ and $j \neq i$}
\State{Identify the $r$-simplices: $\alpha_i^v(\tau _{i,j}^v)  \sim  \alpha_i^v(\sigma_i)$, where the identifications of vertices are consistent with respective lexicographic orderings. Denote the identified simplex by $\alpha_i^v$;}
\EndFor

\For{\textbf{every }edge $\{u,v\} \in E$ with $v\in V_i$ and $u\in V_j$}
\State{Identify the $r$-simplices:  $\beta_{i,j}^{v,u}(\tau_{j,i}^u) \sim \beta_{i,j}^{v,u}(\tau_{i,j}^v)$, by respecting the respective lexicographic orderings. Denote the identified simplex by $\beta_{i,j}^{v,u}$};
\EndFor

\end{algorithmic}
\end{algorithm}	

\begin{definition}[Admissible simplices of $\altcomplex(G)$] \label{def:admitlg}
The simplices $\alpha_i^v$ for every $i \in [j]$ and $v\in V_i$, and the simplices  $\beta_{i,j}^{v,u}$ for every edge $\{u,v\} \in E$ with $v\in V_i$ and $u\in V_j$ are said to be the \emph{admissible simplices}  of $\altcomplex(G)$.
\end{definition}

\begin{proposition}
The complex $\altcomplex(G)$   formed from identifying vertices in $\complex'$ is a simplicial complex. 
\end{proposition}
\begin{proof} This follows immediately from \Cref{lem:pasting,lem:satisfymunkres}. 
\end{proof}
This completes the construction of complex $\altcomplex(G)$.  Please refer to \Cref{fig:sidefigbnt} for a schematic illustration.


\paragraph*{Choice of input boundary.} For the abstract simplex $V$, let $\partial V$ denote the set of facets of $V$. Note that although $V \not\in \altcomplex(G)$, every facet of $V$ is in $\altcomplex(G)$. In fact,  the complex  $\hat{\Delta}\sigma_i$ for every $i$,
is a simplical $r$-manifold with $\partial V$ as its boundary.
We choose $\partial V$ as our input boundary that we want to make nontrivial.

\paragraph*{Choice of parameter.} Let $\left(k + \binom{k}{2} = \binom{k+1}{2}\right)$ be the parameter  for \createcycle on the complex $\altcomplex(G)$.

\begin{proposition} \label{prop:forwardcliquetwo} If there exists a $k$-clique $H$ of $G$ such that every vertex of $H$ has a different color, then a set of $\binom{k+1}{2} $ $r$-simplices in $\altcomplex(G)$  meets every chain $\xi$ with $\partial \xi = \partial V$.
\end{proposition}
\begin{proof} 

As in \Cref{sub:wonetop}, we construct a  set $\SCC$ of $r$-simplices that mimics the graphical structure of $H$ as follows:
\[\SCC_{\alpha}=\left\{ \alpha_{i}^{v}\,\,|\,\,v\in V_{i}\medcap V_{H}\right\} \]
\[\SCC_{\beta}=\left\{ \beta_{i,j}^{v,u}\,\,|\,\,v\in V_{i},\,u\in V_{j},\,\left\{ i,j\right\} \in E_{H}\right\} \]
Set $\SCC =   \SCC_{\alpha} \bigcup \SCC_{\beta}$.


Now  we want to show that at least one element from the  solution set $\SCC$ has coefficient $1$ in every chain $\xi$ that satisfies  $\partial \xi = \partial V$. Thus, we aim to show that  removing $\SCC$ from $\altcomplex(G)$ makes $\partial V$ nontrivial.
Before we proceed, we introduce some notations and definitions. To begin with let $\ACC$ denote the set of all admissible simplices in $\altcomplex(G)$ (described in \Cref{def:admitlg}).

\begin{notation}
For an $r$-simplex $\omega$, let $[\omega]_\xi$ denote the simplices of $[\omega]$ in $\xi$.
\end{notation}

\begin{definition}[\textsf{Type-1} gadget belonging to chain $\xi$] \label{def:belongsone}
If there exists an $r$-simplex $\omega \in \hat{\Delta}\sigma_i \setminus \ACC$ such that $\partial ([\omega]_\xi)  = 1$, then we say that $\sigma_{i}$ \emph{belongs to} $\xi$.
\end{definition}

\begin{definition}[\textsf{Type-2} gadget belonging to chain $\xi$] \label{def:belongstwo}
If there exists an $r$-simplex $\omega \in\Delta\tau _{i,j}^v \setminus \ACC $ such that $\partial ([\omega]_\xi) = 1$,   then we say that $\tau_{i,j}^{v}$ \emph{belongs to} $\xi$.
\end{definition}

Before we can finish the proof of \Cref{prop:forwardcliquetwo}, we need a few auxillary lemmas. For the lemmas that follow, we let $\xi$ be a chain that satisfies $\partial \xi = \partial V$.


\begin{figure}[H]
 \hspace*{-3em}{
\begin{tikzpicture}[line cap=round,line join=round,>=triangle 45,x=1cm,y=1cm]
\clip(-11.750431028496505,-3.166930732934382) rectangle (3.884917222852047,0.5626807348293423);
\fill[line width=1pt,color=sexdts,fill=sexdts,fill opacity=0.1] (-4,0) -- (-5,-2) -- (-4,-3) -- cycle;
\fill[line width=1pt,color=sexdts,fill=sexdts,fill opacity=0.1] (-4,-3) -- (-3,-2) -- (-4,0) -- cycle;
\fill[line width=1pt,color=ffttww,fill=ffttww,fill opacity=0.5] (-4.453386611607825,-1.1512548326892225) -- (-4.541131839767461,-1.6484777922604945) -- (-4.131654108355826,-1.546108359407586) -- cycle;
\fill[line width=1pt,color=ffttww,fill=ffttww,fill opacity=0.5] (-3.7159473819406323,-1.1442703432762205) -- (-3.803692610100268,-1.6414933028474925) -- (-3.3942148786886333,-1.539123869994584) -- cycle;
\draw [line width=1pt,color=sexdts] (-4,0)-- (-5,-2);
\draw [line width=1pt,color=zzwwff] (-5,-2)-- (-4,-3);
\draw [line width=1pt,color=sexdts] (-4,-3)-- (-4,0);
\draw [line width=1pt,color=zzwwff] (-4,-3)-- (-3,-2);
\draw [line width=1pt,color=sexdts] (-3,-2)-- (-4,0);
\draw [line width=1pt,color=sexdts] (-4,0)-- (-4,-3);
\draw [line width=0.8pt,dash pattern=on 2pt off 2pt,color=zzwwff] (-5,-2)-- (-3,-2);
\draw [line width=1pt,color=DarkSlateGray] (-4.453386611607825,-1.1512548326892225)-- (-4.541131839767461,-1.6484777922604945);
\draw [line width=1pt,color=DarkSlateGray] (-4.541131839767461,-1.6484777922604945)-- (-4.131654108355826,-1.546108359407586);
\draw [line width=1pt,color=DarkSlateGray] (-4.131654108355826,-1.546108359407586)-- (-4.453386611607825,-1.1512548326892225);
\draw [line width=1pt,color=DarkSlateGray] (-3.7159473819406323,-1.1442703432762205)-- (-3.803692610100268,-1.6414933028474925);
\draw [line width=1pt,color=DarkSlateGray] (-3.803692610100268,-1.6414933028474925)-- (-3.3942148786886333,-1.539123869994584);
\draw [line width=1pt,color=DarkSlateGray] (-3.3942148786886333,-1.539123869994584)-- (-3.7159473819406323,-1.1442703432762205);
\end{tikzpicture}
}
\caption{The figure depicts the gadget $\hat{\Delta}\sigma_i$ in a simplistic manner, that is, without the full triangulation and without the inadmissible simplices. Also, in this figure, \emph{$\sigma_i $ belongs to $\xi$}. That is, there exists an $r$-simplex $\omega \in \hat{\Delta}\sigma_i \setminus \ACC$ such that $\partial ([\omega]_\xi) = 1$. 
The  simplices of $\hat{\Delta}\sigma_i$ that lie in  $\ACC$ are shown in red.
Then, according to \Cref{lem:describeboundaries}, the boundary of the part of the complex in green equals the boundary of  triangles in red (i.e., the black edges) + $\partial V$ (shown in purple).}
\label{fig:sigmabelongs}
\end{figure}
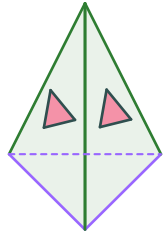

\begin{lemma}\label{lem:describeboundaries} If $\sigma_{i}$  belongs to
$\xi$, then $\partial\left(\left(\hat{\Delta}\sigma_i\setminus\ACC\right)\bigcap\xi\right) = \partial V  + \partial\left(\hat{\Delta}\sigma_i\bigcap\ACC\right) $.
\end{lemma}

\begin{proof} Please refer to \Cref{fig:sigmabelongs} for an  illustration of the statement of the lemma. Since $\sigma_{i}$  belongs to $\xi$, there exists an $r$-simplex $\omega_1 \in \hat{\Delta}\sigma_i \setminus \ACC$ such that $\partial ([\omega_1]_\xi) \neq 0$. This implies that there exists a facet $\varsigma$ of $\omega_1$ such that $\varsigma \in \partial ([\omega_1]_\xi)$ and $\varsigma \not\in \partial \xi$. Hence, there must be an  $r$-simplex $\omega_2$ with $\varsigma$ as a facet  such that 
$\omega_2 \in \hat{\Delta}\sigma_i \setminus \ACC$, $\partial ([\omega_2]_\xi) \neq 0$ and $\varsigma$ vanishes in $\partial ( [\omega_1]_\xi + [\omega_2]_\xi )$.
Repeating the argument above, we inductively add classes $[\omega_j]_\xi$, where $\omega_j \in \hat{\Delta}\sigma_i \setminus \ACC$ such that $\partial ([\omega_j]_\xi) \neq 0$. 
Note that by construction, $\bigcup_j [\omega_j]_\xi = \ \left(\hat{\Delta}\sigma_i\setminus\ACC\right)\bigcap\xi$, where $j$ indexes the  simplices in  $ \hat{\Delta}\sigma_{i} \setminus \ACC$. Clearly, the induction stops when 
$\partial(\bigcup_j [\omega_j]_\xi) = \partial V + \partial\left(\hat{\Delta}\sigma_i\bigcap\ACC\right)$.
\end{proof}

\begin{figure}[H]
 \hspace*{-3em}{
\begin{tikzpicture}[line cap=round,line join=rounDarkSlateGrayd,>=triangle 45,x=1cm,y=1cm]
\clip(-11.750431028496505,-3.166930732934382) rectangle (3.884917222852047,0.5626807348293423);
\fill[line width=1pt,color=rvwvcq,fill=rvwvcq,fill opacity=0.1] (-4,0) -- (-5,-2) -- (-4,-3) -- cycle;
\fill[line width=1pt,color=rvwvcq,fill=rvwvcq,fill opacity=0.1] (-4,-3) -- (-3,-2) -- (-4,0) -- cycle;
\fill[line width=1pt,color=ffttww,fill=ffttww,fill opacity=0.5] (-4.453386611607825,-1.1512548326892225) -- (-4.541131839767461,-1.6484777922604945) -- (-4.131654108355826,-1.546108359407586) -- cycle;
\fill[line width=1pt,color=ffttww,fill=ffttww,fill opacity=0.5] (-3.7159473819406323,-1.1442703432762205) -- (-3.803692610100268,-1.6414933028474925) -- (-3.3942148786886333,-1.539123869994584) -- cycle;
\draw [line width=1pt,color=rvwvcq] (-4,0)-- (-5,-2);
\draw [line width=1pt,color=rvwvcq] (-5,-2)-- (-4,-3);
\draw [line width=1pt,color=rvwvcq] (-4,-3)-- (-4,0);
\draw [line width=1pt,color=rvwvcq] (-4,-3)-- (-3,-2);
\draw [line width=1pt,color=rvwvcq] (-3,-2)-- (-4,0);
\draw [line width=1pt,color=rvwvcq] (-4,0)-- (-4,-3);
\draw [line width=0.8pt,dash pattern=on 2pt off 2pt,color=rvwvcq] (-5,-2)-- (-3,-2);
\draw [line width=1pt,color=DarkSlateGray] (-4.453386611607825,-1.1512548326892225)-- (-4.541131839767461,-1.6484777922604945);
\draw [line width=1pt,color=DarkSlateGray] (-4.541131839767461,-1.6484777922604945)-- (-4.131654108355826,-1.546108359407586);
\draw [line width=1pt,color=DarkSlateGray] (-4.131654108355826,-1.546108359407586)-- (-4.453386611607825,-1.1512548326892225);
\draw [line width=1pt,color=DarkSlateGray] (-3.7159473819406323,-1.1442703432762205)-- (-3.803692610100268,-1.6414933028474925);
\draw [line width=1pt,color=DarkSlateGray] (-3.803692610100268,-1.6414933028474925)-- (-3.3942148786886333,-1.539123869994584);
\draw [line width=1pt,color=DarkSlateGray] (-3.3942148786886333,-1.539123869994584)-- (-3.7159473819406323,-1.1442703432762205);
\end{tikzpicture}
}
\caption{The figure is a simplistic depiction of  gadget $\Delta\tau _{i,j}^v$. In particular, the full triangulation and the the inadmissible simplices of  $\Delta\tau _{i,j}^v$ are not shown. 
In this figure, \emph{$\tau _{i,j}^v $ belongs to $\xi$}. That is, there exists an $r$-simplex $\omega \in \Delta\tau _{i,j}^v \setminus \ACC$ such that $\partial ([\omega]_\xi) = 1$. 
The  simplices of $\Delta\tau _{i,j}^v$ that lie in $\ACC$  are shown in red.
Then, according to \Cref{lem:describeboundariestwo},  the boundary of the part of the complex in blue equals the boundary of  triangles in red (i.e., the black edges).}
\label{fig:taubelongs}
\end{figure}
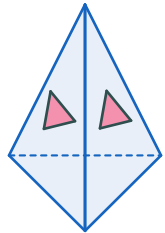

\begin{lemma} \label{lem:describeboundariestwo} If  $\tau_{i,j}^{v}$ belongs to $\xi$, then $  \partial\left(\left(\Delta\tau _{i,j}^v\setminus\ACC\right)\bigcap\xi\right) = \partial\left(\Delta\tau _{i,j}^v\bigcap\ACC\right) $.
\end{lemma}
\begin{proof} The argument is identical to the proof of~\Cref{lem:describeboundaries}. Please refer to \Cref{fig:taubelongs} for an  illustration of the statement of the lemma. 
\end{proof}

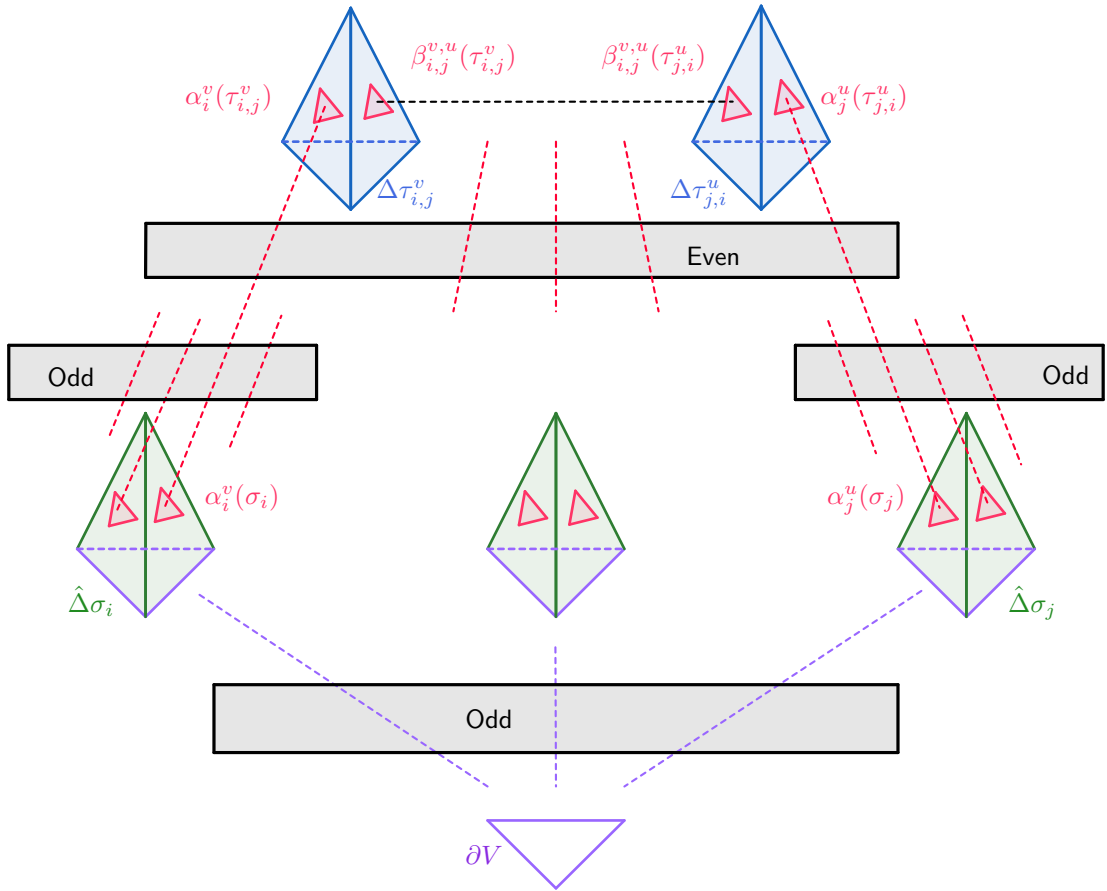
\begin{figure}
\begin{centering}

 \hspace*{-3em}{
\begin{tikzpicture}[scale=0.9, line cap=round,line join=round,>=triangle 45,x=1cm,y=1cm]
\clip(-12.190233905111764,-7.386234566696782) rectangle (4.351172667925901,6.359621767949472);
\fill[line width=1pt,color=rvwvcq,fill=rvwvcq,fill opacity=0.10000000149011612] (-7,5.97) -- (-8,4) -- (-7,3) -- cycle;
\fill[line width=1pt,color=rvwvcq,fill=rvwvcq,fill opacity=0.10000000149011612] (-7,5.97) -- (-6,4) -- (-7,3) -- cycle;
\fill[line width=1pt,color=sexdts,fill=sexdts,fill opacity=0.1] (-4,0) -- (-5,-2) -- (-4,-3) -- cycle;
\fill[line width=1pt,color=sexdts,fill=sexdts,fill opacity=0.1] (-4,-3) -- (-3,-2) -- (-4,0) -- cycle;
\fill[line width=1pt,color=sexdts,fill=sexdts,fill opacity=0.1] (-10,0) -- (-11,-2) -- (-10,-3) -- cycle;
\fill[line width=1pt,color=sexdts,fill=sexdts,fill opacity=0.1] (-10,0) -- (-9,-2) -- (-10,-3) -- cycle;
\fill[line width=1pt,color=sexdts,fill=sexdts,fill opacity=0.1] (1,-2) -- (2,-3) -- (2,0) -- cycle;
\fill[line width=1pt,color=sexdts,fill=sexdts,fill opacity=0.1] (2,0) -- (3,-2) -- (2,-3) -- cycle;
\fill[line width=1pt,color=rvwvcq,fill=rvwvcq,fill opacity=0.10000000149011612] (-1,6) -- (-2,4) -- (-1,3) -- cycle;
\fill[line width=1pt,color=rvwvcq,fill=rvwvcq,fill opacity=0.10000000149011612] (-1,6) -- (0,4) -- (-1,3) -- cycle;
\fill[line width=1pt,color=ffttww,fill=ffttww,fill opacity=0.1] (-6.708124509329276,4.8371672633900875) -- (-6.795869737488912,4.3399443038188155) -- (-6.386392006077277,4.442313736671724) -- cycle;
\fill[line width=1pt,color=ffttww,fill=ffttww,fill opacity=0.1] (-7.447655493016648,4.77916483329696) -- (-7.535400721176284,4.281941873725688) -- (-7.125922989764649,4.384311306578597) -- cycle;
\fill[line width=1pt,color=ffttww,fill=ffttww,fill opacity=0.1] (-1.4754628147350501,4.803838208010506) -- (-1.5632080428946862,4.306615248439234) -- (-1.153730311483051,4.408984681292143) -- cycle;
\fill[line width=1pt,color=ffttww,fill=ffttww,fill opacity=0.1] (-0.7088945410046864,4.903825374149251) -- (-0.7966397691643224,4.4066024145779785) -- (-0.3871620377526873,4.508971847430887) -- cycle;
\fill[line width=1pt,color=ffttww,fill=ffttww,fill opacity=0.1] (-9.774100968384111,-1.1031621260286695) -- (-9.861846196543748,-1.6003850855999413) -- (-9.452368465132112,-1.498015652747033) -- cycle;
\fill[line width=1pt,color=ffttww,fill=ffttww,fill opacity=0.1] (1.5478301230951717,-1.1650892647678186) -- (1.4600848949355356,-1.6623122243390906) -- (1.869562626347171,-1.559942791486182) -- cycle;
\fill[line width=1pt,color=ffttww,fill=ffttww,fill opacity=0.1] (-4.453386611607825,-1.1512548326892225) -- (-4.541131839767461,-1.6484777922604945) -- (-4.131654108355826,-1.546108359407586) -- cycle;
\fill[line width=1pt,color=ffttww,fill=ffttww,fill opacity=0.1] (-3.7159473819406323,-1.1442703432762205) -- (-3.803692610100268,-1.6414933028474925) -- (-3.3942148786886333,-1.539123869994584) -- cycle;
\fill[line width=1pt,color=ffttww,fill=ffttww,fill opacity=0.1] (-10.446217959748253,-1.1656816656468323) -- (-10.53396318790789,-1.6629046252181043) -- (-10.124485456496254,-1.5605351923651958) -- cycle;
\fill[line width=1pt,color=ffttww,fill=ffttww,fill opacity=0.1] (2.2413966294521828,-1.082761694764035) -- (2.1536514012925467,-1.5799846543353069) -- (2.563129132704182,-1.4776152214823983) -- cycle;
\fill[line width=0.3pt,fill=black,fill opacity=0.1] (-12,1) -- (-11.999298519088835,0.20047837662820206) -- (-7.5,0.2) -- (-7.5,1) -- cycle;
\fill[line width=0.3pt,fill=black,fill opacity=0.1] (-0.5,1) -- (-0.5,0.2) -- (4,0.2) -- (4,1) -- cycle;
\fill[line width=0.3pt,fill=black,fill opacity=0.1] (-9,-4) -- (-9,-5) -- (1,-5) -- (1,-4) -- cycle;
\fill[line width=0.3pt,fill=black,fill opacity=0.1] (-9.996032603185531,2.795953472306316) -- (-10,2) -- (1,2) -- (1,2.8) -- cycle;
\draw [line width=1pt,color=rvwvcq] (-7,5.97)-- (-8,4);
\draw [line width=1pt,color=rvwvcq] (-8,4)-- (-7,3);
\draw [line width=1pt,color=rvwvcq] (-7,3)-- (-7,5.97);
\draw [line width=1pt,color=rvwvcq] (-7,5.97)-- (-6,4);
\draw [line width=1pt,color=rvwvcq] (-6,4)-- (-7,3);
\draw [line width=1pt,color=rvwvcq] (-7,3)-- (-7,5.97);
\draw [line width=1pt,dash pattern=on 2pt off 2pt,color=RoyalBlue] (-8,4)-- (-6,4);
\draw [line width=1pt,color=sexdts] (-4,0)-- (-5,-2);
\draw [line width=1pt,color=zzwwff] (-5,-2)-- (-4,-3);
\draw [line width=1pt,color=sexdts] (-4,-3)-- (-4,0);
\draw [line width=1pt,color=zzwwff] (-4,-3)-- (-3,-2);
\draw [line width=1pt,color=sexdts] (-3,-2)-- (-4,0);
\draw [line width=1pt,color=sexdts] (-4,0)-- (-4,-3);
\draw [line width=1pt,dash pattern=on 2pt off 2pt,color=zzwwff] (-5,-2)-- (-3,-2);
\draw [line width=1pt,color=sexdts] (-10,0)-- (-11,-2);
\draw [line width=1pt,color=zzwwff] (-11,-2)-- (-10,-3);
\draw [line width=1pt,color=sexdts] (-10,-3)-- (-10,0);
\draw [line width=1pt,color=sexdts] (-10,0)-- (-9,-2);
\draw [line width=1pt,color=zzwwff] (-9,-2)-- (-10,-3);
\draw [line width=1pt,color=sexdts] (-10,-3)-- (-10,0);
\draw [line width=1pt,dash pattern=on 2pt off 2pt,color=zzwwff] (-11,-2)-- (-9,-2);
\draw [line width=1pt,color=zzwwff] (1,-2)-- (2,-3);
\draw [line width=1pt,color=sexdts] (2,-3)-- (2,0);
\draw [line width=1pt,color=sexdts] (2,0)-- (1,-2);
\draw [line width=1pt,color=sexdts] (2,0)-- (3,-2);
\draw [line width=1pt,color=zzwwff] (3,-2)-- (2,-3);
\draw [line width=1pt,color=sexdts] (2,-3)-- (2,0);
\draw [line width=1pt,dash pattern=on 2pt off 2pt,color=zzwwff] (1,-2)-- (3,-2);
\draw [line width=1pt,color=rvwvcq] (-1,6)-- (-2,4);
\draw [line width=1pt,color=rvwvcq] (-2,4)-- (-1,3);
\draw [line width=1pt,color=rvwvcq] (-1,3)-- (-1,6);
\draw [line width=1pt,color=rvwvcq] (-1,6)-- (0,4);
\draw [line width=1pt,color=rvwvcq] (0,4)-- (-1,3);
\draw [line width=1pt,color=rvwvcq] (-1,3)-- (-1,6);
\draw [line width=1pt,dash pattern=on 2pt off 2pt,color=RoyalBlue] (-2,4)-- (0,4);
\draw [line width=1pt,color=ffttww] (-6.708124509329276,4.8371672633900875)-- (-6.795869737488912,4.3399443038188155);
\draw [line width=1pt,color=ffttww] (-6.795869737488912,4.3399443038188155)-- (-6.386392006077277,4.442313736671724);
\draw [line width=1pt,color=ffttww] (-6.386392006077277,4.442313736671724)-- (-6.708124509329276,4.8371672633900875);
\draw [line width=1pt,color=ffttww] (-7.447655493016648,4.77916483329696)-- (-7.535400721176284,4.281941873725688);
\draw [line width=1pt,color=ffttww] (-7.535400721176284,4.281941873725688)-- (-7.125922989764649,4.384311306578597);
\draw [line width=1pt,color=ffttww] (-7.125922989764649,4.384311306578597)-- (-7.447655493016648,4.77916483329696);
\draw [line width=1pt,color=ffttww] (-1.4754628147350501,4.803838208010506)-- (-1.5632080428946862,4.306615248439234);
\draw [line width=1pt,color=ffttww] (-1.5632080428946862,4.306615248439234)-- (-1.153730311483051,4.408984681292143);
\draw [line width=1pt,color=ffttww] (-1.153730311483051,4.408984681292143)-- (-1.4754628147350501,4.803838208010506);
\draw [line width=1pt,color=ffttww] (-0.7088945410046864,4.903825374149251)-- (-0.7966397691643224,4.4066024145779785);
\draw [line width=1pt,color=ffttww] (-0.7966397691643224,4.4066024145779785)-- (-0.3871620377526873,4.508971847430887);
\draw [line width=1pt,color=ffttww] (-0.3871620377526873,4.508971847430887)-- (-0.7088945410046864,4.903825374149251);
\draw [line width=1pt,color=ffttww] (-9.774100968384111,-1.1031621260286695)-- (-9.861846196543748,-1.6003850855999413);
\draw [line width=1pt,color=ffttww] (-9.861846196543748,-1.6003850855999413)-- (-9.452368465132112,-1.498015652747033);
\draw [line width=1pt,color=ffttww] (-9.452368465132112,-1.498015652747033)-- (-9.774100968384111,-1.1031621260286695);
\draw [line width=1pt,color=ffttww] (1.5478301230951717,-1.1650892647678186)-- (1.4600848949355356,-1.6623122243390906);
\draw [line width=1pt,color=ffttww] (1.4600848949355356,-1.6623122243390906)-- (1.869562626347171,-1.559942791486182);
\draw [line width=1pt,color=ffttww] (1.869562626347171,-1.559942791486182)-- (1.5478301230951717,-1.1650892647678186);
\draw [line width=1pt,color=zzwwff] (-5,-6)-- (-3,-6);
\draw [line width=1pt,color=zzwwff] (-4,-7)-- (-5,-6);
\draw [line width=1pt,color=zzwwff] (-4,-7)-- (-3,-6);
\draw [line width=0.8pt,dash pattern=on 2pt off 2pt,color=ffqqtt] (-9.711959248618388,-1.3659497179250482)-- (-7.378441377708626,4.512968233096253);
\draw [line width=0.8pt,dash pattern=on 2pt off 2pt,color=Black] (-6.610998651150528,4.587269947586744)-- (-1.416014324417937,4.589830818850211);
\draw [line width=0.8pt,dash pattern=on 2pt off 2pt,color=ffqqtt] (-0.6320190206580366,4.633999568357811)-- (1.6095450168526677,-1.3950347394296219);
\draw [line width=1pt,color=ffttww] (-4.453386611607825,-1.1512548326892225)-- (-4.541131839767461,-1.6484777922604945);
\draw [line width=1pt,color=ffttww] (-4.541131839767461,-1.6484777922604945)-- (-4.131654108355826,-1.546108359407586);
\draw [line width=1pt,color=ffttww] (-4.131654108355826,-1.546108359407586)-- (-4.453386611607825,-1.1512548326892225);
\draw [line width=1pt,color=ffttww] (-3.7159473819406323,-1.1442703432762205)-- (-3.803692610100268,-1.6414933028474925);
\draw [line width=1pt,color=ffttww] (-3.803692610100268,-1.6414933028474925)-- (-3.3942148786886333,-1.539123869994584);
\draw [line width=1pt,color=ffttww] (-3.3942148786886333,-1.539123869994584)-- (-3.7159473819406323,-1.1442703432762205);
\draw [line width=1pt,color=ffttww] (-10.446217959748253,-1.1656816656468323)-- (-10.53396318790789,-1.6629046252181043);
\draw [line width=1pt,color=ffttww] (-10.53396318790789,-1.6629046252181043)-- (-10.124485456496254,-1.5605351923651958);
\draw [line width=1pt,color=ffttww] (-10.124485456496254,-1.5605351923651958)-- (-10.446217959748253,-1.1656816656468323);
\draw [line width=1pt,color=ffttww] (2.2413966294521828,-1.082761694764035)-- (2.1536514012925467,-1.5799846543353069);
\draw [line width=1pt,color=ffttww] (2.1536514012925467,-1.5799846543353069)-- (2.563129132704182,-1.4776152214823983);
\draw [line width=1pt,color=ffttww] (2.563129132704182,-1.4776152214823983)-- (2.2413966294521828,-1.082761694764035);
\draw [line width=0.8pt,dash pattern=on 2pt off 2pt,color=ffqqtt] (-10.412426270878699,-1.4204401828798467)-- (-9.192692232399477,1.4100486194876807);
\draw [line width=0.8pt,dash pattern=on 2pt off 2pt,color=ffqqtt] (-10.522880657754708,-0.3211586778680315)-- (-9.798473686528345,1.4772534769119874);
\draw [line width=0.8pt,dash pattern=on 2pt off 2pt,color=ffqqtt] (-8.765378312632537,-0.49178997351096093)-- (-8.014600611803647,1.4363436672541419);
\draw [line width=1.2pt] (-12,1)-- (-11.999298519088835,0.20047837662820206);
\draw [line width=1.2pt] (-11.999298519088835,0.20047837662820206)-- (-7.5,0.2);
\draw [line width=1.2pt] (-7.5,0.2)-- (-7.5,1);
\draw [line width=1.2pt] (-7.5,1)-- (-12,1);
\draw [line width=1.2pt] (-0.5,1)-- (-0.5,0.2);
\draw [line width=1.2pt] (-0.5,0.2)-- (4,0.2);
\draw [line width=1.2pt] (4,0.2)-- (4,1);
\draw [line width=1.2pt] (4,1)-- (-0.5,1);
\draw [line width=0.8pt,dash pattern=on 2pt off 2pt,color=zzwwff] (-9.211553878841203,-2.675677060436653)-- (-5,-5.5);
\draw [line width=0.8pt,dash pattern=on 2pt off 2pt,color=zzwwff] (-4.008204534611733,-3.449147908903197)-- (-4,-5.5);
\draw [line width=0.8pt,dash pattern=on 2pt off 2pt,color=zzwwff] (1.3422078328478801,-2.6186651966837657)-- (-3,-5.5);
\draw [line width=1.2pt] (-9,-4)-- (-9,-5);
\draw [line width=1.2pt] (-9,-5)-- (1,-5);
\draw [line width=1.2pt] (1,-5)-- (1,-4);
\draw [line width=1.2pt] (1,-4)-- (-9,-4);
\draw [line width=1.2pt] (-9.996032603185531,2.795953472306316)-- (-10,2);
\draw [line width=1.2pt] (-10,2)-- (1,2);
\draw [line width=1.2pt] (1,2)-- (1,2.8);
\draw [line width=1.2pt] (1,2.8)-- (-9.996032603185531,2.795953472306316);
\draw [line width=0.8pt,dash pattern=on 2pt off 2pt,color=ffqqtt] (2.3062209828086697,-1.3164511507729089)-- (1.2278455529293726,1.4468858882927969);
\draw [line width=0.8pt,dash pattern=on 2pt off 2pt,color=ffqqtt] (-0.01902603786856602,1.357021269136189)-- (0.6785224571463321,-0.5906919063302037);
\draw [line width=0.8pt,dash pattern=on 2pt off 2pt,color=ffqqtt] (1.946762506182236,1.4356528108982212)-- (2.8004763881700128,-0.7547972810441069);
\draw [line width=0.8pt,dash pattern=on 2pt off 2pt,color=ffqqtt] (-4,4)-- (-4,1.5);
\draw [line width=0.8pt,dash pattern=on 2pt off 2pt,color=ffqqtt] (-5,4)-- (-5.5,1.5);
\draw [line width=0.8pt,dash pattern=on 2pt off 2pt,color=ffqqtt] (-3,4)-- (-2.5,1.5);
\draw (-2.2291617990621067,2.5887990166884047) node[anchor=north west] {\textsf{Even}};
\draw (-11.572629010314907,0.8055527894910529) node[anchor=north west] {\textsf{Odd}};
\draw (2.971973030263508,0.8427037525576644) node[anchor=north west] {\textsf{Odd}};
\draw (-5.46129558585731,-4.228402706034805) node[anchor=north west] {\textsf{Odd}};

\draw (-11.272629010314907,-2.4055527894910529) node[anchor=north west] {\textcolor{ForestGreen}{$\hat{\Delta} \sigma_i$}};
\draw (2.471973030263508,-2.4055527894910529) node[anchor=north west] {\textcolor{ForestGreen}{$\hat{\Delta} \sigma_j$}};

\draw (-2.471973030263508,3.5887990166884047) node[anchor=north west] {\textcolor{RoyalBlue}{$\Delta\tau_{j,i}^u$}};
\draw (-6.76129558585731,3.5887990166884047) node[anchor=north west] {\textcolor{RoyalBlue}{$\Delta\tau_{i,j}^v$}};

\draw (-5.46129558585731,-6.228402706034805) node[anchor=north west] {\textcolor{BlueViolet}{$\partial V$}};

\draw (-9.272629010314907,-0.9055527894910529) node[anchor=north west] {\textcolor{ffttww}{$\alpha_i^v(\sigma_i)$}};
\draw (-0.171973030263508,-0.9055527894910529) node[anchor=north west] {\textcolor{ffttww}{$\alpha_j^u(\sigma_j)$}};

\draw (-3.471973030263508,5.5887990166884047) node[anchor=north west] {\textcolor{ffttww}{$\beta_{i,j}^{v,u}(\tau_{j,i}^u)$}};
\draw (-6.26129558585731,5.5887990166884047) node[anchor=north west] {\textcolor{ffttww}{$\beta_{i,j}^{v,u}(\tau_{i,j}^v)$}};

\draw (-0.271973030263508,4.9887990166884047) node[anchor=north west] {\textcolor{ffttww}{$\alpha_j^u(\tau_{j,i}^u)$}};
\draw (-9.56129558585731,4.9887990166884047) node[anchor=north west] {\textcolor{ffttww}{$\alpha_i^v(\tau_{i,j}^v)$}};

\end{tikzpicture}
}
\end{centering} 
\caption{
As in the case of \Cref{fig:mainfigbnt}, this figure also depicts some of the gadgets of $\altcomplex(G)$.
However, we depict \underline{only} those gadgets that \emph{belong} to some chain $\xi$. See \Cref{def:belongsone,def:belongstwo} for what it means for a gadget to belong to a chain.
 The dashed red lines show identifications of the (red) congruent faces  of  \typeone gadgets  that belong to $\xi$ to the (red) congruent faces  of   \typetwo gadgets that belong to $\xi$.  The dashed black line shows identifications along the (red) congruent faces  of two distinct  \typetwo gadgets  that belong to $\xi$.\\
 The odd count of purple dashed lines is the content of the \Cref{lem:countone}. The odd count of each group of  red dashed lines in the middle is the content of \Cref{lem:counttwo}.  Hence, for every purple dashed line, there is a group of red dashed lines of odd cardinality. 
On the one hand, since an odd sum of odd numbers is odd,  by \Cref{lem:countone,lem:counttwo}, the total number of red dashed lines should be odd.
 On the other hand, \Cref{lem:countthree} says that the cardinality of the red dashed lines (counted from above) is even. 
 The main idea of \Cref{prop:forwardcliquetwo} which uses \Cref{lem:countone,lem:counttwo,lem:countthree}, and a proof by contradiction is that an odd sum of odd numbers cannot be even.}
\label{fig:sidefigbnt}
\end{figure}

\begin{lemma}\label{lem:countone} The cardinality of the set $\left\{ i\in[k]\,\,\middle |\,\,\sigma_{i}\text{ belongs to }\xi\right\} $ is odd.
\end{lemma}
\begin{proof} Please see (the bottom portion of) \Cref{fig:sidefigbnt} for an  illustration of the statement of the lemma.  
First, note that by construction, $\partial V\bigcap\partial\left(\SCC\right)=\emptyset$.  Then, using~\Cref{lem:describeboundaries},  we have \mbox{$\partial V\subset\partial\left((\hat{\Delta}\sigma_i \setminus \ACC)\bigcap\xi\right)$} for every $\sigma_i$ that belongs to $\xi$.
Since $\partial V$ only occurs in the boundaries of \typeone gadgets and $\partial \xi =\partial V$, the cardinality of the set  $\left\{ i\in[k]\,\,\middle |\,\,\sigma_{i}\text{ belongs to }\xi\right\} $ must be odd. 
\end{proof}

\begin{lemma}\label{lem:counttwo} If    $\partial \xi \bigcap \SCC = \emptyset$, and if $\sigma_{i}$ belongs to $\xi$  for some $i\in [k]$, then 
the cardinality of 
\[I =  \left\{   \tau_{i,j}^{v} \, \middle | \,\,   v\in V_i\bigcap V_H, \,\, j\in[k]\setminus\{i\} \text{ and } \tau_{i,j}^{v} \text{ belongs to } \xi \right\} \]
 is odd. On the other hand, if $\sigma_{i}$ does not belong to $\xi$, then $I$ is even.
\end{lemma}
\begin{proof}Please see (the middle portion of) \Cref{fig:sidefigbnt} for an  illustration of the statement of the lemma.  

\begin{description}

\item[Case 1:] $\sigma_i$ belongs to $\xi$.

Since $H$ is a multicolored clique, for color $i$, there exists a vertex $v \in V_i \bigcap V_H$. Hence, by construction, $\alpha_{i}^{v} \in \SCC$. Moreover, $\alpha_{i}^{v}$ is the only $r$-simplex that is common to $\hat{\Delta}\sigma_i$ and $\Delta\tau _{i,j}^v$ for every $\tau _{i,j}^v \in I$. Note that $\SCC \subseteq \ACC$, and simplices in $\ACC$ have disjoint boundaries. 

\noindent{Since  $\sigma_i$ belongs to $\xi$, we obtain}
\begin{equation} \label{eq:sigmain}
\partial \alpha_{i}^{v}  \subset     \partial\left(\hat{\Delta}\sigma_i\bigcap\SCC\right)  \subset  \partial\left(\hat{\Delta}\sigma_i\bigcap \ACC \right) = \partial\left(\left(\hat{\Delta}\sigma_i\setminus\ACC\right)\bigcap\xi\right).
\end{equation}
 where the last equality uses \Cref{lem:describeboundaries}.

Moreover  for every $\tau _{i,j}^v \in I$, we obtain
\begin{equation} \label{eq:tauin}
\partial \alpha_{i}^{v}  \subset  \partial\left(\Delta\tau _{i,j}^v\bigcap \SCC \right)  \subset \partial\left(\Delta\tau _{i,j}^v\bigcap \ACC \right)  = \partial\left(\left(\Delta\tau _{i,j}^v\setminus\ACC\right)\bigcap\xi\right),
\end{equation}
where the last equality uses  \Cref{lem:describeboundariestwo}.

For  $j\in [k]\setminus\{i\} $ such that $\tau _{i,j}^v  \not \in I$, 
\begin{equation}\label{eq:tauout}
 \partial\left(\left(\Delta\tau _{i,j}^v\setminus\ACC\right)\bigcap\xi\right) = 0,
 \end{equation}
 which is a simple consequence of \Cref{def:belongstwo}.

Using the  assumption $\partial\xi\bigcap\SCC=\emptyset$, and \Cref{eq:sigmain,eq:tauin,eq:tauout} we get
\[ ((I+1)\text{ mod } 2) \cdot\partial \alpha_{i}^{v} \subset \partial\xi. \]

Since $\ensuremath{\partial\xi=\partial V}$, and $\alpha_{i}^{v} \bigcap \partial {V} = \emptyset$, $I+1$ should be even,  proving the first claim.

\item[Case 2:] $\sigma_i$ does not belong to $\xi$.

In this case,
\[\partial \alpha_{i}^{v}  \subset     \partial\left(\hat{\Delta}\sigma_i\bigcap\SCC\right)  \subset  \partial\left(\hat{\Delta}\sigma_i\bigcap \ACC \right) \not\subset \partial\left(\left(\hat{\Delta}\sigma_i\setminus\ACC\right)\bigcap\xi\right),\]
where the last non-inclusion follows from $ \partial\left(\left(\hat{\Delta}\sigma_i\setminus\ACC\right)\bigcap\xi\right) = 0$ (as a simple consequence of \Cref{def:belongsone}).
But for every $\tau _{i,j}^v \in I$, we still have
\[\partial \alpha_{i}^{v}  \subset  \partial\left(\Delta\tau _{i,j}^v\bigcap \SCC \right)  \subset \partial\left(\Delta\tau _{i,j}^v\bigcap \ACC \right)  = \partial\left(\left(\Delta\tau _{i,j}^v\setminus\ACC\right)\bigcap\xi\right).\]
which gives
\[ (I\text{ mod } 2) \cdot\partial \alpha_{i}^{v} \subset \partial\xi. \]
Since $\ensuremath{\partial\xi=\partial V}$, and $\alpha_{i}^{v} \bigcap \partial {V} = \emptyset$, $I$ should be even,  proving the second claim. \qedhere
\end{description}
\end{proof}

\begin{lemma}\label{lem:countthree} Assuming $\partial \xi \bigcap \SCC = \emptyset$, we define the set  $\PCC$ as \[\PCC =  \left\{ (i,j) \,\middle|\,\,\sigma_{i}\text{{ and }}\tau_{i,j}^{v}\text{ for some  } v\in V_i\bigcap V_H \text{ belong to }\xi\right\}. \] 
Then, $|\PCC|$ is even.
\end{lemma}
\begin{proof} Please see (the top portion of) \Cref{fig:sidefigbnt} for an  illustration of the statement of the lemma.  
\noindent{First, we define $\PCC'$ as follows.}
 \[\PCC' =  \left\{ (i,j) \,\middle|\,\, \tau_{i,j}^{v}\text{ for some  } v\in V_i\bigcap V_H \text{ belongs to }\xi\right\}. \] 
Suppose $(i,j)\in\PCC'$ for some $i\in[k]$, and $j\neq i$.  Since $H$ is a multicolored clique, for color $i$,  there exists a vertex $v \in V_i \bigcap V_H$.
Also, there exists a vertex $u \in V_j \bigcap V_H$ and an edge $\{u,v\} \in E_H$. By construction of $\SCC$, $\beta_{i,j}^{v,u} \in \SCC$.
Once again, we will use the facts: \begin{inparaenum}  \item $\SCC \subseteq \ACC$, and \item the simplices in $\ACC$ have disjoint boundaries. \end{inparaenum}

\noindent{For every $\tau _{i,j}^v \in I$, we obtain}
\begin{equation} \label{eq:tauintwo}
\partial \beta_{i,j}^{v,u}  \subset  \partial\left(\Delta\tau _{i,j}^v\bigcap \SCC \right)  \subset \partial\left(\Delta\tau _{i,j}^v\bigcap \ACC \right)  = \partial\left(\left(\Delta\tau _{i,j}^v\setminus\ACC\right)\bigcap\xi\right),
\end{equation}
where the last equality uses  \Cref{lem:describeboundariestwo}.

\noindent{Moreover, by construction, $\beta_{i,j}^{v,u}$ belongs to only two gadgets of $\altcomplex(G)$: $\Delta\tau _{j,i}^u$ and $\Delta\tau _{i,j}^v$.}

\noindent{Using $\partial\xi=\partial V$ and  $\partial V \bigcap \partial \beta_{i,j}^{v,u} = \emptyset$,  we deduce that $ \partial \beta_{i,j}^{v,u} \not\subset \partial\xi$.}
Then, using $\partial\xi\bigcap\SCC=\emptyset$, we have
\begin{equation}\label{eq:tauintwoplus}
\partial \beta_{i,j}^{v,u}  \subset \partial\left(\left(\Delta\tau _{j,i}^u\setminus\ACC\right)\bigcap\xi\right) 
\end{equation}
But this forces $\tau_{j,i}^{u}$ to belong to $\xi$, and hence the pair $(j,i)$  belongs to $\PCC'$.
Therefore, using \Cref{eq:tauintwo,eq:tauintwoplus}, $\PCC'$ is of even cardinality.
\noindent{Now,  define $\PCC''$ as follows.}
 \[\PCC'' =  \left\{ (i,j) \,\middle|\,\, \tau_{i,j}^{v}\text{ for some  } v\in V_i\bigcap V_H \text{ belongs to }\xi, \text{ and }\sigma_i \text{ does not belong to }\xi \right\}.\] 
By inductively applying Case 2 of \Cref{lem:counttwo}, we deduce that $\PCC''$ is of even cardinality.
Finally, $\PCC = \PCC' -\PCC''$. Hence, $\PCC$ is of even cardinality. 
\end{proof}

Now,  observe that if the conditions of~\Cref{lem:countone,lem:counttwo,lem:countthree} are simultaneously satisfied, then we reach a contradiction. This is because using~\Cref{lem:countone,lem:counttwo}, $|\PCC|$ is an odd set of odd numbers and hence odd, whereas according to \Cref{lem:countthree}, $|\PCC|$ is even. So if the chain $\xi$  has $\partial V$ as its boundary, then the assumption $\partial \xi \bigcap \SCC = \emptyset$ cannot be satisfied. This concludes the proof of \Cref{prop:forwardcliquetwo}.
\end{proof}

\begin{lemma} 
If there exists a chain $\xi'$ with $\partial \xi' = \partial V$ such that only the inadmissible simplices of $\RCC$ have coefficient $1$ in $ \xi'$, then the size of $\RCC$ is at least $m$.
\label{lem:biginadtwo}
\end{lemma}
\begin{proof} We skip the proof since it is identical to the proof of \Cref{lem:biginad}.
\end{proof}

\begin{lemma} \label{lem:admissiblewtoo}
Let $\RCC$ be a solution set for \createcycle on complex $\altcomplex(G)$. 
Then, 
\begin{enumerate}
\item For every $\hat{\Delta}\sigma_i$, there is at least one facet $\alpha_{i}^{v} $ with $v \in V_i$ that is included in $\RCC$. 
\item For every unordered pair $(i,j)$, where  $i , j\in [k] $, there exists a simplex $ \beta_{i,j}^{v,u}$ for some $v,u$ that is included in $\RCC$.
\item  If $|\RCC|  \leq \binom{k+1}{2} $, then $|\RCC|  =  |A_\RCC| =\binom{k+1}{2} $, where $A_\RCC$ denotes the set of admissible simplices of  $\RCC$.
\end{enumerate}
\end{lemma}
\begin{proof}

The proof is analogous to the proof of  \Cref{lem:admissiblew}. We repeat it here for the sake of clarity and completeness.

Let $A_\RCC$ denote the set of admissible simplices of  $\RCC$.
If $ \xi $ is such that $ \partial \xi = \partial V$ and $\xi \bigcap A_\RCC = \emptyset$, then  we are forced to include  inadmissible simplices. 
In that case, \Cref{lem:biginadtwo} applies, and  $\RCC$ is of cardinality at least $m=n^3$. But, if we include a total of (more than)  $n^3$ facets in $\RCC$, we exceed the budget of  $\binom{k+1}{2} $. So, going forward, we assume that at least one simplex in $A_\RCC$ has coefficient $1$ in every chain $\xi$, where $\xi = \partial V$.

Note that if at least one  simplex from  $A_\RCC$ has coefficient $1$ in every chain $\xi$ with $\xi = \partial V$, then we do not need  simplices that are inadmissible in $\RCC$. Next, we prove the three claims in the lemma.

\begin{enumerate}

\item Let $\xi = \hat{\Delta}\sigma_i$ for some $i \in [k]$. Then,  $\partial (\hat{\Delta}\sigma_i) = \partial V$. So if we do not include an admissible simplex $\alpha_{i}^{v} $ for some $v\in V_i$ in $\RCC$, then we would be forced to include some  inadmissible simplices of $\hat{\Delta}\sigma_i$.

\item Next, for some fixed $i$ and  $j \in [k] \setminus \{i\}$, let  $\xi = \hat{\Delta}\sigma_i  + \sum\limits_{v\in V_i}  (\Delta\tau _{i,j}^v)$.  Then, $ \partial (\hat{\Delta}\sigma_i)  + \sum\limits_{v\in V_i}  \partial(\Delta\tau _{i,j}^v) = \partial V$. So unless some admissible facet $ \beta_{i,j}^{v,u}$ for some $v,u$ is included in $\RCC$, the coefficient of all admissible simplices  in $\xi$ will be zero, and we would be forced to include  inadmissible simplices, which according to \Cref{lem:biginadtwo} is prohibitively expensive. 

\item The third claim follows immediately from the first two. \qedhere
\end{enumerate}
\end{proof}

\begin{lemma} \label{lem:mainreversetwo}
If $|\RCC|  =  |A_\RCC| =\binom{k+1}{2} $, then one can obtain a $k$-clique $H$ of $G$ from $\RCC$.
\end{lemma}
\begin{proof} Structurally the proof is identical to \Cref{lem:mainreverse}. The roles of  $\sigma_i$ and $\tau_{i,j}^{u}$  are played by  $\hat{\Delta}\sigma_i$ 
and  $\Delta\tau _{i,j}^v$, respectively. Moreover, there is a difference of $1$ in the cardinality of solution set $\RCC$, because for \hitcycles, we need to remove $V$ whereas the simplex $V$ is not a part of the complex $\altcomplex(G)$ in $\createcycle$.
\end{proof}

\Cref{prop:forwardcliquetwo} and \Cref{lem:mainreversetwo} together provide a parameterized reduction from \kmulticolorclique to \createcycle. 
Using~\Cref{thm:goodfellow}, we obtain the following result.

\begin{theorem} \createcycle is $\Wone$-hard.
\end{theorem}

\section{FPT algorithms}

\subsection{FPT algorithm for \hitcycles} \label{sub:fpthit}

In \Cref{sub:wonetop} we showed that \hitcycles is \Wone-hard with the  solution size $k$ as the parameter. This motivates the search of other meaningful parameters that make the problem tractable.
With that in mind, in this section, we  prove an important structural property about the connectivity of the minimal solution sets for \hitcycles.
First, we start with a definition. 

\begin{definition}[Induced subgraphs in Hasse graphs]
Given a $d$-dimensional complex $\complex$  with Hasse graph $\hasse$, and a set  $\SCC$ of $r$-simplices for some $r<d$,  the subgraph of $\hasse$ \emph{induced by} $\SCC$ is the  union of $\SCC$  with the set of $(r+1)$-dimensional simplices incident on $\SCC$.
\end{definition}

 \begin{lemma}\label{lem:connected} Given a $d$-dimensional complex $\complex$, a minimal solution of \hitcycles for a non-bounding cycle $\zeta \in \cycr(\complex)$ for some $r<d$ induces a connected subgraph of $\hasse$.
 \end{lemma}
 \begin{proof}
 Let $H_\SCC$ be the subgraph of the Hasse graph $\hasse$ induced by a minimum topological hitting set $\SCC$ of a non-bounding cycle $\zeta \in \cycr(\complex)$. Targeting a contradiction, assume there exist two components $C_1$ and $C_2$  such that $C_1$ and $C_2$  have no edges in common. Note that we do not assume that $C_1$ and $C_2$ are connected components, merely that they are components that do not share an edge.
 Since $\SCC$ is minimal, there exists a cycle  $\phi \in [\zeta]$  that is incident on an $r$-simplex in $C_1$ but not on any $r$-simplices in  $C_2$, and a cycle $\psi \in [\zeta]$  that is incident on an $r$-simplex in   $C_2$ but not on any $r$-simplices in  $C_1$.
 Then, $\phi =  \psi + \partial b$, for some $(r+1)$ chain  $b$. Let $b'$ be an $(r+1)$-chain obtained from $b$ by removing exactly those $(r+1)$-simplices that are incident on $C_1$. Now, let $\phi' =  \psi + \partial b'$.
 By construction, $\phi'$ is not incident on $C_1$. Also, because $C_1$ and $C_2$ are disconnected, the simplices {removed}  from $b$ to obtain $b'$ are not incident on $C_2$. Hence, $\phi'$ is not incident on $C_2$.
 In other words, $\phi' \in [\zeta] $ does not meet $\SCC$, and $\SCC$ is not a hitting set, a contradiction. Therefore, the induced subgraph of  $\SCC $ is connected.
 \end{proof}

Note that the path from any $r$-simplex to a neighboring $r$-simplex in the Hasse graph is of size $2$.
So it follows from \Cref{lem:connected}  that any minimal solution of size at most $k$ lies in some geodesic ball of radius $2k$ of some $r$-simplex in the Hasse graph. In particular, if we search across the geodesic ball of every $r$-simplex in the complex $\complex$, we will find a solution if one exists.  
So, if we choose $k+\Delta$, where $\Delta$ is the maximum degree of the Hasse graph, the search becomes tractable. In fact, we can even count the number of minimal solutions.
We remark that the degree $\Delta$ of the Hasse graph  $\hasse$ is bounded when the dimension of the complex is bounded and the number of incident cofacets on every simplex is bounded.

 \begin{algorithm}[H]
\caption{ FPT Algorithm for \hitcycles with  \mbox{$k$ +  $\Delta$} as the parameter }\label{alg:cycleskiller}
\begin{algorithmic}[1]

\State{$\min \gets |\complex|;  \,\,\,\, \sol = \complex$;}
\For{\textbf{each}  $r$-simplex $\tau$ of $\complex$}
\State{Consider the set $S_\tau$ of all simplices within the graph distance $2k$ (in $\hasse$) of $\tau$.}
\If{ a connected subset $S\subseteq S_\tau$ with $|S|\leq k$  is a hitting set of $\zeta$ and  $|S| < \min$}
\State{$\min = |S|; \,\,\,\, \sol = S$;}
\EndIf
\EndFor

\If{$\min <k$}
\textbf{return } $\sol$;
\EndIf
\end{algorithmic}
\label{alg:fptths}
\end{algorithm}


\paragraph*{Correctness.} The correctness of the algorithm  immediately follows from \Cref{lem:connected}.

\paragraph*{Complexity.} 

Note that in Line 4 of \Cref{alg:fptths}, we need to enumerate only the connected subsets $S$ of cardinality less than or equal to $k$. 
We use  \Cref{lem:villanger,lem:fomin} by Fomin and Villanger~\cite{fominvillanger} that provide very good bounds for enumerating connected subgraphs of graphs. First, we introduce some notation.

\begin{notation}
The neighborhood of a vertex $v$ is  denoted by $\nbd(v)=\left\{ u\in V:{u,v}\in E\right\} $, whereas the neighborhood of a vertex set $S\subseteq V$ is set to be $\nbd(S)=\bigcup_{v\in S}N(v)\setminus S$.
\end{notation}

\begin{lemma}[{\cite[Lemma~3.1]{fominvillanger}}] \label{lem:villanger}
Let $G=(V,E)$ be a graph. For every $v\in V$, and $b,d\geq0$,
the number of connected vertex subsets $\mathcal{C}\subseteq V$ such
that 
\begin{enumerate}
\item $v\in B$, 
\item $|B|=b+1$, and 
\item $|\nbd(B)|=d$ 
\end{enumerate}
is at most $\binom{b+d}{b}$.
\end{lemma}

\begin{lemma}[{\cite[Lemma~3.2]{fominvillanger}}] \label{lem:fomin}
All connected vertex sets of size $b+1$ with f neighbors of an $n$-vertex
graph $G$ can be enumerated in time $O(n^{2}\cdot b\cdot(b+d)\cdot\binom{b+d}{b})$
by making use of polynomial space.
\end{lemma}

In \Cref{alg:fptths}, $b = O(k)$ and $d = O(k\Delta)$. Therefore,
\begin{equation} \label{eqn:goodbound}
\binom{b+d}{b} = \binom{ O(k \Delta )}{O(k)}  \leq { (k \Delta) }^{O(k)} = 2^{O(k\log(k\Delta))}.
\end{equation}
Hence, by \Cref{lem:fomin}, for a single  $r$-simplex, the number of connected sets enumerated in  Line 4 is  $O(n^{2}\cdot O(k) \cdot O(k\Delta)\cdot 2^{O(k\log(k\Delta))})$  = $O(n^{5}\cdot 2^{O(k\log(k\Delta))})$ time. 
Since we do this for every $r$-simplex $\tau$ in $\complex$, the total time in enumerating all candidate sets in Lines 4-6 is at most $O(n^{6}\cdot 2^{O(k\log(k \Delta))})$.
Using \Cref{thm:easycheck}, one can check if the set is a feasible solution in time $O(n^\omega)$, where $\omega$ is the exponent of matrix multiplication.
Hence, the algorithm runs in  $O(n^{6+\omega}\cdot 2^{O(k\log(k \Delta))})$ time, which is fixed parameter tractable in $k+\Delta$.

\begin{theorem}
 \hitcycles admits an FPT algorithm with respect to the parameter $k+\Delta$, where $\Delta$ is the maximum degree of the Hasse graph and $k$ is the solution size. The algorithm runs in   $O(n^{6+\omega}\cdot 2^{O(k\log(k\Delta))})$ time.
\end{theorem}

\subsubsection{Randomized FPT algorithm for \killcycles} \label{sec:randkill}

Ostensibly, \killcycles looks a lot harder than \hitcycles. However, this is not really the case.
Fortunately, we can exploit the vector space structure of homology to design a randomized algorithm for \killcycles that uses the deterministic FPT algorithm for \hitcycles as a subroutine.

 \begin{algorithm}[H]
\caption{ Randomized FPT Algorithm for \killcycles with  \mbox{$k$ +  $\Delta$} as the parameter }\label{alg:cycleskillerg}
\begin{algorithmic}[1]
\State{Find the $r$-th homology basis of $\complex$. Denote the basis by $\BCC$. Here, $|\BCC| = \beta_r(\complex)$.}
\State{Arrange the cycles in $\BCC$ in a matrix. Denote the matrix by $\boldB$.}
\State{Let $\boldx$ be a uniformly distributed random  binary vector of dimension $ \beta_r(\complex)$.}
\State{With $\boldB\cdot\boldx$ as the input cycle, and \mbox{$k$ +  $\Delta$} as the parameter, invoke \Cref{alg:fptths}.}
\end{algorithmic}
\label{alg:rfptgths}
\end{algorithm}

\begin{proposition}
The probability that a minimal topological hitting set of  the cycle $\boldB \cdot \boldx$ is the optimal solution to \killcycles is at least $\nicefrac{1}{2}$. 
\end{proposition}
\begin{proof} Note that the total number of nontrivial $r$-th homology classes of $\complex$ is $2^{ \beta_r(\complex)}$.
Let $\SCC$ be a optimal solution to $\killcycles$. 
Then, because of the vector space structure of homology groups, the total number of nontrivial homology classes of $\complex_{\SCC}$ is at most $2^{ \beta_r(\complex)-1}$.
In other words,  $\SCC$ is a topological hitting set of at least $2^{ \beta_r(\complex)-1}$ nontrivial classes. Let $\CCC$ be the set of $r$-th homology classes for which $\SCC$ is 
a topological hitting set. Then, the probability that a uniformly random homology class chosen by  $\boldB \cdot \boldx$ belongs to $\CCC$ is at least $\frac{2^{ \beta_r(\complex)-1}}{2^{ \beta_r(\complex)}} = \frac{1}{2}$.
\end{proof}

From the proposition above, the following corollary follows immediately.

\begin{corollary} \Cref{alg:rfptgths} is a randomized FPT algorithm for \killcycles with  \mbox{$k$ +  $\Delta$} as the parameter.
\end{corollary}

\subsection{FPT approximation algorithm for \createcycle} \label{sub:createalgo}

It turns out that we do not have a connectivity lemma analogous to~\Cref{lem:connected} for \createcycle. 
For instance, consider the triangulation of a sphere as the input complex $\complex$, and let the boundary that needs to be made nontrivial  be the equator of the sphere. 
Then, the two triangles at the north pole and the south pole constitute an optimal solution for \createcycle, as the removal of these triangles makes the boundary nontrivial.
Clearly, the solution set consisting of these two triangles is not  connected.   Please refer to \Cref{fig:sphere1} from~\Cref{sec:intro}.
So it is not clear if there is an FPT algorithm for \createcycle with  $k+\Delta$ as the parameter. 

This motivates the search of another parameter that makes the problem tractable. 
To this end, we first  make a few  elementary observations.

\begin{lemma} \label{lem:trivial}
If there are two $(r+1)$-chains $\xi$ and $\xi'$ with $b$ as a boundary, then their sum is an $(r+1)$-cycle. 
Also, if an $(r+1)$-chain $\xi$ has $b$ as a boundary, and $\zeta$ is an $r+1$-cycle, then $\xi + \zeta$ has $b$ as a boundary.
\end{lemma}
\begin{proof}
If $\partial \xi = b $ and $\partial \xi' = b$, then we have $\partial( \xi + \xi' )= 0$.

Next, if $\partial \xi = b $  and $\partial \zeta  = 0$, then we have $\partial (\xi + \zeta) = b$.
\end{proof}

In other words, the number of chains that have $b$ as a boundary is precisely $2^{\dim \mathsf{Z}_{r+1}(\complex)}$.
When the complex is $(r+1)$-dimensional, $\mathsf{Z}_{r+1}(\complex) = \mathsf{H}_{r+1}(\complex)$.
So, for $(r+1)$-dimensional complexes we provide an FPT approximation algorithm with $\beta_{r+1}(\complex)$ as a parameter.
Let $\boldB$ be the $(r+1)$-th boundary matrix of complex $\complex$. The algorithm can be described as follows.

\begin{algorithm}[H]
\caption{FPT approximation algorithm for \createcycle}\label{alg:bdrykiller}
\begin{algorithmic}[1]
\vspace{1.5mm}

\State{$\boldB' \gets \boldB$; \,\,\,\,\, $X = \{\}$;}

\While{$\boldB'\cdot \boldx = \zeta$ has a solution}

\State{$ X \gets X \cup \{\boldx\} $.}

\State{$Y \gets$  the set of all chains generated from odd linear combinations of elements in $X$.}

\State{Let $\boldY$ denote the matrix with the chains of $Y$ as its columns.}

\State{Let $\sets$ be the collection of row indices of $\boldY$  and $\elements$ be the collection of column indices.}

\State{In the natural way, interpret $\sets$ as a collection of sets,  $\elements$ as a collection of elements, and $\boldY$ as the incidence matrix between sets and elements.
\State Solve the \setcover problem approximately for the instance described above using the greedy method~\cite[Chapter 2.1]{vazirani}.}

\State{Let $\SCC \subseteq \sets$ be the approximate solution  for the setcover problem.}

\State{Let  $\boldB'$ be the matrix formed by deleting  from $\boldB$ the columns specified by the $r+1$-simplices in $\SCC$.}

\EndWhile

\State{\textbf{Return} $\SCC$.}

\end{algorithmic}
\end{algorithm}

\begin{lemma} \label{lem:one} The algorithm terminates in $O(2^\beta  \beta n \cdot\min(n,2^\beta))$ time, where $\beta = \beta_{r+1}(\complex) + 1$ and $n$ is the number of simplices in $\complex$.
\end{lemma}
\begin{proof} First, we note that the algorithm terminates. This is because in each iteration of the while loop we add a vector $\boldx$ that is linearly independent  to the vectors in set $X$. 
By \Cref{lem:trivial}, the number of iterations is bounded by $\beta = \beta_{r+1}(\complex) + 1$. 
With an appropriate choice of data structures, Lines 4 and 5 can be executed in $O(2^\beta n)$ time. Note that the resulting matrix $\boldY$ has  $O(2^\beta)$ columns and $O(n)$ rows.

The most expensive step in the while loop is Line 8.
A simple implementation of the greedy approximation algorithm for \setcover runs in $O(2^\beta n \cdot\min(n,2^\beta))$ time~\cite[Chapter 35.3]{cormen}.
\end{proof}

\begin{lemma}\label{lem:two} When the while loop terminates $\SCC$ covers every chain whose boundary is $\zeta$. The set $\SCC$ returned at Line 11 provides an $O(\log n)$-factor approximation to \createcycle. 
\end{lemma}
\begin{proof} This follows from the fact that after deleting some columns of $\boldB$ specified by set $\SCC$, if the loop terminates, then  the new matrix $\boldB'$ has no solution to the equation $\boldB'\cdot \boldx = \zeta$.
The algorithm provides an $O(\log n)$ factor approximation because we use the  approximation algorithm for \setcover as a subroutine in Line 8.
\end{proof}

\Cref{lem:one} and \Cref{lem:two} combine to give the following theorem.
\begin{theorem} \createcycle has an $O(\log n)$-factor FPT approximation algorithm that takes bounding $r$-cycles as input  on  $(r+1)$-dimensional complexes, and runs in  $O(2^\beta\beta n \cdot\min(n,2^\beta))$  time,  where $\beta = \beta_{r+1}(\complex) + 1$ and $n$ is the number of simplices in $\complex$.
\end{theorem}

\subsubsection{Randomized FPT approximation algorithm for \createcycleg} \label{sub:randfptapprox}

As in the case of \killcycles in \Cref{sec:randkill}, we now exploit the vector space structure of  the boundary group to design a randomized algorithm for \createcycleg that uses the deterministic FPT approximation algorithm for \createcycle as a subroutine.

 \begin{algorithm}[H]
\caption{ Randomized FPT approximation algorithm for \killcycles with  \mbox{$\beta$} as the parameter }\label{alg:bdrykillerg}
\begin{algorithmic}[1]

\State{Find a basis $\BCC$  for the column space of $\partial_{r+1}(\complex)$.}
\State{Arrange the bounding cycles from $\BCC$ in a matrix  $\boldB$.}
\State{Let $\boldx$ be a uniformly distributed random binary vector of dimension $|\boldB|$.}
\State{With $\boldB\cdot\boldx$ at the input boundary, and \mbox{$\beta$} as the parameter, invoke \Cref{alg:bdrykiller}.}
\end{algorithmic}
\label{alg:rfptbnt}
\end{algorithm}

\begin{proposition}
Let $\RCC$ be a minimal set of simplices  whose removal from $\complex$ makes $\boldB \cdot \boldx$ nontrivial.
The probability that $\RCC$ is the optimal solution to \createcycleg is at least $\nicefrac{1}{2}$.
\end{proposition}
\begin{proof} The total number of elements in the range of $\boldB$ is $2^{|\BCC|}$.
Let $\SCC$ be a optimal solution to $\createcycleg$. 
Suppose that $\boldc$ is a bounding cycle that is made nontrivial by removal of $\SCC$ from $\complex$. 
Then, $\boldc \not \in \im (\partial_{r+1}(\complex_{\SCC}))$.
Suppose that  $\boldc'$ is a bounding cycle that continues to be trivial following removal of $\SCC$ from $\complex$. 
That is, $\boldc'   =\partial_{r+1}(\complex_{\SCC})\cdot \boldx $ for some $\boldx$. 
Then, $\boldc + \boldc' \not \in \im (\partial_{r+1}(\complex_{\SCC}))$, for otherwise, there would exist a vector $\boldy$ such that 
$\boldc + \boldc'   =\partial_{r+1}(\complex_{\SCC})\cdot \boldy $, which gives $\boldc    =\partial_{r+1}(\complex_{\SCC})\cdot (\boldx + \boldy )$, a contradiction.

Now, assume that none of the bounding cycles in  the basis $\BCC$ are made nontrivial by the removal of $\SCC$ from $\complex$. 
But that implies that any linear combination of cycles in $\BCC$ also belongs to $\im (\partial_{r+1}(\complex_{\SCC}))$.
This contradicts the existence of $\boldc$. So there exists at least one bounding cycle $\boldb\in \BCC$ which is made nontrivial by the removal of $\SCC$.
Let $\BCC_\SCC$ be the subset of cycles in $\BCC$ that are made nontrivial by the removal of $\SCC$.
Then, we have two cases:

\begin{description}

\item[Case 1:] Suppose  $\boldb$ is the only cycle in $\BCC_\SCC$.

 Now, let $\boldz$ be any cycle that lies in the span of $\BCC\setminus \{\boldb\}$.
By the argument above, $\boldb + \boldz$ is also made nontrivial by the removal of $\SCC$.
So, the total number of bounding cycles that are made nontrivial by the removal of $\complex_{\SCC}$ is  $2^{ |\BCC|-1}$.

\item[Case 2:] Suppose $\BCC_\SCC \setminus \{\boldb\}$ is nonempty.

 Then, one obtains a new set of vectors $\BCC'$ from $\BCC$ as follows: For every $\bolda \neq \boldb \in \BCC$ such that $\bolda$ is  made nontrivial by the removal of $\SCC$ and $\bolda + \boldb$ is in  $\im (\partial_{r+1}(\complex_{\SCC}))$, replace $\bolda$ by $\bolda + \boldb$.  It is easy to check that $\BCC'$ is also a basis for the column space of $\partial_{r+1}(\complex)$.
Moreover, if  $\boldz'$ is any cycle that lies in the span of $\BCC'\setminus \{\boldb\}$, then $\boldb + \boldz'$ is also made nontrivial by the removal of $\SCC$.
So, the total number of bounding cycles that are made nontrivial by the removal of $\complex_{\SCC}$ is  at least $2^{ |\BCC|-1}$.

\end{description}

From the above analysis, we conclude that there are at least $2^{ |\BCC|-1}$ bounding $r$-cycles that are made nontrivial by the removal of  $\SCC$. Let $\CCC$ be the set of bounding $r$-cycles for which $\SCC$ is a \createcycle solution. Then, the probability that a uniformly random bounding cycle chosen by  $\boldB \cdot \boldx$ belongs to $\CCC$ is at least $\frac{2^{ |\BCC|-1}}{2^{ |\BCC|}} = \frac{1}{2}$.
\end{proof}

From the proposition above, we obtain the following corollary immediately.

\begin{corollary} 
 \createcycleg has an  $O(\log n)$-factor randomized FPT approximation algorithm for $r$-th homology  on  $(r+1)$-dimensional complexes, with $\beta$ as the parameter.
 The algorithm runs in $O(2^\beta\beta n \cdot\min(n,2^\beta))$ time.
\end{corollary}






\section{Conclusion and Discussion}

In this paper, we devise a polynomial time algorithm for \hitcycles on closed surfaces. We believe that our algorithm should also easily generalize to surfaces with boundary.

Moreover, we show how certain cut problems generalize naturally from  graphs to simplicial complexes, motivating a complexity theoretic study of these problems.
For future work, it remains to be shown that \killcycles and \createcycleg are also \Wone-hard.
We believe that the \Wone-hardness reductions for \hitcycles and \createcycle can be extended to establish hardness results  for the global variants.
Finally, a theoretical future direction of our work is to investigate how (the global variants of)  \hitcycles and \createcycle may be used to study high dimensional expansion in simplicial complexes~\cite{dotterrer,lubotzky}.

\bibliography{HDCUT}

\end{document}